\documentclass[a4paper,UKenglish,cleveref,autoref,thm-restate]{lipics-v2021}

\pdfoutput=1 
\hideLIPIcs  

\usepackage{xypic}

\usepackage{prftree}
\global\prfinterspace=1.5em

\usepackage{stmaryrd}
\usepackage{upgreek}
\usepackage{mathtools} 

\usepackage{crossreftools}
\usepackage{bbold}

\newcommand{\subfree}{SFMTT}

\makeatletter
\newcommand{\customlabel}[2]{%
  \protected@write \@auxout {}{\string \newlabel {#1}{{#2}{\thepage}{#2}{#1}{}} }%
  \Hy@raisedlink{\hypertarget{#1}{}}#2%
}
\makeatother

\newcommand{\newrulename}[1]{\customlabel{rule:#1}{\textsc{#1}}}

\newcommand{\inlinerulenamestyle}[1]{\textsc{\footnotesize #1}}
\newcommand{\inlinerulename}[1]{{\footnotesize \ref{rule:#1}}}
\newcommand{\case}{\textsc{case}}

\newcommand{\explanation}[1]{\text{(#1)}}

\newcommand{\unitMod}{\mathbb{1}}
\newcommand{\unitTwocell}{1}

\DeclareMathOperator{\Hom}{\mathrm{Hom}}

\newcommand{\catM}{\mathcal{M}}
\newcommand{\SSyn}{\mathrm{SSyn}}
\newcommand{\SCtx}[1]{\mathrm{SCtx}_{#1}}

\newcommand{\keysym}{\text{{\vbox{\hbox{\includegraphics[width=1em]{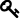}}}}}}
\newcommand{\locksym}{\text{{\vbox{\hbox{\includegraphics[width=.65em]{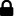}}}}}}

\newcommand{\emptyctx}{{\ensuremath{\cdot}}}
\newcommand{\emptysub}{\ !}
\newcommand{\emptytele}{\emptyctx}
\newcommand{\extendsctx}[2]{#1 \, . \, #2}
\newcommand{\extendctx}[3]{#1 \, . \, (#2 \shortmid #3)}
\newcommand{\extendlocktele}[2]{\locksym_{#1} \, . \, #2}
\newcommand{\judgment}[2]{#1 \, \textcolor{gray}{@ \, #2}}
\newcommand{\key}[4]{\keysym_{#4}^{\twocell{#1}{#2}{#3}}}
	
\newcommand{\lift}[1]{#1^+}
\newcommand{\lock}[2]{#1 \, . \, \locksym_{#2}}
\newcommand{\locks}[1]{\mathsf{locks}\left(#1\right)}
\newcommand{\Locktele}[2]{\mathsf{LockTele}(#2 \to #1)}
\newcommand{\Tele}[2]{\mathsf{sTele}(#2 \to #1)}
\newcommand{\sctx}[2]{\judgment{#1 \ \mathsf{sctx}}{#2}}
\newcommand{\sfexpr}[3]{\judgment{#1 \vdash_{\!\sfsymbol{}} #2 \ \mathsf{expr}}{#3}}
\newcommand{\sfarensub}[4]{\judgment{\vdash_{\!\sfsymbol{}} #1 \ \mathsf{aren/asub}(#2 \to #3)}{#4}}
\newcommand{\sfaren}[4]{\judgment{\vdash_{\!\sfsymbol{}} #1 \ \mathsf{aren}(#2 \to #3)}{#4}}
\newcommand{\sfasub}[4]{\judgment{\vdash_{\!\sfsymbol{}} #1 \ \mathsf{asub}(#2 \to #3)}{#4}}
\newcommand{\sfsub}[4]{\judgment{\vdash_{\!\sfsymbol{}} #1 \ \mathsf{sub}(#2 \to #3)}{#4}}
\newcommand{\sfrensub}[4]{\judgment{\vdash_{\!\sfsymbol{}} #1 \ \mathsf{ren/sub}(#2 \to #3)}{#4}}
\newcommand{\sfsymbol}{\mathsf{sf}}
\newcommand{\sfsigmeqsubsym}{\approx^{\textsf{\textup{obs}}}}
\newcommand{\sfsigmeqsub}[2]{#1 \sfsigmeqsubsym{} #2}
\newcommand{\sfvar}[3]{\judgment{#1 \vdash_{\!\sfsymbol{}} #2 \ \mathsf{var}}{#3}}
\newcommand{\twocell}[3]{#1 \in #2 \Rightarrow #3}
\newcommand{\weaken}[1]{\mathsf{weaken}(#1)}
\newcommand{\addtele}[2]{#1 \, . \, #2}

\newcommand{\mttsymbol}{\mathsf{ws}}
\newcommand{\mttexpr}[3]{\judgment{#1 \vdash_{\!\mttsymbol{}} #2 \ \mathsf{expr}}{#3}}
\newcommand{\mttsub}[4]{\judgment{\vdash_{\!\mttsymbol{}} #1 \ \mathsf{sub}(#2 \to #3)}{#4}}
\newcommand{\sigmeq}{\equiv^\upsigma}
\newcommand{\sigmeqexpr}[4]{\mttexpr{#1}{#2\sigmeq#3}{#4}}
\newcommand{\sigmeqsub}[5]{\mttsub{#1 \sigmeq #2}{#3}{#4}{#5}}

\newcommand{\transf}[4]{#4 \, \left[ \, #3 \, \right]_{\textsf{\textup{2-cell}}}^{#1 \Rightarrow #2}}

\newcommand{\substexpr}[3]{#2 \, \left[ \, #3 \, \right]_{\mathsf{#1}}}
\newcommand{\arensubexpr}[3][]{\substexpr{aren/asub}{#2}{#3}^{#1}}
\newcommand{\arenvar}[3]{\substexpr{aren,var}{#1}{#2}^{#3}}
\newcommand{\asubvar}[3]{\substexpr{asub,var}{#1}{#2}^{#3}}
\newcommand{\arenexpr}[3][]{\substexpr{aren}{#2}{#3}^{#1}}
\newcommand{\asubexpr}[3][]{\substexpr{asub}{#2}{#3}^{#1}}

\newcommand{\subexpr}[2]{\substexpr{sub}{#1}{#2}}
\newcommand{\subexprmtt}[2]{\substexpr{ws}{#1}{#2}}

\newcommand{\translate}[1]{\left\llbracket #1 \right \rrbracket}
\newcommand{\embed}[1]{\mathsf{embed}\!\left(#1\right)}

\newcommand{\vzero}[1]{\mathbf{v}_0^{#1}}
\newcommand{\suc}[1]{\mathsf{suc}\left(#1\right)}
\newcommand{\Bool}{\mathsf{Bool}}
\newcommand{\true}{\mathsf{true}}
\newcommand{\false}{\mathsf{false}}
\newcommand{\ifexpr}[4]{\mathsf{if}\left(#1; #2; #3; #4\right)}
\newcommand{\modfunc}[3]{\left(#1 \shortmid #2\right)\!\to #3}
\newcommand{\lam}[2]{\lambda^{#1}\left(#2\right)}
\newcommand{\app}[3]{\mathsf{app}_{#1}\left(#2; #3\right)}
\newcommand{\modty}[2]{\langle #1 \mid #2 \rangle}
\newcommand{\modtm}[2]{\mathsf{mod}_{#1}\left(#2\right)}
\newcommand{\letmod}[6]{\mathsf{letmod}_{#1, #2}\left(#3; #4; #5; #6\right)}

\newcommand{\id}{\mathsf{id}}
\newcommand{\ida}{\mathsf{id}^{\mathsf{a}}}
\newcommand{\rensubcons}[2]{#1 \, \raisebox{.2ex}{\scalebox{.75}{\textcircled{{\raisebox{.08ex}{\scalebox{.9}{\textsf{a}}}}}}} \, #2}
\newcommand{\concat}[2]{#1 +\!\!+ #2}

\newcommand{\sfmixseq}[4]{\judgment{\vdash_{\!\sfsymbol{}} #1 \ \mathsf{seq}(#2 \to #3)}{#4}}
\newcommand{\idm}{\mathsf{id}^{\mathsf{m}}}
\newcommand{\circm}{\ensuremath{\raisebox{.2ex}{\scalebox{.75}{\textcircled{{\raisebox{.08ex}{\scalebox{.9}{\textsf{m}}}}}}}}}
\newcommand{\mixsnocren}[2]{#1 \, \circm{}_{\mathsf{aren}} \, #2}
\newcommand{\mixsnocsub}[2]{#1 \, \circm{}_{\mathsf{asub}} \, #2}
\newcommand{\seqexpr}[2]{\substexpr{seq}{#1}{#2}}

\newcommand{\inferrule}[2]{%
  \begin{tabular}{l}%
    {{\scriptsize \newrulename{#1}}}\\[1pt]%
    #2%
  \end{tabular}%
}

\title{A Sound and Complete Substitution Algorithm for Multimode Type Theory: Technical Report}

\titlerunning{A Substitution Algorithm for Multimode Type Theory: Technical Report}

\author{Joris Ceulemans\footnote{Joris Ceulemans held a PhD fellowship (1184122N) of the Research Foundation -- Flanders (FWO) while working on this research. This research is partially funded by the Research Fund KU~Leuven and by the Research Foundation - Flanders (FWO; G030320N).}}{DistriNet, KU Leuven, Belgium \and \url{https://distrinet.cs.kuleuven.be/people/JorisCeulemans}}{joris.ceulemans@kuleuven.be}{https://orcid.org/0000-0001-9582-0789}{Held a PhD fellowship (1184122N) of the Research Foundation -- Flanders (FWO) while working on this research. This research is partially funded by the Research Fund KU~Leuven and by the Research Foundation - Flanders (FWO; G030320N).}

\author{Andreas Nuyts\footnote{Andreas nuyts holds a Postdoctoral fellowship (1247922N) of the Research Foundation -- Flanders (FWO).}}{DistriNet, KU Leuven, Belgium \and \url{https://anuyts.github.io/}}{andreas.nuyts@kuleuven.be}{https://orcid.org/0000-0002-1571-5063}{Holds a Postdoctoral fellowship (1247922N) of the Research Foundation -- Flanders (FWO).}

\author{Dominique Devriese}{DistriNet, KU Leuven, Belgium \and \url{https://distrinet.cs.kuleuven.be/people/DominiqueDevriese}}{dominique.devriese@kuleuven.be}{https://orcid.org/0000-0002-3862-6856}{}

\authorrunning{J. Ceulemans, A. Nuyts and D. Devriese}

\Copyright{Joris Ceulemans, Andreas Nuyts and Dominique Devriese}

\nolinenumbers 

\begin{document}

\maketitle

\section{Introduction}

This is the technical report accompanying the paper ``A Sound and Complete Substitution Algorithm for Multimode Type Theory''  \cite{mtt-sub-paper}.
It contains a full definition of WSMTT in \cref{sec:mtt}, including many rules for $\upsigma$-equivalence and a description of all rules that have been omitted.
Furthermore, we present completeness and soundness proofs of the substitution algorithm in full detail.
These can be found in \cref{sec:completeness,sec:soundness} respectively.
In order to make this document relatively self-contained, we also include a description of SFMTT in \cref{sec:sfmtt}.

\section{WSMTT: Full Description \& \texorpdfstring{$\upsigma$}{σ}-equivalence}
\label{sec:mtt}

\subsection{Extrinsically typed syntax}

\begin{figure}
  \centering
  
  \input{common-inference-rules/scoping-ctx}
  
  \caption{Definition of scoping contexts and lock telescopes. This figure is identical to Figure 3 in the paper.}
  \label{fig:scope-ctx}
\end{figure}

\begin{figure}
  \centering

  \input{common-inference-rules/mtt-intrinsically-scoped}
  
  \caption{Definition of WSMTT expressions (partial) and substitutions (full). This figure is identical to Figure 4 in the paper.}
  \label{fig:raw-mtt-expr-sub}
\end{figure}

\begin{figure}
  \centering

  \hspace*{\fill}
  \inferrule{wsmtt-expr-bool}{
    \prftree{
      \prfassumption{\sctx{\hat{\Gamma}}{m}}}{
      \mttexpr{\hat{\Gamma}}{\Bool}{m}}
  }
  \hfill
  \inferrule{wsmtt-expr-true}{
    \prftree{
      \prfassumption{\sctx{\hat{\Gamma}}{m}}}{
      \mttexpr{\hat{\Gamma}}{\true}{m}}
  }
  \hfill
  \inferrule{wsmtt-expr-false}{
    \prftree{
      \prfassumption{\sctx{\hat{\Gamma}}{m}}}{
      \mttexpr{\hat{\Gamma}}{\false}{m}}
  }
  \hspace*{\fill}
  \\[10pt]
  \hspace*{\fill}
  \inferrule{wsmtt-expr-if}{
    \prftree{
      \prfassumption{
        \begin{array}{l}
          \mttexpr{\extendsctx{\hat{\Gamma}}{\unitMod{}}}{A}{m} \\
          \mttexpr{\hat{\Gamma}}{s, t, t'}{m}
        \end{array}
      }}{
      \mttexpr{\hat{\Gamma}}{\ifexpr{A}{s}{t}{t'}}{m}}
    }
  \hfill
  \inferrule{wsmtt-expr-app}{
    \prftree{
      \prfassumption{\mu : m \to n}}{
      \prfassumption{
        \begin{array}{l}
          \mttexpr{\hat{\Gamma}}{f}{n} \\
          \mttexpr{\lock{\hat{\Gamma}}{\mu}}{t}{m}
        \end{array}
      }}{
      \mttexpr{\hat{\Gamma}}{\app{\mu}{f}{t}}{n}}
  }
  \hspace*{\fill}
  \\[10pt]
  \hspace*{\fill}
  \inferrule{wsmtt-expr-mod-ty}{
    \prftree{
      \prfassumption{\mu : m \to n}}{
      \prfassumption{\mttexpr{\lock{\hat{\Gamma}}{\mu}}{A}{m}}}{
      \mttexpr{\hat{\Gamma}}{\modty{\mu}{A}}{n}}
  }
  \hfill
  \inferrule{wsmtt-expr-mod-tm}{
    \prftree{
      \prfassumption{\mu : m \to n}}{
      \prfassumption{\mttexpr{\lock{\hat{\Gamma}}{\mu}}{t}{m}}}{
      \mttexpr{\hat{\Gamma}}{\modtm{\mu}{t}}{n}}
  }
  \hspace*{\fill}
  \\[10pt]
  \hspace*{\fill}
  \inferrule{wsmtt-expr-mod-elim}{
    \prftree{
      \prfassumption{
        \begin{array}{l}
          \mu : m \to n \\
          \nu : n \to o
        \end{array}
      }}{
      \prfassumption{
        \begin{array}{l}
          \mttexpr{\lock{\lock{\hat{\Gamma}}{\nu}}{\mu}}{A}{m} \\
          \mttexpr{\lock{\hat{\Gamma}}{\nu}}{t}{n}
        \end{array}
      }}{
      \prfassumption{
        \begin{array}{l}
          \mttexpr{\extendsctx{\hat{\Gamma}}{\nu}}{B}{o} \\
          \mttexpr{\extendsctx{\hat{\Gamma}}{\nu \circ \mu}}{s}{o}
        \end{array}
      }}{
      \mttexpr{\hat{\Gamma}}{\letmod{\nu}{\mu}{A}{B}{t}{s}}{o}}
  }
  \hspace*{\fill}

  \caption{Remaining constructors for WSMTT expressions, not covered in the paper}
  \label{fig:extra-wsmtt-constructors}
\end{figure}

The definition of scoping contexts and lock telescopes is repeated in \cref{fig:scope-ctx}.
All WSMTT expression and substitution constructors that were already covered by the paper are included in \cref{fig:raw-mtt-expr-sub}.
The other WSMTT constructors for expressions can be found in \cref{fig:extra-wsmtt-constructors}; the description of WSMTT substitutions was already complete in the paper.

The extra constructors for WSMTT expressions include a type of booleans (\inlinerulename{wsmtt-expr-bool}) with corresponding constructors (\inlinerulename{wsmtt-expr-true} and \inlinerulename{wsmtt-expr-false}) and dependent eliminator (\inlinerulename{wsmtt-expr-if}).
We see that when applying a (dependent) $\mu$-modal function to an expression $t$, that argument expression $t$ must be well-scoped in the locked context $\lock{\hat{\Gamma}}{\mu}$ (\inlinerulename{wsmtt-expr-app}).
Furthermore, there are the WSMTT versions of the formation (\inlinerulename{wsmtt-expr-mod-ty}) and introduction (\inlinerulename{wsmtt-expr-mod-tm}) for modal types rules from MTT.
The modal eliminator (\inlinerulename{wsmtt-expr-mod-elim}) corresponds to the MTT expression constructor $\mathsf{let}_\nu \, \modtm{\mu}{x} = t \ \mathsf{in} \ s$, which allows us to view a term $t$ of type $\modty{\mu}{A}$ as if it were of the form $\modtm{\mu}{x}$ when type checking the term $s$.
We refer to \cite{gratzer20-multimodal} for more details on this modal eliminator, as its behaviour with respect to substitution is not special and it does otherwise not play an important role in this report.

We emphasize again that all expression and substitution constructors in WSMTT can be obtained by removing the typing information from the corresponding constructors in MTT.

\subsection{\texorpdfstring{$\upsigma$}{σ}-equivalence}

\begin{figure}
  \centering

  \hspace*{\fill}
  \inferrule{wsmtt-eq-expr-refl}{%
    \prftree{
      \mttexpr{\hat{\Gamma}}{t}{m}}{
      \sigmeqexpr{\hat{\Gamma}}{t}{t}{m}}
  }
  \hfill
  \inferrule{wsmtt-eq-sub-id-right}{%
    \prftree{
      \prfassumption{\mttsub{\sigma}{\hat{\Gamma}}{\hat{\Delta}}{m}}}{
      \sigmeqsub{\sigma \circ \id}{\sigma}{\hat{\Gamma}}{\hat{\Delta}}{m}}
  }
  \hfill
  \inferrule{wsmtt-eq-expr-sub-id}{%
    \prftree{
      \prfassumption{\mttexpr{\hat{\Gamma}}{t}{m}}}{
      \sigmeqexpr{\hat{\Gamma}}{\subexprmtt{t}{\id}}{t}{m}}
  }
  \hspace*{\fill}
  \\[3pt]
  \hspace*{\fill}
  \inferrule{wsmtt-eq-expr-sub-compose}{%
    \prftree{
      \prfassumption{\mttexpr{\hat{\Xi}}{t}{m}}}{
      \prfassumption{\mttsub{\sigma}{\hat{\Delta}}{\hat{\Xi}}{m}}}{
      \prfassumption{\mttsub{\tau}{\hat{\Gamma}}{\hat{\Delta}}{m}}}{
      \sigmeqexpr{\hat{\Gamma}}{\subexprmtt{t}{\sigma \circ \tau}}{\subexprmtt{\subexprmtt{t}{\sigma}}{\tau}}{m}}
  }
  \hspace*{\fill}
  \\[3pt]
  \hspace*{\fill}
  \inferrule{wsmtt-eq-expr-cong-sub}{%
    \prftree{
      \prfassumption{\sigmeqexpr{\hat{\Delta}}{t_1}{t_2}{m}}}{
      \prfassumption{\sigmeqsub{\sigma_1}{\sigma_2}{\hat{\Gamma}}{\hat{\Delta}}{m}}}{
      \sigmeqexpr{\hat{\Gamma}}{\subexprmtt{t_1}{\sigma_1}}{\subexprmtt{t_2}{\sigma_2}}{m}}
  }
  \hspace*{\fill}
  \\[3pt]
  \hspace*{\fill}
  \inferrule{wsmtt-eq-expr-cong-lam}{%
    \prftree{
      \prfassumption{
        \begin{array}{l}
          \mu : m \to n \\
          \sigmeqexpr{\extendsctx{\hat{\Gamma}}{\mu}}{t_1}{t_2}{n}
        \end{array}
      }}{
      \sigmeqexpr{\hat{\Gamma}}{\lam{\mu}{t_1}}{\lam{\mu}{t_2}}{n}}
  }
  \hfill
  \inferrule{wsmtt-eq-expr-cong-app}{%
    \prftree{
      \prfassumption{\mu : m \to n}}{
      \prfassumption{
        \begin{array}{l}
          \sigmeqexpr{\hat{\Gamma}}{f_1}{f_2}{n} \\
          \sigmeqexpr{\lock{\hat{\Gamma}}{\mu}}{t_1}{t_2}{m}
        \end{array}
      }}{
      \sigmeqexpr{\hat{\Gamma}}{\app{\mu}{f_1}{t_1}}{\app{\mu}{f_2}{t_2}}{n}}
  }
  \hspace*{\fill}
  \\[3pt]
  \hspace*{\fill}
  \inferrule{wsmtt-eq-sub-cong-compose}{%
    \prftree{
      \prfassumption{
        \begin{array}{l}
          \sigmeqsub{\sigma_1}{\sigma_2}{\hat{\Delta}}{\hat{\Xi}}{m} \\
          \sigmeqsub{\tau_1}{\tau_2}{\hat{\Gamma}}{\hat{\Delta}}{m}
        \end{array}
      }}{
      \sigmeqsub{\sigma_1 \circ \tau_1}{\sigma_2 \circ \tau_2}{\hat{\Gamma}}{\hat{\Xi}}{m}}
  }
  \hfill
  \inferrule{wsmtt-eq-sub-cong-extend}{%
    \prftree{
      \prfassumption{\mu : m \to n}}{
      \prfassumption{
        \begin{array}{l}
          \sigmeqsub{\sigma_1}{\sigma_2}{\hat{\Gamma}}{\hat{\Delta}}{n} \\
          \sigmeqexpr{\lock{\hat{\Gamma}}{\mu}}{t_1}{t_2}{m}
        \end{array}
      }}{
      \sigmeqsub{\sigma_1 . t_1}{\sigma_2 . t_2}{\hat{\Gamma}}{\extendsctx{\hat{\Delta}}{\mu}}{n}}
  }
  \hspace*{\fill}
  \\[3pt]
  \hspace*{\fill}
  \inferrule{wsmtt-eq-sub-cong-lock}{%
    \prftree{
      \prfassumption{\mu : m \to n}}{
      \prfassumption{\sigmeqsub{\sigma_1}{\sigma_2}{\hat{\Gamma}}{\hat{\Delta}}{n}}}{
      \sigmeqsub{\lock{\sigma_1}{\mu}}{\lock{\sigma_2}{\mu}}{\lock{\hat{\Gamma}}{\mu}}{\lock{\hat{\Delta}}{\mu}}{m}}
  }
  \hspace*{\fill}
  \\[3pt]
  \hspace*{\fill}
  \inferrule{wsmtt-eq-expr-lam-sub}{%
    \prftree{
      \prfassumption{\mu : m \to n}}{
      \prfassumption{\mttexpr{\extendsctx{\hat{\Delta}}{\mu}}{t}{n}}}{
      \prfassumption{\mttsub{\sigma}{\hat{\Gamma}}{\hat{\Delta}}{n}}}{
      \sigmeqexpr{\hat{\Gamma}}{\subexprmtt{\left(\lam{\mu}{t}\right)}{\sigma}}{\lam{\mu}{\subexprmtt{t}{\lift{\sigma}}}}{n}}
  }
  \hspace*{\fill}
  \\[3pt]
  \hspace*{\fill}
  \inferrule{wsmtt-eq-expr-app-sub}{%
    \prftree{
      \prfassumption{\mu : m \to n}}{
      \prfassumption{\mttexpr{\hat{\Delta}}{f}{n}}}{
      \prfassumption{\mttexpr{\lock{\hat{\Delta}}{\mu}}{t}{m}}}{
      \prfassumption{\mttsub{\sigma}{\hat{\Gamma}}{\hat{\Delta}}{n}}}{
      \sigmeqexpr{\hat{\Gamma}}{\subexprmtt{\left(\app{\mu}{f}{t}\right)}{\sigma}}{\app{\mu}{\subexprmtt{f}{\sigma}}{\subexprmtt{t}{\lock{\sigma}{\mu}}}}{n}}
  }
  \hspace*{\fill}
  \\[3pt]
  \hspace*{\fill}
  \inferrule{wsmtt-eq-sub-empty-unique}{%
    \prftree{
      \prfassumption{\mttsub{\sigma}{\hat{\Gamma}}{\emptyctx{}}{m}}}{
      \sigmeqsub{\sigma}{\emptysub{}}{\hat{\Gamma}}{\emptyctx{}}{m}}
  }
  \hfill
  \inferrule{wsmtt-eq-expr-extend-var}{%
    \prftree{
      \prfassumption{\mu : m \to n}}{
      \prfassumption{\mttsub{\sigma}{\hat{\Gamma}}{\hat{\Delta}}{n}}}{
      \prfassumption{\mttexpr{\lock{\hat{\Gamma}}{\mu}}{t}{m}}}{
      \sigmeqexpr{\lock{\hat{\Gamma}}{\mu}}{\subexprmtt{\vzero{}}{\lock{(\sigma . t)}{\mu}}}{t}{m}}
  }
  \hspace*{\fill}
  \\[3pt]
  \hspace*{\fill}
  \inferrule{wsmtt-eq-sub-extend-weaken}{%
    \prftree{
      \prfassumption{
        \begin{array}{l}
          \mu : m \to n \\
          \mttsub{\sigma}{\hat{\Gamma}}{\hat{\Delta}}{n} \\
          \mttexpr{\lock{\hat{\Gamma}}{\mu}}{t}{m}
        \end{array}
      }}{
      \sigmeqsub{\pi \circ (\sigma . t)}{\sigma}{\hat{\Gamma}}{\hat{\Delta}}{n}}
  }
  \hfill
  \inferrule{wsmtt-eq-sub-extend-eta}{%
    \prftree{
      \prfassumption{
        \begin{array}{l}
          \phantom{\locksym_\mu} \\
          \mu : m \to n \\
          \mttsub{\sigma}{\hat{\Gamma}}{\extendsctx{\hat{\Delta}}{\mu}}{n}
        \end{array}
      }}{
      \sigmeqsub{\sigma}{(\pi \circ \sigma) . (\subexprmtt{\vzero{}}{\lock{\sigma}{\mu}})}{\hat{\Gamma}}{\extendsctx{\hat{\Delta}}{\mu}}{n}}
  }
  \hspace*{\fill}
  \\[3pt]
  \hspace*{\fill}
  \inferrule{wsmtt-eq-sub-lock-id}{%
    \prftree{
      \prfassumption{\mu : m \to n}}{
      \prfassumption{\sctx{\hat{\Gamma}}{n}}}{
      \sigmeqsub{\lock{\id}{\mu}}{\id}{\lock{\hat{\Gamma}}{\mu}}{\lock{\hat{\Gamma}}{\mu}}{m}}
  }
  \hspace*{\fill}
  \\[3pt]
  \hspace*{\fill}
  \inferrule{wsmtt-eq-sub-lock-compose}{%
    \prftree{
      \prfassumption{\mu : m \to n}}{
      \prfassumption{\mttsub{\sigma}{\hat{\Delta}}{\hat{\Xi}}{n}}}{
      \prfassumption{\mttsub{\tau}{\hat{\Gamma}}{\hat{\Delta}}{n}}}{
      \sigmeqsub{\lock{(\sigma \circ \tau)}{\mu}}{(\lock{\sigma}{\mu}) \circ (\lock{\tau}{\mu})}{\lock{\hat{\Gamma}}{\mu}}{\lock{\hat{\Xi}}{\mu}}{m}}
  }
  \hspace*{\fill}

  \caption{Definition of $\upsigma$-equivalence for WSMTT expressions and substitutions (see the overview for which rules are omitted, figure continues on the next page).}
  \label{fig:sigma-equiv}
\end{figure}

\begin{figure}
  \ContinuedFloat

  \hspace*{\fill}
  \inferrule{wsmtt-eq-sub-key-natural}{%
    \prftree{
      \prfassumption{\Lambda, \Theta : \Locktele{n}{m}}}{
      \prfassumption{\twocell{\alpha}{\locks{\Lambda}}{\locks{\Theta}}}}{
      \prfassumption{\mttsub{\sigma}{\hat{\Gamma}}{\hat{\Delta}}{m}}}{
      \sigmeqsub{\key{\alpha}{\Lambda}{\Theta}{\hat{\Delta}} \circ (\extendsctx{\sigma}{\Theta})}{(\extendsctx{\sigma}{\Lambda}) \circ \key{\alpha}{\Lambda}{\Theta}{\hat{\Gamma}}}{\extendsctx{\hat{\Gamma}}{\Theta}}{\extendsctx{\hat{\Delta}}{\Lambda}}{n}}
  }
  \hspace*{\fill}
  \\[5pt]
  \hspace*{\fill}
  \inferrule{wsmtt-eq-sub-key-unit}{%
    \prftree{
      \prfassumption{\sctx{\hat{\Gamma}}{m}}}{
      \prfassumption{\Lambda : \Locktele{n}{m}}}{
      \sigmeqsub{\key{1_{\locks{\Lambda}}}{\Lambda}{\Lambda}{\hat{\Gamma}}}{\id}{\extendsctx{\hat{\Gamma}}{\Lambda}}{\extendsctx{\hat{\Gamma}}{\Lambda}}{n}}
  }
  \hspace*{\fill}
  \\[5pt]
  \hspace*{\fill}
  \inferrule{wsmtt-eq-sub-key-compose-vertical}{%
    \prftree{
      \prfassumption{
        \begin{array}{l}
          \sctx{\hat{\Gamma}}{m} \\
          \Lambda, \Theta, \Psi : \Locktele{n}{m}
        \end{array}
      }}{
      \prfassumption{
        \begin{array}{l}
          \twocell{\alpha}{\locks{\Lambda}}{\locks{\Theta}} \\
          \twocell{\beta}{\locks{\Theta}}{\locks{\Psi}}
        \end{array}
      }}{
      \sigmeqsub{\key{\beta \circ \alpha}{\Lambda}{\Psi}{\hat{\Gamma}}}{\key{\alpha}{\Lambda}{\Theta}{\hat{\Gamma}} \circ \key{\beta}{\Theta}{\Psi}{\hat{\Gamma}}}{\extendsctx{\hat{\Gamma}}{\Psi}}{\extendsctx{\hat{\Gamma}}{\Lambda}}{n}}
  }
  \hspace*{\fill}
  \\[5pt]
  \hspace*{\fill}
  \inferrule{wsmtt-eq-sub-key-compose-horizontal}{%
    \prftree{
      \prfassumption{\sctx{\hat{\Gamma}}{m}}}{
      \prfassumption{
        \begin{array}{l}
          \Theta_1, \Theta_2 : \Locktele{o}{n} \\
          \Lambda_1, \Lambda_2 : \Locktele{n}{m}
        \end{array}
      }}{
      \prfassumption{
        \begin{array}{l}
          \twocell{\alpha}{\locks{\Theta_1}}{\locks{\Theta_2}} \\
          \twocell{\beta}{\locks{\Lambda_1}}{\locks{\Lambda_2}}
        \end{array}
      }}{
      \sigmeqsub{\key{\beta \star \alpha}{\extendsctx{\Lambda_1}{\Theta_1}}{\extendsctx{\Lambda_2}{\Theta_2}}{\hat{\Gamma}}}{(\extendsctx{\key{\beta}{\Lambda_1}{\Lambda_2}{\hat{\Gamma}}}{\Theta_1}) \circ \key{\alpha}{\Theta_1}{\Theta_2}{\extendsctx{\hat{\Gamma}}{\Lambda_2}}}{\extendsctx{\extendsctx{\hat{\Gamma}}{\Lambda_2}}{\Theta_2}}{\extendsctx{\extendsctx{\hat{\Gamma}}{\Lambda_1}}{\Theta_1}}{o}}
  }
  \hspace*{\fill}
  
  \caption{Definition of $\upsigma$-equivalence for WSMTT expressions and substitutions (continued).}
\end{figure}

To recall the notation, we make use of a judgement $\sigmeqexpr{\hat{\Gamma}}{t}{s}{m}$ for $\upsigma$-equivalence of WSMTT expressions and $\sigmeqsub{\sigma}{\tau}{\hat{\Gamma}}{\hat{\Delta}}{m}$ for $\upsigma$-equivalence of WSMTT substitutions.
Figure 6 in the paper only provides some of the rules for $\upsigma$-equivalence.
In this section we spell out the full definition, or at least give a description of what the full definition should look like.
Most of the rules for $\upsigma$-equivalence can be found in \cref{fig:sigma-equiv}.
All rules fall into different classes and for each class we describe the rules that have been omitted:

\begin{itemize}
\item There are rules expressing that $\upsigma$-equivalence of expressions and substitutions are equivalence relations (reflexivity, symmetry, transitivity). We show just the rule for reflexivity in \cref{fig:sigma-equiv} (\inlinerulename{wsmtt-eq-expr-refl}).
\item Given a mode $m$, we have a category $\SCtx{m}$ of scoping contexts at $m$.
  Its objects are given by scoping contexts and morphisms by substitutions.
  In order to have a category, we add rules that establish the associativity of composition and the fact that $\id$ is a unit of $\circ$.
  We show just 1 rule in \cref{fig:sigma-equiv}, namely \inlinerulename{wsmtt-eq-sub-id-right}.
\item There are rules that express the functoriality of explicit substitution in expressions, i.e.\ expressions involving the identity (\inlinerulename{wsmtt-eq-expr-sub-id}) and composite substitutions (\inlinerulename{wsmtt-eq-expr-sub-compose}).
\item For every expression and substitution constructor that takes some arguments, there are rules expressing that it preserves $\upsigma$-equivalence.
  We show the rules for $\subexprmtt{\_}{\_}$ (\inlinerulename{wsmtt-eq-expr-cong-sub}), $\lam{\mu}{\_}$ (\inlinerulename{wsmtt-eq-expr-cong-lam}), $\app{\mu}{\_}{\_}$ (\inlinerulename{wsmtt-eq-expr-cong-app}), $\_\circ\_$ (\inlinerulename{wsmtt-eq-sub-cong-compose}), $\_.\_$ (\inlinerulename{wsmtt-eq-sub-cong-extend}) and $\lock{\_}{\mu}$ (\inlinerulename{wsmtt-eq-sub-cong-lock}).
\item Furthermore, we have for every expression constructor a rule expressing how substitutions can be pushed through them.
  We explicitly show the rules for $\lam{\mu}{\_}$ (\inlinerulename{wsmtt-eq-expr-lam-sub}) and $\app{\mu}{\_}{\_}$ (\inlinerulename{wsmtt-eq-expr-app-sub}).
  Note that we make use of a lifting operation on WSMTT substitutions which is defined as follows.
  \begin{equation} \label{eq:wsmtt-sub-lift}
    \lift{\sigma} := (\sigma \circ \pi) . \vzero{}
  \end{equation}
\item The CwF rules governing the empty context (\inlinerulename{wsmtt-eq-sub-empty-unique}) and context extension (\inlinerulename{wsmtt-eq-expr-extend-var}, \inlinerulename{wsmtt-eq-sub-extend-weaken} and \inlinerulename{wsmtt-eq-sub-extend-eta}) are also present, but the ones for context extension are adapted to our modal situation, taking into account that variables are annotated with a modality in the context and that the extension constructor for substitutions takes a term that lives in a locked context.
\item
  We have two strict 2-categories in play: the mode theory $\mathcal{M}$ and $\mathbf{Cat}$, the 2-category of categories.
  We add rules to ensure that the intrinsically scoped WSMTT syntax provides us with a pseudofunctor $\SSyn$
  from $\mathcal{M}^{\mathsf{coop}}$ to $\mathbf{Cat}$ that maps every mode $m$ to the corresponding category $\SCtx m$ of scoping contexts and substitutions:
  \begin{itemize}
  \item A modality $\mu : m \to n$ must then be sent to a functor $\locksym_\mu : \SCtx n \to \SCtx m$, whose object part (action on scoping contexts) is defined in \cref{fig:scope-ctx} (\inlinerulename{sctx-lock}), and whose morphism part (action on substitutions) is defined in \cref{fig:raw-mtt-expr-sub} (\inlinerulename{wsmtt-sub-lock}).
    We add rules expressing the functor laws for this functor: \inlinerulename{wsmtt-eq-sub-lock-id} expresses that $\locksym_\mu$ preserves the identity substitution and \inlinerulename{wsmtt-eq-sub-lock-compose} expresses that it preserves composition of substitutions.
  \item A 2-cell $\twocell \alpha \mu \nu$ must be sent to a natural transformation $\keysym^\alpha : \locksym_\nu \to \locksym_\mu$ whose object part (action on scoping contexts) is defined in \cref{fig:raw-mtt-expr-sub} (\inlinerulename{wsmtt-sub-key}).
    We add a rule \inlinerulename{wsmtt-eq-sub-key-natural} expressing the naturality condition.
    However, we immediately express naturality not only for key substitutions between locks, but more generally for key substitutions between lock \emph{telescopes}.
  \item We add rules expressing that $\SSyn$'s action on Hom-categories is strictly functorial, i.e.\ that identity (\inlinerulename{wsmtt-eq-sub-key-unit}) and composition (\inlinerulename{wsmtt-eq-sub-key-compose-vertical}) of 2-cells are preserved.
  \item $\SSyn$ needs to respect identity up to isomorphism, i.e.\ $\locksym_\unitMod$ needs to be naturally isomorphic to the identity functor.
    An invertible substitution $\lock{\hat \Gamma}{\unitMod} \cong \hat \Gamma$ is given by $\key{\unitTwocell_{\unitMod}}{\, \emptytele \,}{\, \locksym_{\unitMod}}{\hat \Gamma}$,
    and naturality follows from the existing naturality rule.
  \item $\SSyn$ needs to respect composition up to isomorphism, i.e. the diagram
    \[
      \xymatrix{
        \Hom_\catM(n, o) \times \Hom_\catM(m, n)
        \ar[rr]^(.6){{-} \circ {-}}
        \ar[d]^{(\locksym_{\pi_2(-)}, \locksym_{\pi_1(-)})}
        &&
        \Hom_\catM(m, o)
        \ar[d]^{\locksym_{-}}
        \\
        [\SCtx n, \SCtx m] \times [\SCtx o, \SCtx n]
        \ar[rr]^(.6){{-} \circ {-}}
        &&
        [\SCtx o, \SCtx m]
      }
    \]
    must commute up to natural isomorphism.
    For any composable pair of modalities $\mu : m \to n$ and $\nu : n \to o$, an invertible substitution $\lock{\hat \Gamma}{\nu \circ \mu} \cong \lock{\lock{\hat \Gamma}{\nu}}{\mu}$ is given by $\key{\unitTwocell_{\nu \circ \mu}}{\, \lock{\locksym_\nu}{\mu}}{\, \locksym_{\nu \circ \mu}}{\hat \Gamma}$ and naturality with respect to $\hat \Gamma$ follows from the existing naturality rule.
    However, we also need naturality with respect to $\mu$ and $\nu$, so let $\twocell{\alpha}{\mu}{\mu'}$ and $\twocell{\beta}{\nu}{\nu'}$ and thus $\twocell{\beta \star \alpha}{\nu \circ \mu}{\nu' \circ \mu'}$.
    Then we add a rule relating the key substitution for $\beta \star \alpha$ to those for $\beta$ and $\alpha$ (\inlinerulename{wsmtt-eq-sub-key-compose-horizontal}).
  \item The category laws (left and right unit, and associativity) turn into coherence requirements for the isomorphisms established in the previous two points.
    However, these are all proven by reflexivity for the identity 2-cell.
  \end{itemize}
\end{itemize}

\section{\subfree{}: Full Description}
\label{sec:sfmtt}

\subsection{Intrinsically Scoped Syntax for \subfree}
\label{sec:raw-syntax}

\begin{figure}
  \centering

  \input{common-inference-rules/sfmtt-vars}
  
  \caption{Definition of well-scoped \subfree{} variables (identical to Figure 7 in the paper)}
  \label{fig:raw-subfree-var}
\end{figure}

\begin{figure}
  \centering
  \hspace*{\fill}
  \inferrule{sf-expr-var}{%
    \prftree{
      \sfvar{\hat{\Gamma}}{v}{m}}{
      \sfexpr{\hat{\Gamma}}{v}{m}}
  }
  \hfill
  \inferrule{sf-expr-bool}{%
    \prftree{
      \prfassumption{\sctx{\hat{\Gamma}}{m}}}{
      \sfexpr{\hat{\Gamma}}{\Bool}{m}}
  }
  \hspace*{\fill}
  \\[5pt]
  \hspace*{\fill}
  \inferrule{sf-expr-true}{%
    \prftree{
      \prfassumption{
        \begin{array}{l}
          \phantom{\hat{\Gamma}} \\
          \sctx{\hat{\Gamma}}{m}
        \end{array}
      }}{
      \sfexpr{\hat{\Gamma}}{\true}{m}}
  }
  \hfill
  \inferrule{sf-expr-false}{%
    \prftree{
      \prfassumption{
        \begin{array}{l}
          \phantom{\hat{\Gamma}} \\
          \sctx{\hat{\Gamma}}{m}
        \end{array}
      }}{
      \sfexpr{\hat{\Gamma}}{\false}{m}}
  }
  \hfill
  \inferrule{sf-expr-if}{%
    \prftree{
      \prfassumption{
        \begin{array}{l}
          \sfexpr{\extendsctx{\hat{\Gamma}}{\unitMod{}}}{A}{m} \\
          \sfexpr{\hat{\Gamma}}{s, t, t'}{m}
        \end{array}
      }}{
      \sfexpr{\hat{\Gamma}}{\ifexpr{A}{s}{t}{t'}}{m}}
    }
  \hspace*{\fill}
  \\[5pt]
  \hspace*{\fill}
  \inferrule{sf-expr-arrow}{%
    \prftree{
      \prfassumption{
        \begin{array}{l}
          \mu : m \to n \\
          \sfexpr{\lock{\hat{\Gamma}}{\mu}}{A}{m} \\
          \sfexpr{\extendsctx{\hat{\Gamma}}{\mu}}{B}{n}
        \end{array}
      }}{
      \sfexpr{\hat{\Gamma}}{\modfunc{\mu}{A}{B}}{n}}
  }
  \hfill
  \inferrule{sf-expr-lam}{%
    \prftree{
      \prfassumption{
        \begin{array}{l}
          \phantom{\locksym_\mu} \\
          \mu : m \to n \\
          \sfexpr{\extendsctx{\hat{\Gamma}}{\mu}}{t}{n}
        \end{array}
      }}{
      \sfexpr{\hat{\Gamma}}{\lam{\mu}{t}}{n}}
  }
  \hfill
  \inferrule{sf-expr-app}{%
    \prftree{
      \prfassumption{
        \begin{array}{l}
          \mu : m \to n \\
          \sfexpr{\hat{\Gamma}}{f}{n} \\
          \sfexpr{\lock{\hat{\Gamma}}{\mu}}{t}{m}
        \end{array}
      }}{
      \sfexpr{\hat{\Gamma}}{\app{\mu}{f}{t}}{n}}
  }
  \hspace*{\fill}
  \\[5pt]
  \hspace*{\fill}
  \inferrule{sf-expr-mod-ty}{%
    \prftree{
      \prfassumption{\mu : m \to n}}{
      \prfassumption{\sfexpr{\lock{\hat{\Gamma}}{\mu}}{A}{m}}}{
      \sfexpr{\hat{\Gamma}}{\modty{\mu}{A}}{n}}
  }
  \hfill
  \inferrule{sf-expr-mod-tm}{%
    \prftree{
      \prfassumption{\mu : m \to n}}{
      \prfassumption{\sfexpr{\lock{\hat{\Gamma}}{\mu}}{t}{m}}}{
      \sfexpr{\hat{\Gamma}}{\modtm{\mu}{t}}{n}}
  }
  \hspace*{\fill}
  \\[5pt]
  \hspace*{\fill}
  \inferrule{sf-expr-mod-elim}{%
    \prftree{
      \prfassumption{
        \begin{array}{l}
          \mu : m \to n \\
          \nu : n \to o
        \end{array}
      }}{
      \prfassumption{
        \begin{array}{l}
          \sfexpr{\lock{\lock{\hat{\Gamma}}{\nu}}{\mu}}{A}{m} \\
          \sfexpr{\lock{\hat{\Gamma}}{\nu}}{t}{n}
        \end{array}
      }}{
      \prfassumption{
        \begin{array}{l}
          \sfexpr{\extendsctx{\hat{\Gamma}}{\nu}}{B}{o} \\
          \sfexpr{\extendsctx{\hat{\Gamma}}{\nu \circ \mu}}{s}{o}
        \end{array}
      }}{
      \sfexpr{\hat{\Gamma}}{\letmod{\nu}{\mu}{A}{B}{t}{s}}{o}}
  }
  \hspace*{\fill}
  
  \caption{Definition of \subfree{} expressions using the judgement $\sfexpr{\hat{\Gamma}}{t}{m}$.}
  \label{fig:raw-subfree-expr}
\end{figure}

\begin{figure}
  \centering

  \input{common-inference-rules/sfmtt-arensub}
  
  \caption{Definition of atomic \subfree{} renamings and substitutions (identical to Figure 8 in the paper)}
  \label{fig:raw-subfree-arensub}
\end{figure}

\begin{figure}
  \centering

  \input{common-inference-rules/sfmtt-rensub}
  
  \caption{Definition of regular \subfree{} renamings and substitutions (identical to Figure 9 in the paper)}
  \label{fig:raw-subfree-rensub}
\end{figure}

There are not many details regarding \subfree{} that have not already been mentioned in the paper.
We just include some definitions here for this report to be more or less self-contained and to be able to refer to them later.

As mentioned in the paper, \subfree{} syntax is extrinsically typed but intrinsically scoped.
We therefore use a notion of scoping context, whose definition is included in \cref{fig:scope-ctx}.
Accessible SFMTT variables are defined in \cref{fig:raw-subfree-var} and the full definition of \subfree{} expressions can be found in \cref{fig:raw-subfree-expr}.
Note that all SFMTT constructors except \inlinerulename{sf-expr-var} have a counterpart in WSMTT.
Conversely, all WSMTT constructors except \inlinerulename{wsmtt-expr-var} and \inlinerulename{wsmtt-expr-sub} have a counterpart in SFMTT.
Atomic and regular \subfree{} rensubs are defined in \cref{fig:raw-subfree-arensub,fig:raw-subfree-rensub}.

We also recall some of the defined operations for atomic and regular SFMTT rensubs.
First of all, there is a weakening atomic rensub
\begin{equation}
  \label{eq:weakening-arensub}
  \pi := \weaken{\ida}
\end{equation}
from $\extendsctx{\hat{\Gamma}}{\mu}$ to $\hat{\Gamma}$ for any scoping context $\hat{\Gamma}$ and modality $\mu$.
Furthermore, given an atomic rensub $\sigma$ from $\hat{\Gamma}$ to $\hat{\Delta}$, we can construct a new, lifted atomic rensub
\begin{equation}
  \label{eq:sfmtt-arensub-lift}
  \lift{\sigma} := \weaken{\sigma} . \vzero{1_{\mu}}
\end{equation}
from $\extendsctx{\hat{\Gamma}}{\mu}$ to $\extendsctx{\hat{\Delta}}{\mu}$ (here $\vzero{1_\mu}$ is interpreted as a variable in the case of renamings and as an expression in the case of substitutions).
Finally, the lift and lock operations can be extended to regular rensubs by applying those operations to all constituent atomic rensubs.
In other words, we have
\begin{align*}
  \lift{\id} &:= \id & \lock{\id}{\mu} &:= \id \\
  \lift{(\rensubcons{\sigma}{\tau})} &:= \rensubcons{\lift{\sigma}}{\lift{\tau}} & \lock{(\rensubcons{\sigma}{\tau})}{\mu} &:= \rensubcons{(\lock{\sigma}{\mu})}{(\lock{\tau}{\mu})}.
\end{align*}

\subsection{Applying \subfree{} Substitutions}

\paragraph*{Atomic rensubs acting on non-variable expressions}
All cases for applying an atomic rensub to an \subfree{} expression that is not a variable are shown below.
These also include the cases that were omitted in Section 3.2.1 in the paper.

\begingroup\allowdisplaybreaks
\begin{align}
  \arensubexpr{\Bool}{\sigma} &= \Bool \\
  \arensubexpr{\true}{\sigma} &= \true \\
  \arensubexpr{\false}{\sigma} &= \false \\
  \arensubexpr{\ifexpr{A}{s}{t}{t'}}{\sigma} &= \nonumber \\*
                              && \mathllap{\ifexpr{\arensubexpr{A}{\lift{\sigma}}}{\arensubexpr{s}{\sigma}}{\arensubexpr{t}{\sigma}}{\arensubexpr{t'}{\sigma}}} \\
  \arensubexpr{\left(\modfunc{\mu}{A}{B}\right)}{\sigma} &= \modfunc{\mu}{\arensubexpr{A}{\lock{\sigma}{\mu}}}{\arensubexpr{B}{\lift{\sigma}}} \\
  \arensubexpr{\left(\lam{\mu}{t}\right)}{\sigma} &= \lam{\mu}{\arensubexpr{t}{\lift{\sigma}}} \label{eq:push-atomic-lam} \\
  \arensubexpr{\app{\mu}{f}{t}}{\sigma} &= \app{\mu}{\arensubexpr{f}{\sigma}}{\arensubexpr{t}{\lock{\sigma}{\mu}}} \label{eq:push-atomic-app} \\
  \arensubexpr{\modty{\mu}{A}}{\sigma} &= \modty{\mu}{\arensubexpr{A}{\lock{\sigma}{\mu}}} \\
  \arensubexpr{\modtm{\mu}{t}}{\sigma} &= \modtm{\mu}{\arensubexpr{t}{\lock{\sigma}{\mu}}} \label{eq:push-atomic-modtm} \\
  \arensubexpr{\letmod{\nu}{\mu}{A}{B}{t}{s}}{\sigma} &= \nonumber \\*
                              && \mathllap{\mathsf{letmod}_{\nu, \mu}\left(\arensubexpr{A}{\lock{\lock{\sigma}{\nu}}{\mu}}; \arensubexpr{B}{\lift{\sigma}}; \arensubexpr{t}{\lock{\sigma}{\nu}};\right.} \nonumber \\*
                              && \mathllap{\left. \arensubexpr{s}{\lift{\sigma}}\right)}
\end{align}
\endgroup

\paragraph*{Atomic rensubs acting on variables}

For easy reference in the proofs in the next sections, we recall the algorithm for applying an atomic rensub to a variable.
First of all, for applying a 2-cell to a variable, we have the following:
\begin{align}
  \transf{\Theta}{\Psi}{\alpha}{\vzero{\beta}} &= \vzero{(1_{\locks{\Lambda}} \star \alpha) \circ \beta} \label{eq:twocell-vzero} \\
  \transf{\Theta}{\Psi}{\alpha}{\suc{v}} &= \suc{\transf{\Theta}{\Psi}{\alpha}{v}}. \label{eq:twocell-vsuc}
\end{align}
The algorithm for applying a renaming to a variable is given by
\begin{align}
  \arenvar{v}{\ida}{\Lambda} &= v \label{eq:arenvar-id} \\
  \arenvar{v}{\weaken{\sigma}}{\Lambda} &= \suc{\arenvar{v}{\sigma}{\Lambda}} \label{eq:arenvar-weaken} \\
  \arenvar{v}{\lock{\sigma}{\mu}}{\Lambda} &= \arenvar{v}{\sigma}{\extendlocktele{\mu}{\Lambda}} \label{eq:arenvar-lock} \\
  \arenvar{v}{\key{\beta}{\Theta}{\Psi}{\hat{\Gamma}}}{\Lambda}
                             &= \transf{\extendsctx{\Theta}{\Lambda}}{\extendsctx{\Psi}{\Lambda}}{\beta \star 1_{\locks{\Lambda}}}{v} \label{eq:arenvar-key} \\
  \arenvar{\vzero{\alpha}}{\sigma . w}{\Lambda} &= \transf{\locksym_\mu}{\Lambda}{\alpha}{w} \label{eq:arenvar-extend-vzero} \\
  \arenvar{\suc{v}}{\sigma . w}{\Lambda} &= \arenvar{v}{\sigma}{\Lambda}. \label{eq:arenvar-extend-vsuc}
\end{align}
For atomic substitutions we have
\begingroup\allowdisplaybreaks%
\begin{align}
  \asubvar{v}{\ida}{\Lambda} &= v \label{eq:asubvar-id} \\
  \asubvar{v}{\weaken{\sigma}}{\Lambda} &= \arenexpr{\left(\asubvar{v}{\sigma}{\Lambda}\right)}{\extendsctx{\pi}{\Lambda}} \label{eq:asubvar-weaken} \\
  \asubvar{v}{\lock{\sigma}{\mu}}{\Lambda} &= \asubvar{v}{\sigma}{\extendlocktele{\mu}{\Lambda}} \label{eq:asubvar-lock} \\
  \asubvar{v}{\key{\beta}{\Theta}{\Psi}{\hat{\Gamma}}}{\Lambda}
                             &= \transf{\extendsctx{\Theta}{\Lambda}}{\extendsctx{\Psi}{\Lambda}}{\beta \star 1_{\locks{\Lambda}}}{v} \label{eq:asubvar-key} \\
  \asubvar{\vzero{\alpha}}{\sigma . t}{\Lambda} &= \arenexpr{t}{\key{\alpha}{\locksym_\mu}{\Lambda}{\hat{\Gamma}}} \label{eq:asubvar-extend-vzero} \\
  \asubvar{\suc{v}}{\sigma . t}{\Lambda} &= \asubvar{v}{\sigma}{\Lambda}. \label{eq:asubvar-extend-vsuc}
\end{align}%
\endgroup

\subsection{Relating WSMTT and SFMTT}
\label{sec:translation-embedding}

We present the full definitions of the translation function $\translate{\_}$ from WSMTT to SFMTT and the embedding function $\embed{\_}$ in the converse direction.
All interesting cases have been covered in the paper, but we include the definition here for easy reference.

\paragraph*{Translation from WSMTT to SFMTT}
\begin{align*}
  \translate{\modfunc{\mu}{A}{B}} &= \modfunc{\mu}{\translate{A}}{\translate{B}} & \translate{!} &= \rensubcons{\id}{!} \\
  \translate{\lam{\mu}{t}} &= \lam{\mu}{\translate{t}} & \translate{\id} &= \id \\
  \translate{\vzero{}} &= \vzero{1_\mu} & \translate{\pi} &= \rensubcons{\id}{\pi} \\
  \translate{\subexprmtt{t}{\sigma}} &= \subexpr{\translate{t}}{\translate{\sigma}} & \translate{\sigma \circ \tau} &= \concat{\translate{\sigma}}{\translate{\tau}} \\
  \translate{\Bool} &= \Bool & \translate{\lock{\sigma}{\mu}} &= \lock{\translate{\sigma}}{\mu} \\
  \translate{\true} &= \true & \translate{\key{\alpha}{\Theta}{\Psi}{\hat{\Gamma}}} &= \rensubcons{\id}{\key{\alpha}{\Theta}{\Psi}{\hat{\Gamma}}} \\
  \translate{\false} &= \false & \translate{\sigma . t} &= \rensubcons{\lift{\translate{\sigma}}}{(\ida . \translate{t})} \\
  \translate{\ifexpr{A}{s}{t}{t'}} &= \ifexpr{\translate{A}}{\translate{s}}{\translate{t}}{\translate{t'}} \\
  \translate{\app{\mu}{f}{t}} &= \app{\mu}{\translate{f}}{\translate{t}} \\
  \translate{\modty{\mu}{A}} &= \modty{\mu}{\translate{A}} \\
  \translate{\modtm{\mu}{t}} &= \modtm{\mu}{\translate{t}} \\
  \translate{\letmod{\nu}{\mu}{A}{B}{t}{s}} &= \letmod{\nu}{\mu}{\translate{A}}{\translate{B}}{\translate{t}}{\translate{s}}
\end{align*}

\paragraph*{Embedding of SFMTT into WSMTT}

For expressions we have the following.
\begingroup\allowdisplaybreaks%
\begin{align*}
  \embed{\vzero{\alpha}} &= \subexprmtt{\vzero{}}{\key{\alpha}{\locksym_\mu}{\Theta}{\hat{\Gamma}}} \\
  \embed{\suc{v}} &= \subexprmtt{\embed{v}}{\extendsctx{\pi}{\Theta}} \\
  \embed{\Bool} &= \Bool \\
  \embed{\true} &= \true \\
  \embed{\false} &= \false \\
  \embed{\ifexpr{A}{s}{t}{t'}} &= \ifexpr{\embed{A}}{\embed{s}}{\embed{t}}{\embed{t'}} \\
  \embed{\modfunc{\mu}{A}{B}} &= \modfunc{\mu}{\embed{A}}{\embed{B}} \\
  \embed{\lam{\mu}{t}} &= \lam{\mu}{\embed{t}} \\
  \embed{\app{\mu}{f}{t}} &= \app{\mu}{\embed{f}}{\embed{t}} \\
  \embed{\modty{\mu}{A}} &= \modty{\mu}{\embed{A}} \\
  \embed{\modtm{\mu}{t}} &= \modtm{\mu}{\embed{t}} \\
  \embed{\letmod{\nu}{\mu}{A}{B}{t}{s}} &= \letmod{\nu}{\mu}{\embed{A}}{\embed{B}}{\embed{t}}{\embed{s}}
\end{align*}
\endgroup
Embedding SFMTT rensubs (atomic and regular) to WSMTT substitutions is defined as follows.
\begin{align*}
  \embed{!} &= \emptysub{} & \embed{\key{\alpha}{\Lambda}{\Theta}{\hat{\Gamma}}} &= \key{\alpha}{\Lambda}{\Theta}{\hat{\Gamma}} \\
  \embed{\ida{}} &= \id & \embed{\sigma . t} &= \embed{\sigma} . \embed{t} \\
  \embed{\weaken{\sigma}} &= \embed{\sigma} \circ \pi & \embed{\id} &= \id \\
  \embed{\lock{\sigma}{\mu}} &= \lock{\embed{\sigma}}{\mu} & \embed{\rensubcons{\sigma}{\tau}} &= \embed{\sigma} \circ \embed{\tau}
\end{align*}

\section{Completeness}
\label{sec:completeness}

We want to prove the statement that our substitution algorithm is complete with respect to the notion of $\upsigma$-equivalence introduced in \cref{fig:sigma-equiv}.
In other words, whenever two WSMTT expressions are $\upsigma$-equivalent our algorithm should produce the same result.
\begin{theorem}
  \label{thm:completeness}
  If we can deduce $\sigmeqexpr{\hat{\Gamma}}{t}{s}{m}$, then we have that $\translate{t} = \translate{s}$.
\end{theorem}
Before we can prove this theorem, we need some technical machinery that will be developed in the next sections.

\subsection{Observational Equivalence of \subfree{} Substitutions}
\label{sec:observ-equiv}

\subsubsection{Definition \& Proof Technique (Part 1)}

Recall that $\upsigma$-equivalence for WSMTT expressions is defined mutually recursively with $\upsigma$-equivalence for WSMTT substitutions (see \cref{fig:sigma-equiv}).
Therefore, in order to prove \cref{thm:completeness}, we need to first extend it so as to also make a claim about $\upsigma$-equivalent WSMTT substitutions.
However, in \subfree{}, syntactic equality of substitutions is not a good notion of equivalence. Instead, we will use the following:
\begin{definition}[Observational equivalence]
  \label{def:sfmtt-sub-sigma-equiv}
  We say that two \subfree{} substitutions $\sfsub{\sigma, \tau}{\hat{\Gamma}}{\hat{\Delta}}{m}$ are observationally equivalent when $\subexpr{t}{\sigma} = \subexpr{t}{\tau}$ for every expression $\sfexpr{\hat{\Delta}}{t}{m}$.
  We will write this as $\sfsigmeqsub{\sigma}{\tau}$.
\end{definition}
Note that $\sfsigmeqsubsym$ is clearly an equivalence relation.
The requirement for two \subfree{} substitutions to be observationally equivalent is quite strong.
In order to prove this, we will make use of the technique outlined in \cref{prop:sfmtt-sigma-equiv-var,prop:sfmtt-sigma-equiv-var-locktele}.
Both propositions refer to general scoping telescopes which may contain both variables and locks, see \cref{fig:scoping-tele} for their definition.
We will refer to such scoping telescopes with the Greek letter $\Phi$.
They also act on \subfree{} substitutions in the following way.
\begin{align*}
  \addtele{\sigma}{\emptyctx{}} &= \sigma \\
  \addtele{\sigma}{(\extendsctx{\Phi}{\mu})} &= \lift{(\addtele{\sigma}{\Phi})} \\
  \addtele{\sigma}{(\lock{\Phi}{\mu})} &= \lock{(\addtele{\sigma}{\Phi})}{\mu}
\end{align*}
(Recall that the $\text{\locksym{}}_\mu$ and $^+$ operations on \subfree{} substitutions apply the corresponding operations to all atomic substitutions.)
In other words, whenever $\sfsub{\sigma}{\hat{\Gamma}}{\hat{\Delta}}{m}$ is an SFMTT substitution and $\Phi : \Tele{n}{m}$ a scoping telescope, we get an SFMTT substitution $\sfsub{\addtele{\sigma}{\Phi}}{\addtele{\hat{\Gamma}}{\Phi}}{\addtele{\hat{\Delta}}{\Phi}}{n}$.

\begin{figure}
  \centering

  \hspace*{\fill}
  \inferrule{stele-empty}{%
    \prftree{
      \prfassumption{\phantom{\Phi}}
    }{
      \emptyctx : \Tele{m}{m}}}
  \hfill
  \inferrule{stele-extend}{%
    \prftree{
      \prfassumption{\Phi : \Tele{m}{n}}}{
      \prfassumption{\mu : o \to m}}{
      \extendsctx{\Phi}{\mu} : \Tele{m}{n}}}
  \hfill
  \inferrule{stele-lock}{%
    \prftree{
      \prfassumption{\Phi : \Tele{m}{n}}}{
      \prfassumption{\mu : o \to m}}{
      \lock{\Phi}{\mu} : \Tele{o}{n}}}
  \hspace*{\fill}
  \\[5pt]
  \hspace*{\fill}
  $\addtele{\hat{\Gamma}}{\emptyctx{}} = \hat{\Gamma}$
  \hfill
  $\addtele{\hat{\Gamma}}{(\extendsctx{\Phi}{\mu})} = \extendsctx{(\addtele{\hat{\Gamma}}{\Phi})}{\mu}$
  \hfill
  $\addtele{\hat{\Gamma}}{(\lock{\Phi}{\mu})} = \lock{(\addtele{\hat{\Gamma}}{\Phi})}{\mu}$
  \hspace*{\fill}
  
  \caption{Definition of scoping telescopes and how to append them to a scoping context (note that a scoping telescope $\Phi : \Tele{m}{n}$ can be appended to a scoping context at mode $n$ to obtain a scoping context at mode $m$)}
  \label{fig:scoping-tele}
\end{figure}

\begin{proposition}
  \label{prop:sfmtt-sigma-equiv-var}
  Let $\sfsub{\sigma, \tau}{\hat{\Gamma}}{\hat{\Delta}}{m}$ be two \subfree{} substitutions and suppose that $\subexpr{v}{\addtele{\sigma}{\Phi}} = \subexpr{v}{\addtele{\tau}{\Phi}}$ for every scoping telescope $\Phi : \Tele{n}{m}$ and every variable $\sfvar{\extendsctx{\hat{\Delta}}{\Phi}}{v}{n}$.
  Then $\sfsigmeqsub{\sigma}{\tau}$.
\end{proposition}
\begin{proof}
  We will prove that $\subexpr{t}{\addtele{\sigma}{\Phi}} = \subexpr{t}{\addtele{\tau}{\Phi}}$ for all $\Phi : \Tele{n}{m}$ and all expressions $\sfexpr{\extendsctx{\hat{\Delta}}{\Phi}}{t}{n}$.
  The result then follows by taking $\Phi$ to be the empty scoping telescope.

  The proof proceeds by induction and case analysis on the expression $t$.
  We will describe only a few cases since there is a lot of similarity.
  \begin{itemize}
  \item \case{} $\sfexpr{\extendsctx{\hat{\Delta}}{\Phi}}{v}{n}$ for some $\sfvar{\extendsctx{\hat{\Delta}}{\Phi}}{v}{n}$ (\inlinerulename{sf-expr-var}) \\
    In this case the assumptions of the proposition we are proving tell us exactly that $\subexpr{v}{\addtele{\sigma}{\Phi}} = \subexpr{v}{\addtele{\tau}{\Phi}}$.
  \item \case{} $\sfexpr{\extendsctx{\hat{\Delta}}{\Phi}}{\lam{\mu}{t}}{n}$ for some $\sfexpr{\extendsctx{\extendsctx{\hat{\Delta}}{\Phi}}{\mu}}{t}{n}$ (\inlinerulename{sf-expr-lam}) \\
    Recall that an \subfree{} substitution is just a sequence of atomic \subfree{} substitutions which are applied sequentially to an expression.
    Following \cref{eq:push-atomic-lam} each of these atomic substitutions will be pushed through the $\lambda^\mu$ constructor, applying a lifting ($^+$) to that atomic substitution.
    Since the lifting of regular \subfree{} substitutions applies the lifting to all its constituent atomic substitutions, we have that
    \[
      \subexpr{\left(\lam{\mu}{t}\right)}{\addtele{\sigma}{\Phi}} = \lam{\mu}{\subexpr{t}{\lift{(\addtele{\sigma}{\Phi})}}} = \lam{\mu}{\subexpr{t}{\addtele{\sigma}{(\extendsctx{\Phi}{\mu})}}},
    \]
    and similar for $\tau$.
    We can now apply the induction hypothesis for the structurally smaller term $t$ to obtain that $\subexpr{t}{\addtele{\sigma}{(\extendsctx{\Phi}{\mu})}} = \subexpr{t}{\addtele{\tau}{(\extendsctx{\Phi}{\mu})}}$.
  \item \case{} $\sfexpr{\extendsctx{\hat{\Delta}}{\Phi}}{\modtm{\mu}{t}}{n}$ for some $\sfexpr{\lock{\extendsctx{\hat{\Delta}}{\Phi}}{\mu}}{t}{o}$ (\inlinerulename{sf-expr-mod-tm}) \\
    We can follow a similar style of reasoning as in the previous case, taking into account that applying a lock to a regular \subfree{} substitution applies that lock to all constituent atomic substitutions.
    Using \cref{eq:push-atomic-modtm} for every atomic substitution, we then get that
    \[
      \subexpr{\left(\modtm{\mu}{t}\right)}{\addtele{\sigma}{\Phi}} = \modtm{\mu}{\subexpr{t}{\lock{(\addtele{\sigma}{\Phi})}{\mu}}} = \modtm{\mu}{\subexpr{t}{\addtele{\sigma}{(\lock{\Phi}{\mu})}}},
    \]
    and similar for $\tau$.
    The induction hypothesis for $t$ gives us that $\subexpr{t}{\addtele{\sigma}{(\lock{\Phi}{\mu})}} = \subexpr{t}{\addtele{\tau}{(\lock{\Phi}{\mu})}}$.
    \qedhere
  \end{itemize}
\end{proof}

\subsubsection{Mixed Sequences of Atomic Rensubs}

Using \cref{prop:sfmtt-sigma-equiv-var} to prove observational equivalence is still far from trivial.
Therefore, \cref{prop:sfmtt-sigma-equiv-var-locktele} will relax the requirement so that we only have to check the equality of substituted variables after extending the context with an arbitrary \emph{lock} telescopes instead of a \emph{scoping} telescope.
However, in order to prove this proposition we will need some auxiliary results.

First of all, we will formulate a generalisation of \cref{prop:sfmtt-sigma-equiv-var} that applies to sequences consisting of both atomic renamings and atomic substitutions.
This generalisation is needed in the proof of \cref{prop:sfmtt-sigma-equiv-var-locktele}, but also in the completeness proof itself.
We define such mixed sequences in \cref{fig:mix-arensub-seq}.
That figure also contains definitions for the operations of lifting a sequence, applying a lock to a sequence, applying a sequence to an SFMTT expression, and applying a scoping telescope to a sequence.
These operations just apply the corresponding operations to the constituent atomic renamings and substitutions.
To distinguish a mixed sequence from atomic or regular rensubs, we will refer to such a sequence with an overlined Greek letter (so e.g.\ $\bar{\sigma}$).

\begin{figure}
  \centering

  \hspace*{\fill}
  \inferrule{sf-mix-id}{%
    \prftree{
      \prfassumption{
        \begin{array}{l}
          \phantom{\hat{\Gamma}} \\
          \sctx{\hat{\Gamma}}{m}
        \end{array}
      }}{
      \sfmixseq{\idm{}}{\hat{\Gamma}}{\hat{\Gamma}}{m}}
  }
  \hfill
  \inferrule{sf-mix-aren}{%
    \prftree{
      \prfassumption{
        \begin{array}{l}
          \sfmixseq{\bar{\sigma}}{\hat{\Delta}}{\hat{\Xi}}{m} \\
          \sfaren{\tau}{\hat{\Gamma}}{\hat{\Delta}}{m}
        \end{array}
      }}{
      \sfmixseq{\mixsnocren{\bar{\sigma}}{\tau}}{\hat{\Gamma}}{\hat{\Xi}}{m}}
  }
  \hfill
  \inferrule{sf-mix-asub}{%
    \prftree{
      \prfassumption{
        \begin{array}{l}
          \sfmixseq{\bar{\sigma}}{\hat{\Delta}}{\hat{\Xi}}{m} \\
          \sfasub{\tau}{\hat{\Gamma}}{\hat{\Delta}}{m}
        \end{array}
      }}{
      \sfmixseq{\mixsnocsub{\bar{\sigma}}{\tau}}{\hat{\Gamma}}{\hat{\Xi}}{m}}
  }
  \hspace*{\fill}
  \\[2pt]
  \begingroup\setlength{\mathindent}{0cm}
  \begin{align*}
    \lift{(\idm{})} &:= \idm
    & \lift{(\mixsnocren{\bar{\sigma}}{\tau})} &:= \mixsnocren{\lift{\bar{\sigma}}}{\lift{\tau}}
    & \lift{(\mixsnocsub{\bar{\sigma}}{\tau})} &:= \mixsnocsub{\lift{\bar{\sigma}}}{\lift{\tau}}
    \\[2pt]
    \lock{\idm{}}{\mu} &:= \idm
    & \lock{(\mixsnocren{\bar{\sigma}}{\tau})}{\mu} &:= \mixsnocren{\lock{\bar{\sigma}}{\mu}}{\lock{\tau}{\mu}}
    & \lock{(\mixsnocsub{\bar{\sigma}}{\tau})}{\mu} &:= \mixsnocsub{\lock{\bar{\sigma}}{\mu}}{\lock{\tau}{\mu}}
    \\[2pt]
    \seqexpr{t}{\idm} &:= t
    & \seqexpr{t}{\mixsnocren{\bar{\sigma}}{\tau}} &:= \arenexpr{\seqexpr{t}{\bar{\sigma}}}{\tau}
    & \seqexpr{t}{\mixsnocsub{\bar{\sigma}}{\tau}} &:= \asubexpr{\seqexpr{t}{\bar{\sigma}}}{\tau}
    \\[10pt]
    \addtele{\bar{\sigma}}{\emptytele} &:= \bar{\sigma}
    & \addtele{\bar{\sigma}}{(\extendsctx{\Phi}{\mu})} &:= \lift{(\addtele{\bar{\sigma}}{\Phi})}
    & \addtele{\bar{\sigma}}{(\lock{\Phi}{\mu})} &:= \lock{(\addtele{\bar{\sigma}}{\Phi})}{\mu}
  \end{align*}
  \endgroup
  
  \caption{Definition of mixed sequences of atomic rensubs and associated operations of lifting, locking and application to an SFMTT expression. We also show how to apply a scoping telescope to a mixed sequence.}
  \label{fig:mix-arensub-seq}
\end{figure}

\begin{proposition}
  \label{prop:mixseq-equiv-tele}
  Let $\sfmixseq{\bar{\sigma}, \bar{\tau}}{\hat{\Gamma}}{\hat{\Delta}}{m}$ be two mixed sequences of atomic renamings and substitutions and suppose that $\seqexpr{v}{\addtele{\bar{\sigma}}{\Phi}} = \seqexpr{v}{\addtele{\bar{\tau}}{\Phi}}$ for every scoping telescope $\Phi : \Tele{n}{m}$ and every variable $\sfvar{\extendsctx{\hat{\Delta}}{\Phi}}{v}{n}$.
  Then $\seqexpr{t}{\bar{\sigma}} = \seqexpr{t}{\bar{\tau}}$ for every SFMTT expression $\sfexpr{\hat{\Delta}}{t}{m}$.
\end{proposition}
\begin{proof}
  The reasoning is exactly the same as in the proof of \cref{prop:sfmtt-sigma-equiv-var}.
\end{proof}

\subsubsection{Action of Lifted Atomic Rensubs on Variables}

\begin{lemma}
  \label{lem:lift-aren-var}
  Given an atomic renaming $\sfaren{\sigma}{\hat{\Gamma}}{\hat{\Delta}}{m}$ and a lock telescope $\Lambda : \Locktele{n}{m}$, we have that
  $\arenexpr[\Lambda]{\vzero{\alpha}}{\lift{\sigma}} = \vzero{\alpha}$ and $\arenexpr[\Lambda]{\suc{v}}{\lift{\sigma}} = \suc{\arenexpr[\Lambda]{v}{\sigma}}$.
\end{lemma}
Note that we will no longer include $\mathsf{var}$ in the subscript of $\arenvar{v}{\sigma}{\Lambda}$ but just write $\arenexpr[\Lambda]{v}{\sigma}$.
\begin{proof}
  Recall that $\lift{\sigma}$ is defined as $\weaken{\sigma}.\vzero{1_\mu}$.
  We can then compute that
  \[
    \arenexpr[\Lambda]{\vzero{\alpha}}{\lift{\sigma}}
    = \arenexpr[\Lambda]{\vzero{\alpha}}{\weaken{\sigma}.\vzero{1_\mu}}
    = \transf{\locksym_\mu}{\Lambda}{\alpha}{\vzero{1_\mu}},
  \]
  where the last step makes use of \cref{eq:arenvar-extend-vzero}.
  By the definition of $\transf{\_}{\_}{\_}{\_}$ (see \cref{eq:twocell-vzero}), this last expression is equal to
  $\vzero{(1_\unitMod \star \alpha) \circ 1_\mu}$.
  From the laws of a strict 2-category, it follows that $(1_\unitMod \star \alpha) \circ 1_\mu = \alpha$ so the variable we obtain is really $\vzero{\alpha}$.

  In the case for $\suc{v}$, we can compute that
  \begin{align*}
    \arenexpr[\Lambda]{\suc{v}}{\lift{\sigma}}
    &= \arenexpr[\Lambda]{\suc{v}}{\weaken{\sigma}.\vzero{1_\mu}} \\
    &= \arenexpr[\Lambda]{v}{\weaken{\sigma}} && \text{(\cref{eq:arenvar-extend-vsuc})} \\
    &= \suc{\arenexpr[\Lambda]{v}{\sigma}}. && \text{(\cref{eq:arenvar-weaken})} && \qedhere
  \end{align*}
\end{proof}
Repeatedly applying \cref{lem:lift-aren-var} and realising that the lifting of a regular renaming consists of the liftings of its individual atomic renamings, one can see that the statement of \cref{lem:lift-aren-var} also holds for regular renamings.

For atomic substitutions we have the following result.
\begin{lemma}
  \label{lem:lift-asub-var}
  Given an atomic substitution $\sfasub{\sigma}{\hat{\Gamma}}{\hat{\Delta}}{m}$ and a lock telescope $\Lambda : \Locktele{n}{m}$, we have that $\asubexpr[\Lambda]{\vzero{\alpha}}{\lift{\sigma}} = \vzero{\alpha}$ and $\asubexpr[\Lambda]{\suc{v}}{\lift{\sigma}} = \arenexpr[\Lambda]{\asubexpr[\Lambda]{v}{\sigma}}{\pi}$ for every $\sfvar{\addtele{\hat{\Delta}}{\Lambda}}{v}{n}$.
\end{lemma}
\begin{proof}
  For $\vzero{\alpha}$ the computation proceeds as follows.
  \begin{align*}
    \asubexpr[\Lambda]{\vzero{\alpha}}{\lift{\sigma}}
    &= \asubexpr[\Lambda]{\vzero{\alpha}}{\weaken{\sigma}.\vzero{1_\mu}} \\
    &= \arenexpr{\vzero{1_\mu}}{\key{\alpha}{\locksym_\mu}{\Lambda}{\extendsctx{\hat{\Gamma}}{\mu}}} && \text{(\cref{eq:asubvar-extend-vzero})} \\
    &= \transf{\locksym_\mu}{\Lambda}{\alpha}{\vzero{1_\mu}} && \text{(\cref{eq:arenvar-key})} \\
    &= \vzero{(1_\unitMod \star \alpha) \circ 1_\mu} && \text{(\cref{eq:twocell-vzero})} \\
    &= \vzero{\alpha}
  \end{align*}

  For $\suc{v}$ we have
  \begin{align*}
    \asubexpr[\Lambda]{\suc{v}}{\lift{\sigma}}
    &= \asubexpr[\Lambda]{\suc{v}}{\weaken{\sigma}.\vzero{1_\mu}} \\
    &= \asubexpr[\Lambda]{v}{\weaken{\sigma}} && \text{(\cref{eq:asubvar-extend-vsuc})} \\
    &= \arenexpr[\Lambda]{\asubexpr[\Lambda]{v}{\sigma}}{\pi}. && \text{(\cref{eq:asubvar-weaken})} && \qedhere
  \end{align*}
\end{proof}

\subsubsection{Lifted Atomic Rensubs and \texorpdfstring{$\pi$}{π}}

\begin{lemma}
  \label{lem:pi-lift-commute-ren}
  Let $\Phi : \Tele{n}{m}$ be a scoping telescope, $\sfaren{\sigma}{\hat{\Gamma}}{\hat{\Delta}}{m}$ an atomic \subfree{} renaming and $\sfexpr{\extendsctx{\hat{\Delta}}{\Phi}}{t}{n}$ an expression.
  Then $\arenexpr{\arenexpr{t}{\addtele{\pi}{\Phi}}}{\addtele{\lift{\sigma}}{\Phi}} = \arenexpr{\arenexpr{t}{\addtele{\sigma}{\Phi}}}{\addtele{\pi}{\Phi}}$.
\end{lemma}
\begin{proof}
  We use \cref{prop:mixseq-equiv-tele} with the two sequences $\bar{\sigma}$ and $\bar{\tau}$ each consisting of the two atomic renamings on both sides of the lemma.
  In other words, we need to prove that $\arenexpr{\arenexpr{v}{\addtele{\pi}{\Phi}}}{\addtele{\lift{\sigma}}{\Phi}} = \arenexpr{\arenexpr{v}{\addtele{\sigma}{\Phi}}}{\addtele{\pi}{\Phi}}$ for every variable $\sfvar{\addtele{\hat{\Delta}}{\Phi}}{v}{n}$.
  We will do this by induction on the number of variables in $\Phi$.
  \begin{itemize}
  \item \case{} $\Phi = \Lambda$, so $\Phi$ contains only locks. \\
    Now we can compute that
    \begin{align*}
      \arenexpr{\arenexpr{v}{\addtele{\pi}{\Lambda}}}{\addtele{\lift{\sigma}}{\Lambda}}
      &= \arenexpr[\Lambda]{\arenexpr[\Lambda]{v}{\pi}}{\lift{\sigma}} \\
      &= \arenexpr[\Lambda]{\suc{v}}{\lift{\sigma}} \\
      &= \suc{\arenexpr[\Lambda]{v}{\sigma}} && \text{(\cref{lem:lift-aren-var})} \\
      &= \arenexpr[\Lambda]{\arenexpr[\Lambda]{v}{\sigma}}{\pi} \\
      &= \arenexpr{\arenexpr{v}{\addtele{\sigma}{\Lambda}}}{\addtele{\pi}{\Lambda}}
    \end{align*}
  \item \case{} $\Phi = \extendsctx{\extendsctx{\Phi'}{\rho}}{\Lambda}$ \\
    We now have to distinguish two cases for the variable $v$.
    \begin{itemize}
    \item \case{} $v = \vzero{\alpha}$ \\
      The computations go as follows.
      \begin{align*}
        &\arenexpr{\arenexpr{\vzero{\alpha}}{\addtele{\pi}{\extendsctx{\extendsctx{\Phi'}{\rho}}{\Lambda}}}}{\addtele{\lift{\sigma}}{\extendsctx{\extendsctx{\Phi'}{\rho}}{\Lambda}}} \\
        &= \arenexpr[\Lambda]{\arenexpr[\Lambda]{\vzero{\alpha}}{\lift{(\addtele{\pi}{\Phi'})}}}{\lift{(\addtele{\lift{\sigma}}{\Phi'})}} \\
        &= \arenexpr[\Lambda]{\vzero{\alpha}}{\lift{(\addtele{\lift{\sigma}}{\Phi'})}} && \text{(\cref{lem:lift-aren-var})} \\
        &= \vzero{\alpha} && \text{(\cref{lem:lift-aren-var})}
      \end{align*}
      \begin{align*}
        &\arenexpr{\arenexpr{\vzero{\alpha}}{\addtele{\sigma}{\extendsctx{\extendsctx{\Phi'}{\rho}}{\Lambda}}}}{\addtele{\pi}{\extendsctx{\extendsctx{\Phi'}{\rho}}{\Lambda}}} \\
        &= \arenexpr[\Lambda]{\arenexpr[\Lambda]{\vzero{\alpha}}{\lift{(\addtele{\sigma}{\Phi'})}}}{\lift{(\addtele{\pi}{\Phi'})}} \\
        &= \arenexpr[\Lambda]{\vzero{\alpha}}{\lift{(\addtele{\pi}{\Phi'})}} && \text{(\cref{lem:lift-aren-var})} \\
        &= \vzero{\alpha} && \text{(\cref{lem:lift-aren-var})}
      \end{align*}
    \item \case{} $v = \suc{v'}$ \\
      Now we can compute
      \begin{align*}
        &\arenexpr{\arenexpr{\suc{v'}}{\addtele{\pi}{\extendsctx{\extendsctx{\Phi'}{\rho}}{\Lambda}}}}{\addtele{\lift{\sigma}}{\extendsctx{\extendsctx{\Phi'}{\rho}}{\Lambda}}} \\
        &= \arenexpr[\Lambda]{\arenexpr[\Lambda]{\suc{v'}}{\lift{(\addtele{\pi}{\Phi'})}}}{\lift{(\addtele{\lift{\sigma}}{\Phi'})}} \\
        &= \arenexpr[\Lambda]{\suc{\arenexpr[\Lambda]{v'}{\addtele{\pi}{\Phi'}}}}{\lift{(\addtele{\lift{\sigma}}{\Phi'})}} && \text{(\cref{lem:lift-aren-var})} \\
        &= \suc{\arenexpr{\arenexpr{v'}{\addtele{\addtele{\pi}{\Phi'}}{\Lambda}}}{\addtele{\addtele{\lift{\sigma}}{\Phi'}}{\Lambda}}} && \text{(\cref{lem:lift-aren-var})}
      \end{align*}
      \begin{align*}
        &\arenexpr{\arenexpr{\suc{v'}}{\addtele{\sigma}{\extendsctx{\extendsctx{\Phi'}{\rho}}{\Lambda}}}}{\addtele{\pi}{\extendsctx{\extendsctx{\Phi'}{\rho}}{\Lambda}}} \\
        &= \arenexpr[\Lambda]{\arenexpr[\Lambda]{\suc{v'}}{\left\lift{(\addtele{\sigma}{\Phi'}\right)}}}{\lift{(\addtele{\pi}{\Phi'})}} \\
        &= \arenexpr[\Lambda]{\suc{\arenexpr[\Lambda]{v'}{\addtele{\sigma}{\Phi'}}}}{\lift{(\addtele{\pi}{\Phi'})}} && \text{(\cref{lem:lift-aren-var})} \\
        &= \suc{\arenexpr{\arenexpr{v'}{\addtele{\addtele{\sigma}{\Phi'}}{\Lambda}}}{\addtele{\addtele{\pi}{\Phi'}}{\Lambda}}}. && \text{(\cref{lem:lift-aren-var})}
      \end{align*}
      Hence the result directly follows from the induction hypothesis with scoping telescope $\extendsctx{\Phi'}{\Lambda}$ (which has one variable less than $\Phi$). \qedhere
    \end{itemize}
  \end{itemize}
\end{proof}

\begin{corollary}
  \label{cor:pi-lift-commute-ren-general}
  Let $\Phi_1 : \Tele{n}{m}$ and $\Phi_2 : \Tele{o}{n}$
  be two scoping telescopes, $\sfasub{\sigma}{\hat{\Gamma}}{\hat{\Delta}}{m}$ an atomic substitution and $\sfexpr{\addtele{\hat{\Delta}}{\extendsctx{\Phi_1}{\Phi_2}}}{t}{o}$ an SFMTT expression.
  Then we have that $\arenexpr{\arenexpr{t}{\addtele{\pi}{\Phi_2}}}{\addtele{\sigma}{\addtele{\addtele{\Phi_1}{\mu}}{\Phi_2}}} = \arenexpr{\arenexpr{t}{\addtele{\addtele{\sigma}{\Phi_1}}{\Phi_2}}}{\addtele{\pi}{\Phi_2}}$.
\end{corollary}
\begin{proof}
  This follows directly from \cref{lem:pi-lift-commute-ren} by taking $\sigma$ to be $\addtele{\sigma}{\Phi_1}$ and $\Phi$ to be $\Phi_2$, and realising that $\extendsctx{\addtele{\sigma}{\Phi_1}}{\mu} = \lift{(\addtele{\sigma}{\Phi_1})}$.
\end{proof}

We also need a result like \cref{lem:pi-lift-commute-ren}, but where $\sigma$ is an atomic substitution instead of an atomic renaming.
\begin{lemma}
  \label{lem:pi-lift-commute-sub}
  Let $\Phi : \Tele{n}{m}$ be a scoping telescope, $\sfasub{\sigma}{\hat{\Gamma}}{\hat{\Delta}}{m}$ an atomic \subfree{} substitution and $\sfexpr{\extendsctx{\hat{\Delta}}{\Phi}}{t}{n}$ an expression.
  Then $\asubexpr{\arenexpr{t}{\addtele{\pi}{\Phi}}}{\addtele{\lift{\sigma}}{\Phi}} = \arenexpr{\asubexpr{t}{\addtele{\sigma}{\Phi}}}{\addtele{\pi}{\Phi}}$.
\end{lemma}
\begin{proof}
  The proof is similar to that of \cref{lem:pi-lift-commute-ren}.
  We make use of \cref{prop:mixseq-equiv-tele}, and now we really have two sequences both consisting of an atomic renaming and an atomic substitution.
  Hence, we have to show that $\asubexpr{\arenexpr{v}{\addtele{\pi}{\Phi}}}{\addtele{\lift{\sigma}}{\Phi}} = \arenexpr{\asubexpr{v}{\addtele{\sigma}{\Phi}}}{\addtele{\pi}{\Phi}}$ for every variable $\sfvar{\extendsctx{\hat{\Delta}}{\Phi}}{v}{n}$.
  We will do this by induction on the number of variables in the scoping telescope $\Phi$.
  \begin{itemize}
  \item \case{} $\Phi = \Lambda$, so $\Phi$ contains no variables. \\
    Now we can compute that
    \begin{align*}
      \asubexpr{\arenexpr{v}{\addtele{\pi}{\Lambda}}}{\addtele{\lift{\sigma}}{\Lambda}}
      &= \asubexpr[\Lambda]{\arenexpr[\Lambda]{v}{\pi}}{\lift{\sigma}} \\
      &= \asubexpr[\Lambda]{\suc{v}}{\lift{\sigma}} \\
      &= \arenexpr[\Lambda]{\asubexpr[\Lambda]{v}{\sigma}}{\pi} && \text{(\cref{lem:lift-asub-var})} \\
      &= \arenexpr{\asubexpr{v}{\addtele{\sigma}{\Lambda}}}{\addtele{\pi}{\Lambda}}.
    \end{align*}
  \item \case{} $\Phi = \extendsctx{\extendsctx{\Phi'}{\rho}}{\Lambda}$ \\
    We now have to distinguish two cases for the variable $v$.
    \begin{itemize}
    \item \case{} $v = \vzero{\alpha}$ \\
      The computations go as follows.
      \begin{align*}
        &\asubexpr{\arenexpr{\vzero{\alpha}}{\addtele{\pi}{\extendsctx{\extendsctx{\Phi'}{\rho}}{\Lambda}}}}{\addtele{\lift{\sigma}}{\extendsctx{\extendsctx{\Phi'}{\rho}}{\Lambda}}} \\
        &= \asubexpr[\Lambda]{\arenexpr[\Lambda]{\vzero{\alpha}}{\lift{(\addtele{\pi}{\Phi'})}}}{\lift{(\addtele{\lift{\sigma}}{\Phi'})}} \\
        &= \asubexpr[\Lambda]{\vzero{\alpha}}{\lift{(\addtele{\lift{\sigma}}{\Phi'})}} && \text{(\cref{lem:lift-aren-var})} \\
        &= \vzero{\alpha} && \text{(\cref{lem:lift-asub-var})}
      \end{align*}
      \begin{align*}
        &\arenexpr{\asubexpr{\vzero{\alpha}}{\addtele{\sigma}{\extendsctx{\extendsctx{\Phi'}{\rho}}{\Lambda}}}}{\addtele{\pi}{\extendsctx{\extendsctx{\Phi'}{\rho}}{\Lambda}}} \\
        &= \arenexpr[\Lambda]{\asubexpr[\Lambda]{\vzero{\alpha}}{\lift{(\addtele{\sigma}{\Phi'})}}}{\lift{(\addtele{\pi}{\Phi'})}} \\
        &= \arenexpr[\Lambda]{\vzero{\alpha}}{\lift{(\addtele{\pi}{\Phi'})}} && \text{(\cref{lem:lift-asub-var})} \\
        &= \vzero{\alpha} && \text{(\cref{lem:lift-aren-var})}
      \end{align*}
    \item \case{} $v = \suc{v'}$ \\
      Now we can compute
      \begin{align*}
        &\asubexpr{\arenexpr{\suc{v'}}{\addtele{\pi}{\extendsctx{\extendsctx{\Phi'}{\rho}}{\Lambda}}}}{\addtele{\lift{\sigma}}{\extendsctx{\extendsctx{\Phi'}{\rho}}{\Lambda}}} \\
        &= \asubexpr[\Lambda]{\arenexpr[\Lambda]{\suc{v'}}{\lift{(\addtele{\pi}{\Phi'})}}}{\lift{(\addtele{\lift{\sigma}}{\Phi'})}} \\
        &= \asubexpr[\Lambda]{\suc{\arenexpr[\Lambda]{v'}{\addtele{\pi}{\Phi'}}}}{\lift{(\addtele{\lift{\sigma}}{\Phi'})}} && \text{(\cref{lem:lift-aren-var})} \\
        &= \arenexpr{\asubexpr{\arenexpr{v'}{\addtele{\pi}{\extendsctx{\Phi'}{\Lambda}}}}{\addtele{\lift{\sigma}}{\extendsctx{\Phi'}{\Lambda}}}}{\addtele{\pi}{\Lambda}} && \text{(\cref{lem:lift-asub-var})}
      \end{align*}
      \begin{align*}
        &\arenexpr{\asubexpr{\suc{v'}}{\addtele{\sigma}{\extendsctx{\extendsctx{\Phi'}{\rho}}{\Lambda}}}}{\addtele{\pi}{\extendsctx{\extendsctx{\Phi'}{\rho}}{\Lambda}}} \\
        &= \arenexpr{\asubexpr[\Lambda]{\suc{v'}}{\left\lift{(\addtele{\sigma}{\Phi'}\right)}}}{\addtele{\pi}{\extendsctx{\extendsctx{\Phi'}{\rho}}{\Lambda}}} \\
        &= \arenexpr{\arenexpr{\asubexpr{v'}{\addtele{\sigma}{\extendsctx{\Phi'}{\Lambda}}}}{\addtele{\pi}{\Lambda}}}{\addtele{\pi}{\extendsctx{\extendsctx{\Phi'}{\rho}}{\Lambda}}} && \text{(\cref{lem:lift-asub-var})}
      \end{align*}
      The induction hypothesis with scoping telescope $\extendsctx{\Phi'}{\Lambda}$ (which has one variable less than $\Phi$) gives us that $\asubexpr{\arenexpr{v'}{\addtele{\pi}{\extendsctx{\Phi'}{\Lambda}}}}{\addtele{\lift{\sigma}}{\extendsctx{\Phi'}{\Lambda}}} = \arenexpr{\asubexpr{v'}{\addtele{\sigma}{\extendsctx{\Phi'}{\Lambda}}}}{\addtele{\pi}{\extendsctx{\Phi'}{\Lambda}}}$.
      The result then follows from \cref{cor:pi-lift-commute-ren-general} with $t = \asubexpr{v'}{\addtele{\sigma}{\extendsctx{\Phi'}{\Lambda}}}$, $\sigma = \pi$, $\Phi_1 = \Phi'$, $\mu = \rho$ and $\Phi_2 = \Lambda$. \qedhere
    \end{itemize}
  \end{itemize}
\end{proof}

Combining \cref{lem:pi-lift-commute-ren,lem:pi-lift-commute-sub}, we get the following result.
\begin{lemma}
  \label{lem:pi-lift-commute-mix}
  Let $\Phi : \Tele{n}{m}$ be a scoping telescope, $\sfmixseq{\bar{\sigma}}{\hat{\Gamma}}{\hat{\Delta}}{m}$ a mixed sequence of atomic renamings and substitution and $\sfexpr{\extendsctx{\hat{\Delta}}{\Phi}}{t}{n}$ an SFMTT expression.
  Then $\seqexpr{\arenexpr{t}{\addtele{\pi}{\Phi}}}{\addtele{\lift{\bar{\sigma}}}{\Phi}} = \arenexpr{\seqexpr{t}{\addtele{\bar{\sigma}}{\Phi}}}{\addtele{\pi}{\Phi}}$.
\end{lemma}
\begin{proof}
  In \cref{fig:mix-arensub-seq} we see that the lifting and lock operations on mixed sequences of atomic rensubs consist of applying these operations to all constituent atomic rensubs.
  From this we deduce that also applying a general scoping telescope $\Phi$ to such a mixed sequence amounts to applying $\Phi$ to every constituent atomic rensub.
  Hence the result follows by repeatedly using \cref{lem:pi-lift-commute-ren,lem:pi-lift-commute-sub} for every atomic rensub in $\bar{\sigma}$.
\end{proof}

\subsubsection{Proof Technique (Part 2)}

Using the results from the previous sections, we can now relax the requirement from \cref{prop:mixseq-equiv-tele} so that we only need to check the equality of applying two mixed sequences to a variable after adding a lock telescope instead of a general scoping telescope.
\begin{proposition}
  \label{prop:mixseq-equiv-locktele}
  If $\sfmixseq{\bar{\sigma}, \bar{\tau}}{\hat{\Gamma}}{\hat{\Delta}}{m}$ are two mixed sequences of SFMTT atomic rensubs such that $\seqexpr{v}{\addtele{\bar{\sigma}}{\Lambda}} = \seqexpr{v}{\addtele{\bar{\tau}}{\Lambda}}$ for every lock telescope $\Lambda : \Locktele{n}{m}$ and every variable $\sfvar{\addtele{\hat{\Delta}}{\Lambda}}{v}{n}$, then $\seqexpr{t}{\bar{\sigma}} = \seqexpr{t}{\bar{\tau}}$ for all expressions $\sfexpr{\hat{\Delta}}{t}{m}$.
\end{proposition}
\begin{proof}
  We make use of \cref{prop:mixseq-equiv-tele}, so we have to show that $\seqexpr{v}{\addtele{\bar{\sigma}}{\Phi}} = \seqexpr{v}{\addtele{\bar{\tau}}{\Phi}}$ for every scoping telescope $\Phi : \Tele{n}{m}$ and every variable $\sfvar{\extendsctx{\hat{\Delta}}{\Phi}}{v}{n}$.
  We do this by induction on the number of variables in the scoping telescope $\Phi$.
  \begin{itemize}
  \item \case{} $\Phi = \Lambda$, so there are no variables in $\Phi$. \\
    The result is exactly the assumption of the proposition we are proving.
  \item \case{} $\Phi = \extendsctx{\extendsctx{\Phi'}{\mu}}{\Lambda}$ with $\Lambda$ a lock telescope \\
    We distinguish between the two different cases for the variable $v$.
    \begin{itemize}
    \item \case{} $v = \vzero{\alpha}$ \\
      For every atomic rensub $\sfarensub{\chi}{\hat{\Gamma}}{\hat{\Delta}}{m}$ we have that
      \begin{align*}
        \arensubexpr{\vzero{\alpha}}{\addtele{\chi}{\extendsctx{\extendsctx{\Phi'}{\mu}}{\Lambda}}}
        = \arensubexpr[\Lambda]{\vzero{\alpha}}{\lift{(\addtele{\chi}{\Phi'})}}
        = \vzero{\alpha}.
        && \text{(\cref{lem:lift-aren-var,lem:lift-asub-var})}
      \end{align*}
      By repeatedly applying this result it follows that the same is true for sequences of atomic rensubs.
      In particular, we can conclude that $\seqexpr{\vzero{\alpha}}{\addtele{\bar{\sigma}}{\extendsctx{\extendsctx{\Phi'}{\mu}}{\Lambda}}} = \vzero{\alpha} = \seqexpr{\vzero{\alpha}}{\addtele{\bar{\tau}}{\extendsctx{\extendsctx{\Phi'}{\mu}}{\Lambda}}}$.
    \item \case{} $v = \suc{v'}$ \\
      For any sequence of atomic rensubs $\sfmixseq{\bar{\chi}}{\hat{\Gamma}}{\hat{\Delta}}{m}$ we can compute as follows
      \begin{align*}
        \seqexpr{\suc{v'}}{\addtele{\bar{\chi}}{\extendsctx{\extendsctx{\Phi'}{\mu}}{\Lambda}}}
        &= \seqexpr{\suc{v'}}{\addtele{\lift{(\addtele{\bar{\chi}}{\Phi'})}}{\Lambda}} \\
        &= \seqexpr{\arenexpr{v'}{\addtele{\pi}{\Lambda}}}{\addtele{\lift{(\addtele{\bar{\chi}}{\Phi'})}}{\Lambda}} \\
        &= \arenexpr{\seqexpr{v'}{\addtele{\addtele{\bar{\chi}}{\Phi'}}{\Lambda}}}{\addtele{\pi}{\Lambda}} && \text{(\cref{lem:pi-lift-commute-mix})}
      \end{align*}
      By the induction hypothesis we know that $\seqexpr{v'}{\addtele{\addtele{\bar{\sigma}}{\Phi'}}{\Lambda}} = \seqexpr{v'}{\addtele{\addtele{\bar{\tau}}{\Phi'}}{\Lambda}}$.
      Hence we can conclude that
      \begin{align*}
        \seqexpr{\suc{v'}}{\addtele{\bar{\sigma}}{\extendsctx{\extendsctx{\Phi'}{\mu}}{\Lambda}}}
        &= \arenexpr{\seqexpr{v'}{\addtele{\addtele{\bar{\sigma}}{\Phi'}}{\Lambda}}}{\addtele{\pi}{\Lambda}} \\
        &= \arenexpr{\seqexpr{v'}{\addtele{\addtele{\bar{\tau}}{\Phi'}}{\Lambda}}}{\addtele{\pi}{\Lambda}} \\
        &= \seqexpr{\suc{v'}}{\addtele{\bar{\tau}}{\extendsctx{\extendsctx{\Phi'}{\mu}}{\Lambda}}}. \qedhere
      \end{align*}
    \end{itemize}
  \end{itemize}
\end{proof}

In particular, we have the following proof technique for observational equivalence of regular SFMTT substitutions.
\begin{proposition}
  \label{prop:sfmtt-sigma-equiv-var-locktele}
  Let $\sfsub{\sigma, \tau}{\hat{\Gamma}}{\hat{\Delta}}{m}$ be two \subfree{} substitutions and suppose that $\subexpr{v}{\addtele{\sigma}{\Lambda}} = \subexpr{v}{\addtele{\tau}{\Lambda}}$ for every lock telescope $\Lambda : \Locktele{n}{m}$ and every variable $\sfvar{\extendsctx{\hat{\Delta}}{\Lambda}}{v}{n}$.
  Then $\sfsigmeqsub{\sigma}{\tau}$.
\end{proposition}
\begin{proof}
  Given the definition of observational equivalence for SFMTT substitutions, this follows immediately from \cref{prop:mixseq-equiv-locktele} where both sequences consist of only atomic substitutions (so no renamings).%
  \footnote{Strictly speaking, we should define an embedding of regular SFMTT substitutions into mixed sequences of atomic rensubs and prove that their actions on SFMTT expressions correspond, but this is trivial.}
\end{proof}
\begin{example} \label{ex:subst-existence}
  If we instantiate \subfree{} on the trivial mode theory (by which we mean the terminal 2-category) then variables are non-modal De Bruijn indices and lock telescopes can be essentially ignored.
  In this setting, what \cref{prop:sfmtt-sigma-equiv-var-locktele} really says is that a substitution is uniquely determined, up to observational equivalence, by its action on De Bruijn indices.
  Since there exists exactly one De Bruijn index for every variable in the context, this means that we have an injection from substitutions, up to  observational equivalence, to vectors of terms.
  In plain dependent type theory, substitutions are often \emph{defined} as vectors of terms, or at least it is clear that they can be uniquely represented in this way.
  In other words, the aforementioned injection is actually a bijection.
  
  Thus, it is natural to ask whether this idea carries over to general SFMTT.
  Could we define an SFMTT substitution $\sfsub{\sigma}{\hat{\Gamma}}{\hat{\Delta}}{m}$ as a thing that assigns, to every lock telescope $\Lambda : \Tele{n}{m}$ and every variable $\sfvar{\extendsctx{\hat{\Delta}}{\Lambda}}{v}{n}$ a term $\subexpr{v}{\addtele{\sigma}{\Lambda}}$, perhaps satisfying some coherence conditions?
  Let us call such an assignment a \emph{substitution observation}.
  Then \cref{prop:sfmtt-sigma-equiv-var-locktele} asserts that there is an injection from substitutions, up to observational equivalence, to substitution observations.
  We are asking if this injection is in fact a bijection.
  
  The answer is no.
  Consider, as mode theory, the walking arrow, i.e.\ the 2-category with two modes $m$ and $n$, one modality $\mu : m \to n$, and only identity 2-cells.
  Then a substitution observation in a context of the form $\hat \Delta = (\lock{\extendsctx{\emptyctx}{\unitMod}}{\mu})$ carries no information.
  Indeed, $\hat \Delta$ lives at mode $m$ and no lock telescope can get us back to $n$, which is where the only introduced variable lives.
  Thus, for any other context $\hat \Gamma$, there exists a unique substitution observation from $\hat \Gamma$ to $\hat \Delta$.
  However, if we and instantiate $\hat \Gamma = \emptyctx$, then there exists no substitution $\sfsub{\sigma}{\hat{\Gamma}}{\hat{\Delta}}{m}$. Indeed, since the only 2-cell with codomain $\mu$ is the identity, it is impossible to get rid of $\locksym_\mu$ in the domain of $\sigma$.
  
  A cleaner argument can be given in the typed case.
  There, we could type $\hat \Delta$ as $\Delta = (\lock{\extendctx{\emptyctx}{\unitMod}{\mathsf{Empty}}}{\mu})$ and instantiate $\Gamma = (\lock{\emptyctx}{\mu})$ and now $\locksym_\mu$ is no longer the problem, but we still cannot construct a substitution as there are no closed terms of the empty type $\mathsf{Empty}$.
  
  This situation is caused by an intentional underspecification of what $\locksym_\mu$ does.
  For a \emph{general} model of WSMTT with said mode theory, it is not sound to allow mentions of the variable in context $\Delta$, and thus substitution observations to $\Delta$ are devoid of information.
  However, $\mu$ \emph{could} be the identity modality, in which case a substitution from $\Gamma$ to $\Delta$ should really not exist, but there would be no qualms against mentioning the variable in context $\Delta$.
\end{example}

\subsection{Preservation of Observational Equivalence of \subfree{} Substitutions}

\Cref{def:sfmtt-sub-sigma-equiv} tells us that two \subfree{} substitutions are observationally equivalent if they yield equal results when applied to any expression.
It is not immediately clear that this property is preserved by some of the operations that act on substitutions, such as $\locksym_\mu$ or lifting.
The following lemmas tell us that this is indeed the case.

\begin{lemma}
  \label{lem:lock-preserves-sigma-eq}
  Let $\sfsub{\sigma, \tau}{\hat{\Gamma}}{\hat{\Delta}}{n}$ be two \subfree{} substitutions and $\mu : m \to n$ a modality.
  If $\sfsigmeqsub{\sigma}{\tau}$, then also $\sfsigmeqsub{\lock{\sigma}{\mu}}{\lock{\tau}{\mu}}$.
\end{lemma}
\begin{proof}
  Take an arbitrary expression $\sfexpr{\lock{\hat{\Delta}}{\mu}}{t}{m}$.
  Then we can apply \inlinerulename{sf-expr-mod-tm} to see that $\sfexpr{\hat{\Delta}}{\modtm{\mu}{t}}{n}$.
  Hence, since $\sfsigmeqsub{\sigma}{\tau}$, the definition of observational equivalence tells us that $\subexpr{\left(\modtm{\mu}{t}\right)}{\sigma} = \subexpr{\left(\modtm{\mu}{t}\right)}{\tau}$.
  Since applying a lock to a regular \subfree{} substitution amounts to applying the lock to all its constituent atomic substitutions, it follows that $\subexpr{\left(\modtm{\mu}{t}\right)}{\sigma} = \modtm{\mu}{\subexpr{t}{\lock{\sigma}{\mu}}}$ (and similarly for $\tau$).
  We therefore have that $\modtm{\mu}{\subexpr{t}{\lock{\sigma}{\mu}}} = \modtm{\mu}{\subexpr{t}{\lock{\tau}{\mu}}}$ and by injectivity of expression constructors it follows that $\subexpr{t}{\lock{\sigma}{\mu}} = \subexpr{t}{\lock{\tau}{\mu}}$.
  As this holds for arbitrary $t$, we have proven that $\sfsigmeqsub{\lock{\sigma}{\mu}}{\lock{\tau}{\mu}}$.
\end{proof}
\begin{lemma}
  \label{lem:lift-preserves-sigma-eq}
  Let $\sfsub{\sigma, \tau}{\hat{\Gamma}}{\hat{\Delta}}{m}$ be two \subfree{} substitutions.
  If $\sfsigmeqsub{\sigma}{\tau}$, then also $\sfsigmeqsub{\lift{\sigma}}{\lift{\tau}}$.
\end{lemma}
\begin{proof}
  We can apply the same reasoning as in the proof of \cref{lem:lock-preserves-sigma-eq}, but with the expression constructor $\lam{\mu}{\_}$ instead of $\modtm{\mu}{\_}$.
\end{proof}
\begin{corollary}
  \label{cor:tel-preserves-sigma-eq}
  If $\sfsub{\sigma, \tau}{\hat{\Gamma}}{\hat{\Delta}}{m}$ are two \subfree{} substitutions and $\Phi : \Tele{n}{m}$ is a scoping telescope, then $\sfsigmeqsub{\sigma}{\tau}$ implies $\sfsigmeqsub{\addtele{\sigma}{\Phi}}{\addtele{\tau}{\Phi}}$.
\end{corollary}
We note that the converse of \cref{prop:sfmtt-sigma-equiv-var} immediately follows from \cref{cor:tel-preserves-sigma-eq}.
Furthermore, if we restrict the scoping telescopes in this corollary to lock telescopes, the converse of \cref{prop:sfmtt-sigma-equiv-var-locktele} can also be derived.

\subsection{Relating WSMTT and \subfree{} Lifting}

\begin{lemma}
  \label{lem:lift-mtt-to-sfmtt}
  Given a WSMTT substitution $\mttsub{\sigma}{\hat{\Gamma}}{\hat{\Delta}}{m}$, we have $\sfsigmeqsub{\translate{\lift{\sigma}}}{\lift{\translate{\sigma}}}$.
\end{lemma}
\begin{proof}
  First of all, we can calculate that
  \begin{align*}
    \translate{\lift{\sigma}}
    &= \translate{(\sigma \circ \pi) . \vzero{}} && \text{(Definition of $^+$, \cref{eq:wsmtt-sub-lift})} \\
    &= \rensubcons{\lift{\translate{\sigma \circ \pi}}}{(\ida{} . \translate{\vzero{}})} && \explanation{Definition of $\translate{\_}$} \\
    &= \rensubcons{\lift{(\concat{\translate{\sigma}}{\translate{\pi}})}}{\left(\ida{} . \vzero{1_\mu}\right)} && \explanation{Definition of $\translate{\_}$} \\
    &= \rensubcons{\rensubcons{\lift{\translate{\sigma}}}{\lift{\pi}}}{\left(\ida{} . \vzero{1_\mu}\right)}.
  \end{align*}
  The last step combines the definition of $\translate{\pi}$ with the fact that lifting a regular substitution amounts to lifting all of its constituent substitutions.
  By the definition of $\sfsigmeqsubsym$ it now suffices to prove that $\asubexpr{\asubexpr{t}{\lift{\pi}}}{\ida{} . \vzero{1_\mu}} = t$ for every expression $\sfexpr{\extendsctx{\hat{\Gamma}}{\mu}}{t}{m}$.
  For this we use \cref{prop:mixseq-equiv-locktele}, so we have to show that $\asubexpr{\asubexpr{v}{\addtele{\lift{\pi}}{\Lambda}}}{\addtele{\left(\ida{} . \vzero{1_\mu}\right)}{\Lambda}} = v$ for every lock telescope $\Lambda : \Locktele{n}{m}$ and every variable $\sfvar{\addtele{\extendsctx{\hat{\Gamma}}{\mu}}{\Lambda}}{v}{n}$.
  We distinguish between two cases for the variable $v$.
  \begin{itemize}
  \item \case{} $v = \vzero{\alpha}$ \\
    We can now compute that
    \begin{align*}
      &\asubexpr{\asubexpr{\vzero{\alpha}}{\addtele{\lift{\pi}}{\Lambda}}}{\addtele{\left(\ida{} . \vzero{1_\mu}\right)}{\Lambda}} \\
      &= \asubexpr[\Lambda]{\asubexpr[\Lambda]{\vzero{\alpha}}{\lift{\pi}}}{\ida{} . \vzero{1_\mu}} \\
      &= \asubexpr[\Lambda]{\vzero{\alpha}}{\ida{}.\vzero{1_\mu}} && \text{(\cref{lem:lift-asub-var})} \\
      &= \transf{\locksym_\mu}{\Lambda}{\alpha}{\vzero{1_\mu}} && \explanation{\cref{eq:asubvar-extend-vzero,eq:arenvar-key}} \\
      &= \vzero{\alpha}.
    \end{align*}
  \item \case{} $v = \suc{v'}$ \\
    Then we have that
    \begin{align*}
      &\asubexpr{\asubexpr{\suc{v'}}{\addtele{\lift{\pi}}{\Lambda}}}{\addtele{\left(\ida{} . \vzero{1_\mu}\right)}{\Lambda}} \\
      &= \asubexpr[\Lambda]{\asubexpr[\Lambda]{\asubexpr[\Lambda]{v'}{\pi}}{\pi}}{\ida{} . \vzero{1_\mu}} && \text{(\cref{lem:lift-asub-var})} \\
      &= \asubexpr[\Lambda]{\suc{\suc{v'}}}{\ida{} . \vzero{1_\mu}} \\
      &= \asubexpr[\Lambda]{\suc{v'}}{\ida{}} && \explanation{\cref{eq:asubvar-extend-vsuc}} \\
      &= \suc{v'}. && \qedhere
    \end{align*}
  \end{itemize}
\end{proof}

\subsection{Properties of Key Renamings}

In order to prove the completeness of the substitution algorithm, we need a counterpart in SFMTT for every rule in \cref{fig:sigma-equiv} relating to key substitutions.
That is exactly what will be covered in this section, but we start with two auxiliary results.
\begin{lemma}
  \label{lem:key-ren-suc-var}
  Let $\Lambda : \Locktele{n}{m}$ and $\Theta, \Psi : \Locktele{o}{n}$ and $\Omega : \Locktele{p}{o}$ be lock telescopes, $\twocell{\alpha}{\locks{\Theta}}{\locks{\Psi}}$ a 2-cell, and $\sfvar{\addtele{\addtele{\addtele{\hat{\Gamma}}{\Lambda}}{\Theta}}{\Omega}}{v}{p}$ a variable. Then $\arenexpr{\suc{v}}{\addtele{\key{\alpha}{\Theta}{\Psi}{\addtele{\extendsctx{\hat{\Gamma}}{\mu}}{\Lambda}}}{\Omega}} = \suc{\arenexpr{v}{\addtele{\key{\alpha}{\Theta}{\Psi}{\addtele{\hat{\Gamma}}{\Lambda}}}{\Omega}}}$.
\end{lemma}
\begin{proof}
  We can compute that
  \begin{align*}
    \arenexpr{\suc{v}}{\addtele{\key{\alpha}{\Theta}{\Psi}{\addtele{\extendsctx{\hat{\Gamma}}{\mu}}{\Lambda}}}{\Omega}}
    &= \arenexpr[\Omega]{\suc{v}}{\key{\alpha}{\Theta}{\Psi}{\addtele{\extendsctx{\hat{\Gamma}}{\mu}}{\Lambda}}} \\
    &= \transf{\addtele{\Theta}{\Omega}}{\addtele{\Psi}{\Omega}}{\alpha \star \unitTwocell_{\locks{\Omega}}}{\suc{v}} && \explanation{\cref{eq:arenvar-key}} \\
    &= \suc{\transf{\addtele{\Theta}{\Omega}}{\addtele{\Psi}{\Omega}}{\alpha \star \unitTwocell_{\locks{\Omega}}}{v}} && \explanation{\cref{eq:twocell-vsuc}} \\
    &= \suc{\arenexpr{v}{\addtele{\key{\alpha}{\Theta}{\Psi}{\addtele{\hat{\Gamma}}{\Lambda}}}{\Omega}}} && \explanation{\cref{eq:arenvar-key}} \qedhere
  \end{align*}
\end{proof}
\begin{lemma}
  \label{lem:key-pi-commute}
  Key renamings commute with $\pi$ renamings.
  In other words, we have $\arenexpr{\arenexpr{t}{\key{\alpha}{\Lambda}{\Theta}{\hat{\Gamma}}}}{\addtele{\pi}{\Theta}} = \arenexpr{\arenexpr{t}{\addtele{\pi}{\Lambda}}}{\key{\alpha}{\Lambda}{\Theta}{\extendsctx{\hat{\Gamma}}{\mu}}}$ for every expression $\sfexpr{\addtele{\hat{\Gamma}}{\Lambda}}{t}{m}$.
\end{lemma}
\begin{proof}
  We use \cref{prop:mixseq-equiv-locktele}, so we take an arbitrary lock telescope $\Psi$ and a variable $\sfvar{\addtele{\hat{\Gamma}}{\extendsctx{\Lambda}{\Psi}}}{v}{n}$.
  Then we can compute that
  \begin{align*}
    \arenexpr[\Psi]{\arenexpr[\Psi]{v}{\key{\alpha}{\Lambda}{\Theta}{\hat{\Gamma}}}}{\addtele{\pi}{\Theta}}
    &= \suc{\arenexpr[\Psi]{v}{\key{\alpha}{\Lambda}{\Theta}{\hat{\Gamma}}}} \\
    &= \arenexpr[\Psi]{\suc{v}}{\key{\alpha}{\Lambda}{\Theta}{\extendsctx{\hat{\Gamma}}{\mu}}} && \explanation{\cref{lem:key-ren-suc-var}} \\
    &= \arenexpr[\Psi]{\arenexpr[\Psi]{v}{\addtele{\pi}{\Lambda}}}{\key{\alpha}{\Lambda}{\Theta}{\extendsctx{\hat{\Gamma}}{\mu}}}. && \qedhere
  \end{align*}
\end{proof}

\begin{proposition}
  \label{prop:key-unit-cell}
  For every lock telescope $\Lambda : \Locktele{n}{m}$ and \subfree{} expression $\sfexpr{\addtele{\hat{\Gamma}}{\Lambda}}{t}{n}$ we have that $\arenexpr{t}{\key{1_{\locks{\Lambda}}}{\Lambda}{\Lambda}{\hat{\Gamma}}} = t$.
\end{proposition}
\begin{proof}
  We use \cref{prop:mixseq-equiv-locktele}, so we have to show that $\arenexpr{v}{\addtele{\key{1_{\locks{\Lambda}}}{\Lambda}{\Lambda}{\hat{\Gamma}}}{\Theta}} = v$ for all lock telescopes $\Theta : \Locktele{o}{n}$ and variables $\sfvar{\addtele{\hat{\Gamma}}{\extendsctx{\Lambda}{\Theta}}}{v}{o}$.
  This proof proceeds by induction on the variable $v$.
  \begin{itemize}
  \item \case{} $v = \vzero{\alpha}$ with $\hat{\Gamma} = \addtele{\extendsctx{\hat{\Gamma}'}{\mu}}{\Psi}$ \\
    We have
    \begin{align*}
      &\arenexpr{\vzero{\alpha}}{\addtele{\key{1_{\locks{\Lambda}}}{\Lambda}{\Lambda}{\addtele{\extendsctx{\hat{\Gamma}'}{\mu}}{\Psi}}}{\Theta}} \\
      &= \arenexpr[\Theta]{\vzero{\alpha}}{\key{1_{\locks{\Lambda}}}{\Lambda}{\Lambda}{\addtele{\extendsctx{\hat{\Gamma}'}{\mu}}{\Psi}}} \\
      &= \transf{\extendsctx{\Lambda}{\Theta}}{\extendsctx{\Lambda}{\Theta}}{1_{\locks{\Lambda}} \star 1_{\locks{\Theta}}}{\vzero{\alpha}} && \explanation{\cref{eq:arenvar-key}} \\
      &= \vzero{(1_{\locks{\Psi}} \star (1_{\locks{\Lambda}} \star 1_{\locks{\Theta}})) \circ \alpha} && \explanation{\cref{eq:twocell-vzero}} \\
      &= \vzero{\alpha}. && \explanation{Strict 2-category laws}
    \end{align*}
  \item \case{} $v = \suc{v'}$ with $\hat{\Gamma} = \addtele{\hat{\Gamma}'}{\extendsctx{\mu}{\Psi}}$ \\
    Now we can compute
    \begin{align*}
      &\arenexpr{\suc{v'}}{\addtele{\key{1_{\locks{\Lambda}}}{\Lambda}{\Lambda}{\addtele{\hat{\Gamma}'}{\extendsctx{\mu}{\Psi}}}}{\Theta}} \\
      &= \suc{\arenexpr{v'}{\addtele{\key{1_{\locks{\Lambda}}}{\Lambda}{\Lambda}{\addtele{\hat{\Gamma}'}{\Psi}}}{\Theta}}} && \explanation{\cref{lem:key-ren-suc-var}} \\
      &= \suc{v'}. && \explanation{Induction hypothesis} \qedhere
    \end{align*}
  \end{itemize}
\end{proof}

\begin{proposition}
  \label{prop:key-vert-comp-cell}
  If $\Lambda_1, \Lambda_2, \Lambda_3 : \Locktele{n}{m}$ are lock telescopes, $\twocell{\alpha}{\locks{\Lambda_1}}{\locks{\Lambda_2}}$ and $\twocell{\beta}{\locks{\Lambda_2}}{\locks{\Lambda_3}}$ are 2-cells and $\sfexpr{\addtele{\hat{\Gamma}}{\Lambda_1}}{t}{n}$ is an expression, then $\arenexpr{t}{\key{\beta \circ \alpha}{\Lambda_1}{\Lambda_3}{\hat{\Gamma}}} = \arenexpr{\arenexpr{t}{\key{\alpha}{\Lambda_1}{\Lambda_2}{\hat{\Gamma}}}}{\key{\beta}{\Lambda_2}{\Lambda_3}{\hat{\Gamma}}}$.
\end{proposition}
\begin{proof}
  The proof is similar to that of \cref{prop:key-unit-cell}, so we use \cref{prop:mixseq-equiv-locktele} and take an arbitrary lock telescope $\Theta :\Locktele{o}{n}$ and variable $\sfvar{\addtele{\hat{\Gamma}}{\extendsctx{\Lambda_1}{\Theta}}}{v}{o}$. Then we prove that $\arenexpr[\Theta]{v}{\key{\beta \circ \alpha}{\Lambda_1}{\Lambda_3}{\hat{\Gamma}}} = \arenexpr[\Theta]{\arenexpr[\Theta]{v}{\key{\alpha}{\Lambda_1}{\Lambda_2}{\hat{\Gamma}}}}{\key{\beta}{\Lambda_2}{\Lambda_3}{\hat{\Gamma}}}$ by induction on $v$.
  \begin{itemize}
  \item \case{} $v = \vzero{\gamma}$ with $\hat{\Gamma} = \addtele{\extendsctx{\hat{\Gamma}'}{\mu}}{\Psi}$ and $\twocell{\gamma}{\mu}{\locks{\addtele{\addtele{\Psi}{\Lambda_1}}{\Theta}}}$ \\
    Now we have
    \begin{align*}
      &\arenexpr[\Theta]{\vzero{\gamma}}{\key{\beta \circ \alpha}{\Lambda_1}{\Lambda_3}{\addtele{\extendsctx{\hat{\Gamma}'}{\mu}}{\Psi}}} \\
      &= \transf{\extendsctx{\Lambda_1}{\Theta}}{\extendsctx{\Lambda_3}{\Theta}}{(\beta \circ \alpha) \star 1_{\locks{\Theta}}}{\vzero{\gamma}} && \explanation{\cref{eq:arenvar-key}} \\
      &= \vzero{\left( 1_{\locks{\Psi}} \star ((\beta \circ \alpha) \star 1_{\locks{\Theta}}) \right) \circ \gamma} && \explanation{\cref{eq:twocell-vzero}} \\
      &= \vzero{\left( 1_{\locks{\Psi}} \star (\beta \star 1_{\locks{\Theta}}) \right) \circ \left( 1_{\locks{\Psi}} \star (\alpha \star 1_{\locks{\Theta}}) \right) \circ \gamma} && \explanation{Strict 2-category laws} \\
      &= \transf{\extendsctx{\Lambda_2}{\Theta}}{\extendsctx{\Lambda_3}{\Theta}}{\beta}{\vzero{\left( 1_{\locks{\Psi}} \star (\alpha \star 1_{\locks{\Theta}}) \right) \circ \gamma}} && \explanation{\cref{eq:twocell-vzero}} \\
      &= \transf{\extendsctx{\Lambda_2}{\Theta}}{\extendsctx{\Lambda_3}{\Theta}}{\beta}{\transf{\extendsctx{\Lambda_1}{\Theta}}{\extendsctx{\Lambda_2}{\Theta}}{\alpha}{\vzero{\gamma}}} && \explanation{\cref{eq:twocell-vzero}} \\
      &= \arenexpr[\Theta]{\arenexpr[\Theta]{\vzero{\gamma}}{\key{\alpha}{\Lambda_1}{\Lambda_2}{\hat{\Gamma}}}}{\key{\beta}{\Lambda_2}{\Lambda_3}{\hat{\Gamma}}}. && \explanation{\cref{eq:arenvar-key}}
    \end{align*}
  \item \case{} $v = \suc{v'}$ with $\hat{\Gamma} = \addtele{\extendsctx{\hat{\Gamma}'}{\mu}}{\Psi}$ \\
    Similarly as in the proof of \cref{prop:key-unit-cell} we compute
    \begin{align*}
      &\arenexpr[\Theta]{\suc{v'}}{\key{\beta \circ \alpha}{\Lambda_1}{\Lambda_3}{\addtele{\extendsctx{\hat{\Gamma}'}{\mu}}{\Psi}}} \\
      &= \suc{\arenexpr[\Theta]{v'}{\key{\beta \circ \alpha}{\Lambda_1}{\Lambda_3}{\addtele{\hat{\Gamma}'}{\Psi}}}} && \explanation{\cref{lem:key-ren-suc-var}} \\
      &= \suc{\arenexpr[\Theta]{\arenexpr[\Theta]{v'}{\key{\alpha}{\Lambda_1}{\Lambda_2}{\addtele{\hat{\Gamma}'}{\Psi}}}}{\key{\beta}{\Lambda_2}{\Lambda_3}{\addtele{\hat{\Gamma}'}{\Psi}}}} && \explanation{Induction hypothesis} \\
      &= \arenexpr[\Theta]{\arenexpr[\Theta]{\suc{v'}}{\key{\alpha}{\Lambda_1}{\Lambda_2}{\addtele{\extendsctx{\hat{\Gamma}'}{\mu}}{\Psi}}}}{\key{\beta}{\Lambda_2}{\Lambda_3}{\addtele{\extendsctx{\hat{\Gamma}'}{\mu}}{\Psi}}}. && \explanation{\cref{lem:key-ren-suc-var}} \qedhere
    \end{align*}
  \end{itemize}
\end{proof}

\begin{proposition}
  \label{prop:key-hor-comp-cell}
  Given lock telescopes $\Lambda_1, \Lambda_2 : \Locktele{n}{m}$, $\Theta_1, \Theta_2 : \Locktele{o}{n}$ and 2-cells $\twocell{\beta}{\locks{\Lambda_1}}{\locks{\Lambda_2}}$ and $\twocell{\alpha}{\locks{\Theta_1}}{\locks{\Theta_2}}$, the following two equations hold for any expression $\sfexpr{\addtele{\hat{\Gamma}}{\extendsctx{\Lambda_1}{\Theta_1}}}{t}{o}$
  \begin{align*}
    \arenexpr{t}{\key{\beta \star \alpha}{\extendsctx{\Lambda_1}{\Theta_1}}{\extendsctx{\Lambda_2}{\Theta_2}}{\hat{\Gamma}}}
    &= \arenexpr{\arenexpr{t}{\addtele{\key{\beta}{\Lambda_1}{\Lambda_2}{\hat{\Gamma}}}{\Theta_1}}}{\key{\alpha}{\Theta_1}{\Theta_2}{\addtele{\hat{\Gamma}}{\Lambda_2}}} \\
    &= \arenexpr{\arenexpr{t}{\key{\alpha}{\Theta_1}{\Theta_2}{\addtele{\hat{\Gamma}}{\Lambda_1}}}}{\addtele{\key{\beta}{\Lambda_1}{\Lambda_2}{\hat{\Gamma}}}{\Theta_2}}.
  \end{align*}
\end{proposition}
\begin{proof}
  We only prove the first equality, the second one can be proved similarly.
  Making use of \cref{prop:mixseq-equiv-locktele}, we introduce a lock telescope $\Psi : \Locktele{p}{o}$ and a variable $\sfvar{\addtele{\hat{\Gamma}}{\extendsctx{\Lambda_1}{\extendsctx{\Theta_1}{\Psi}}}}{v}{p}$, and then we need to show that $\arenexpr[\Psi]{v}{\key{\beta \star \alpha}{\extendsctx{\Lambda_1}{\Theta_1}}{\extendsctx{\Lambda_2}{\Theta_2}}{\hat{\Gamma}}} = \arenexpr[\Psi]{\arenexpr[\Psi]{t}{\addtele{\key{\beta}{\Lambda_1}{\Lambda_2}{\hat{\Gamma}}}{\Theta_1}}}{\key{\alpha}{\Theta_1}{\Theta_2}{\addtele{\hat{\Gamma}}{\Lambda_2}}}$.
  This proof proceeds by induction on $v$.
  \begin{itemize}
    \item \case{} $v = \vzero{\gamma}$ with $\hat{\Gamma} = \addtele{\extendsctx{\hat{\Gamma}'}{\mu}}{\Omega}$ and $\twocell{\gamma}{\mu}{\locks{\extendsctx{\extendsctx{\extendsctx{\Omega}{\Lambda_1}}{\Theta_1}}{\Psi}}}$ \\
      Now we compute that
      \begin{align*}
        &\arenexpr[\Psi]{\vzero{\gamma}}{\key{\beta \star \alpha}{\extendsctx{\Lambda_1}{\Theta_1}}{\extendsctx{\Lambda_2}{\Theta_2}}{\addtele{\extendsctx{\hat{\Gamma}'}{\mu}}{\Omega}}} \\
        &= \transf{\extendsctx{\Lambda_1}{\extendsctx{\Theta_1}{\Psi}}}{\extendsctx{\Lambda_2}{\extendsctx{\Theta_2}{\Psi}}}{(\beta \star \alpha) \star 1_{\locks{\Psi}}}{\vzero{\gamma}} && \explanation{\cref{eq:arenvar-key}} \\
        &= \vzero{(1_{\locks{\Omega}} \star (\beta \star \alpha) \star 1_{\locks{\Psi}}) \circ \gamma} && \explanation{\cref{eq:twocell-vzero}} \\
        &= \vzero{(1_{\locks{\Omega}} \star (1_{\locks{\Lambda_2}} \star \alpha) \star 1_{\locks{\Psi}}) \circ (1_{\locks{\Omega}} \star (\beta \star 1_{\locks{\Theta_1}}) \star 1_{\locks{\Psi}}) \circ \gamma} && \explanation{Strict 2-category laws} \\
        &= \vzero{(1_{\locks{\addtele{\Omega}{\Lambda_2}}} \star (\alpha \star 1_{\locks{\Psi}})) \circ (1_{\locks{\Omega}} \star (\beta \star 1_{\locks{\addtele{\Theta_1}{\Psi}}}) ) \circ \gamma} && \explanation{Strict 2-category laws} \\
        &= \mathrlap{\transf{\extendsctx{\Theta_1}{\Psi}}{\extendsctx{\Theta_2}{\Psi}}{\alpha \star 1_{\locks{\Psi}}}{\transf{\extendsctx{\extendsctx{\Lambda_1}{\Theta_1}}{\Psi}}{\extendsctx{\extendsctx{\Lambda_2}{\Theta_1}}{\Psi}}{\beta \star 1_{\locks{\addtele{\Theta_1}{\Psi}}}}{\vzero{\gamma}}}} \\
        &&& \explanation{\cref{eq:twocell-vzero}} \\
        &= \arenexpr[\Psi]{\arenexpr[\addtele{\Theta_1}{\Psi}]{\vzero{\gamma}}{\key{\beta}{\Lambda_1}{\Lambda_2}{\hat{\Gamma}}}}{\key{\alpha}{\Theta_1}{\Theta_2}{\addtele{\hat{\Gamma}}{\Lambda_2}}} && \explanation{\cref{eq:arenvar-key}} \\
        &= \arenexpr[\Psi]{\arenexpr[\Psi]{\vzero{\gamma}}{\addtele{\key{\beta}{\Lambda_1}{\Lambda_2}{\hat{\Gamma}}}{\Theta_1}}}{\key{\alpha}{\Theta_1}{\Theta_2}{\addtele{\hat{\Gamma}}{\Lambda_2}}}
      \end{align*}
    \item \case{} $v = \suc{v'}$ with $\hat{\Gamma} = \addtele{\extendsctx{\hat{\Gamma}'}{\mu}}{\Omega}$ \\
      In this case we have
      \begin{align*}
        &\arenexpr[\Psi]{\suc{v'}}{\key{\beta \star \alpha}{\extendsctx{\Lambda_1}{\Theta_1}}{\extendsctx{\Lambda_2}{\Theta_2}}{\addtele{\extendsctx{\hat{\Gamma}'}{\mu}}{\Omega}}} \\
        &= \suc{\arenexpr[\Psi]{v'}{\key{\beta \star \alpha}{\extendsctx{\Lambda_1}{\Theta_1}}{\extendsctx{\Lambda_2}{\Theta_2}}{\addtele{\hat{\Gamma}'}{\Omega}}}} && \explanation{\cref{lem:key-ren-suc-var}} \\
        &= \suc{\arenexpr[\Psi]{\arenexpr[\Psi]{v'}{\addtele{\key{\beta}{\Lambda_1}{\Lambda_2}{\hat{\Gamma}}}{\Theta_1}}}{\key{\alpha}{\Theta_1}{\Theta_2}{\addtele{\hat{\Gamma}}{\Lambda_2}}}} && \explanation{Induction hypothesis} \\
        &= \arenexpr[\Psi]{\arenexpr[\Psi]{\suc{v'}}{\addtele{\key{\beta}{\Lambda_1}{\Lambda_2}{\hat{\Gamma}}}{\Theta_1}}}{\key{\alpha}{\Theta_1}{\Theta_2}{\addtele{\hat{\Gamma}}{\Lambda_2}}}. && \explanation{\cref{lem:key-ren-suc-var}} \qedhere
      \end{align*}
  \end{itemize}
\end{proof}

\begin{proposition}
  \label{prop:key-natural}
  Key renamings are natural. In other words, given lock telescopes $\Lambda, \Theta : \Locktele{n}{m}$, a 2-cell $\twocell{\alpha}{\locks{\Lambda}}{\locks{\Theta}}$, a substitution $\sfsub{\sigma}{\hat{\Gamma}}{\hat{\Delta}}{m}$ and an expression $\sfexpr{\addtele{\hat{\Delta}}{\Lambda}}{t}{n}$, we have that $\subexpr{\arenexpr{t}{\key{\alpha}{\Lambda}{\Theta}{\hat{\Delta}}}}{\addtele{\sigma}{\Theta}} = \arenexpr{\subexpr{t}{\addtele{\sigma}{\Lambda}}}{\key{\alpha}{\Lambda}{\Theta}{\hat{\Gamma}}}$.
\end{proposition}
\begin{proof}
  It suffices to prove this lemma for an atomic substitution $\sigma$, for which we use \cref{prop:mixseq-equiv-locktele}.
  Hence for an arbitrary lock telescope $\Psi : \Locktele{o}{n}$ and variable $\sfvar{\addtele{\hat{\Delta}}{\extendsctx{\Lambda}{\Psi}}}{v}{o}$ we show that $\asubexpr[\extendsctx{\Theta}{\Psi}]{\arenexpr[\Psi]{v}{\key{\alpha}{\Lambda}{\Theta}{\hat{\Delta}}}}{\sigma} = \arenexpr[\Psi]{\asubexpr[\extendsctx{\Lambda}{\Psi}]{v}{\sigma}}{\key{\alpha}{\Lambda}{\Theta}{\hat{\Gamma}}}$.
  We do this by induction on $\sigma$.
  \begin{itemize}
  \item \case{} $\sigma = \emptysub$ (\inlinerulename{sf-arensub-empty}) \\
    Now $\hat{\Delta}$ is the empty scoping context. Since there are no variables in $\addtele{\emptyctx}{\extendsctx{\Lambda}{\Psi}}$, this case is trivial.
  \item \case{} $\sigma = \ida$ (\inlinerulename{sf-arensub-id}) \\
    Since the action of $\ida$ on variables is the identity, this case is also trivial.
  \item \case{} $\sigma = \weaken{\sigma'}$ with $\hat{\Gamma} = \extendsctx{\hat{\Gamma}'}{\mu}$ and $\sfasub{\sigma'}{\hat{\Gamma}'}{\hat{\Delta}}{m}$ (\inlinerulename{sf-arensub-weaken}) \\
    We have
    \begin{align*}
      &\asubexpr[\extendsctx{\Theta}{\Psi}]{\arenexpr[\Psi]{v}{\key{\alpha}{\Lambda}{\Theta}{\hat{\Delta}}}}{\weaken{\sigma'}} \\
      &= \arenexpr[\extendsctx{\Theta}{\Psi}]{\asubexpr[\extendsctx{\Theta}{\Psi}]{\arenexpr[\Psi]{v}{\key{\alpha}{\Lambda}{\Theta}{\hat{\Delta}}}}{\sigma'}}{\pi} && \explanation{\cref{eq:asubvar-weaken}} \\
      &= \arenexpr[\extendsctx{\Theta}{\Psi}]{\arenexpr[\Psi]{\asubexpr[\extendsctx{\Lambda}{\Psi}]{v}{\sigma'}}{\key{\alpha}{\Lambda}{\Theta}{\hat{\Gamma}'}}}{\pi} && \explanation{Induction hypothesis} \\
      &= \arenexpr[\Psi]{\arenexpr[\extendsctx{\Lambda}{\Psi}]{\asubexpr[\extendsctx{\Lambda}{\Psi}]{v}{\sigma'}}{\pi}}{\key{\alpha}{\Lambda}{\Theta}{\extendsctx{\hat{\Gamma}'}{\mu}}} && \explanation{\cref{lem:key-pi-commute}} \\
      &= \arenexpr[\Psi]{\asubexpr[\extendsctx{\Lambda}{\Psi}]{v}{\weaken{\sigma'}}}{\key{\alpha}{\Lambda}{\Theta}{\extendsctx{\hat{\Gamma}'}{\mu}}}. && \explanation{\cref{eq:asubvar-weaken}}
    \end{align*}
  \item \case{} $\sigma = \lock{\sigma'}{\mu}$ (\inlinerulename{sf-arensub-lock}) \\
    We compute that
    \begin{align*}
      &\asubexpr[\extendsctx{\Theta}{\Psi}]{\arenexpr[\Psi]{v}{\key{\alpha}{\Lambda}{\Theta}{\lock{\hat{\Delta}}{\mu}}}}{\lock{\sigma'}{\mu}} \\
      &= \mathrlap{\asubexpr[\extendsctx{\Theta}{\Psi}]{\arenexpr[\Psi]{\arenexpr[\Psi]{v}{\addtele{\key{1_\mu}{\,\locksym_\mu}{\,\locksym_\mu}{\hat{\Delta}}}{\Lambda}}}{\key{\alpha}{\Lambda}{\Theta}{\lock{\hat{\Delta}}{\mu}}}}{\lock{\sigma'}{\mu}}} \\
      &&& \explanation{\cref{prop:key-unit-cell}} \\
      &= \asubexpr[\extendsctx{\extendsctx{\locksym_\mu}{\Theta}}{\Psi}]{\arenexpr[\Psi]{v}{\key{1_\mu \star \alpha}{\extendsctx{\,\locksym_\mu}{\Lambda}}{\extendsctx{\,\locksym_\mu}{\Theta}}{\hat{\Delta}}}}{\sigma'} && \explanation{\cref{prop:key-hor-comp-cell,eq:asubvar-lock}} \\
      &= \arenexpr[\Psi]{\asubexpr[\extendsctx{\extendsctx{\locksym_\mu}{\Lambda}}{\Psi}]{v}{\sigma'}}{\key{1_\mu \star \alpha}{\extendsctx{\,\locksym_\mu}{\Lambda}}{\extendsctx{\,\locksym_\mu}{\Theta}}{\hat{\Gamma}}} && \explanation{Induction hypothesis} \\
      &= \arenexpr[\Psi]{\asubexpr[\extendsctx{\extendsctx{\locksym_\mu}{\Lambda}}{\Psi}]{v}{\sigma'}}{\key{\alpha}{\Lambda}{\Theta}{\hat{\Gamma}}} && \explanation{\cref{prop:key-hor-comp-cell,prop:key-unit-cell}} \\
      &= \arenexpr[\Psi]{\asubexpr[\extendsctx{\Lambda}{\Psi}]{v}{\lock{\sigma'}{\mu}}}{\key{\alpha}{\Lambda}{\Theta}{\hat{\Gamma}}}. && \explanation{\cref{eq:asubvar-lock}}
    \end{align*}
  \item \case{} $\sigma = \key{\beta}{\Upsilon}{\Omega}{\hat{\Gamma}}$ (\inlinerulename{sf-arensub-key}) \\
    This case follows directly from \cref{prop:key-hor-comp-cell}.
  \item \case{} $\sigma = \sigma' . t$ with $\hat{\Delta} = \extendsctx{\hat{\Delta}'}{\mu}$ (\inlinerulename{sf-asub-extend}) \\
    We distinguish two cases for the variable $v$.
    \begin{itemize}
      \item \case{} $v = \vzero{\beta}$ with $\twocell{\beta}{\mu}{\locks{\extendsctx{\Lambda}{\Psi}}}$ \\
        Now we have
        \begin{align*}
          &\asubexpr[\extendsctx{\Theta}{\Psi}]{\arenexpr[\Psi]{\vzero{\beta}}{\key{\alpha}{\Lambda}{\Theta}{\extendsctx{\hat{\Delta}'}{\mu}}}}{\sigma' . t} \\
          &= \asubexpr[\extendsctx{\Theta}{\Psi}]{\vzero{(\alpha \star 1_{\locks{\Psi}}) \circ \beta}}{\sigma' . t} && \explanation{\cref{eq:arenvar-key,eq:twocell-vzero}} \\
          &= \arenexpr{t}{\key{(\alpha \star 1_{\locks{\Psi}}) \circ \beta}{\,\locksym_\mu}{\extendsctx{\Theta}{\Psi}}{\hat{\Gamma}}} && \explanation{\cref{eq:asubvar-extend-vzero}} \\
          &= \arenexpr{\arenexpr{t}{\key{\beta}{\,\locksym_\mu}{\extendsctx{\Lambda}{\Psi}}{\hat{\Gamma}}}}{\key{\alpha \star 1_{\locks{\Psi}}}{\extendsctx{\Lambda}{\Psi}}{\extendsctx{\Theta}{\Psi}}{\hat{\Gamma}}} && \explanation{\cref{prop:key-vert-comp-cell}} \\
          &= \arenexpr[\Psi]{\arenexpr{t}{\key{\beta}{\,\locksym_\mu}{\extendsctx{\Lambda}{\Psi}}{\hat{\Gamma}}}}{\key{\alpha}{\Lambda}{\Theta}{\hat{\Gamma}}} && \explanation{\cref{prop:key-hor-comp-cell,prop:key-unit-cell}} \\
          &= \arenexpr[\Psi]{\arenexpr[\extendsctx{\Lambda}{\Psi}]{\vzero{\beta}}{\sigma' . t}}{\key{\alpha}{\Lambda}{\Theta}{\hat{\Gamma}}}. && \explanation{\cref{eq:asubvar-extend-vzero}}
        \end{align*}
      \item \case{} $v = \suc{v'}$ \\
        In this case we get that
        \begin{align*}
          &\asubexpr[\extendsctx{\Theta}{\Psi}]{\arenexpr[\Psi]{\suc{v'}}{\key{\alpha}{\Lambda}{\Theta}{\extendsctx{\hat{\Delta}'}{\mu}}}}{\sigma' . t} \\
          &= \asubexpr[\extendsctx{\Theta}{\Psi}]{\suc{\arenexpr[\Psi]{v'}{\key{\alpha}{\Lambda}{\Theta}{\hat{\Delta}'}}}}{\sigma' . t} && \explanation{\cref{lem:key-ren-suc-var}} \\
          &= \asubexpr[\extendsctx{\Theta}{\Psi}]{\arenexpr[\Psi]{v'}{\key{\alpha}{\Lambda}{\Theta}{\hat{\Delta}'}}}{\sigma'} && \explanation{\cref{eq:asubvar-extend-vsuc}} \\
          &= \arenexpr[\Psi]{\asubexpr[\extendsctx{\Lambda}{\Psi}]{v'}{\sigma'}}{\key{\alpha}{\Lambda}{\Theta}{\hat{\Gamma}}} && \explanation{Induction hypothesis} \\
          &= \arenexpr[\Psi]{\asubexpr[\extendsctx{\Lambda}{\Psi}]{\suc{v'}}{\sigma' . t}}{\key{\alpha}{\Lambda}{\Theta}{\hat{\Gamma}}}. && \explanation{\cref{eq:asubvar-extend-vsuc}} \qedhere
        \end{align*}
    \end{itemize}
  \end{itemize}
\end{proof}

\subsection{Proof of \cref{thm:completeness}}

We can now prove a more general result that includes substitutions (and which can hence be proved by induction) and of which \cref{thm:completeness} is a consequence.
\begin{theorem}[Completeness]
  Given two $\upsigma$-equivalent WSMTT expressions $\sigmeqexpr{\hat{\Gamma}}{t}{s}{m}$, we have that $\translate{t} = \translate{s}$.
  Furthermore, given two $\upsigma$-equivalent WSMTT substitutions $\sigmeqsub{\sigma}{\tau}{\hat{\Gamma}}{\hat{\Delta}}{m}$, we have that $\sfsigmeqsub{\translate{\sigma}}{\translate{\tau}}$.
\end{theorem}
\begin{proof}
  We proceed by induction on a derivation of the $\upsigma$-equivalence judgement.
  To do this, we discuss all the rules from \cref{fig:sigma-equiv} and provide an outline of the argument for all the rules that are omitted in that figure.
  \begin{itemize}
  \item For the rules expressing that $\upsigma$-equivalence is an equivalence relation (e.g.\ \inlinerulename{wsmtt-eq-expr-refl}), we immediately get the desired result since equality of \subfree{} expressions and $\sfsigmeqsubsym$ are also equivalence relations.
  \item \case{} $\sigmeqsub{\sigma \circ \id}{\sigma}{\hat{\Gamma}}{\hat{\Delta}}{m}$ (\inlinerulename{wsmtt-eq-sub-id-right}) \\
    We have that $\translate{\sigma \circ \id} = \concat{\translate{\sigma}}{\translate{\id}}$ which is equal to $\translate{\sigma}$ since $\translate{\id}$ is the empty list of atomic substitutions (see the definition of $\translate{\_}$ in \cref{sec:translation-embedding}).
    This immediately proves that $\sfsigmeqsub{\translate{\sigma \circ \id}}{\translate{\sigma}}$.
    The other two category laws follow similarly from the monoid laws of list concatenation.
  \item \case{} $\sigmeqexpr{\hat{\Gamma}}{\subexprmtt{t}{\id}}{t}{m}$ (\inlinerulename{wsmtt-eq-expr-sub-id}) \\
    The definition of $\translate{\_}$ tells us that $\translate{\subexprmtt{t}{\id}} = \subexpr{\translate{t}}{\translate{\id}}$.
    Since $\translate{\id}$ is the empty list of atomic substitutions, we can directly see that this expression is equal to $\translate{t}$.
  \item \case{} $\sigmeqexpr{\hat{\Gamma}}{\subexprmtt{t}{\sigma \circ \tau}}{\subexprmtt{\subexprmtt{t}{\sigma}}{\tau}}{m}$ (\inlinerulename{wsmtt-eq-expr-sub-compose}) \\
    For the left-hand side we get that $\translate{\subexprmtt{t}{\sigma \circ \tau}} = \subexpr{\translate{t}}{\concat{\translate{\sigma}}{\translate{\tau}}}$, whereas for the right-hand side we have $\translate{\subexprmtt{\subexprmtt{t}{\sigma}}{\tau}} = \subexpr{\subexpr{\translate{t}}{\translate{\sigma}}}{\translate{\tau}}$.
    Since applying a regular substitution to an \subfree{} expression amounts to applying all constituent atomic substitutions, both expressions are equal.
  \item \case{} $\sigmeqexpr{\hat{\Gamma}}{\subexprmtt{t_1}{\sigma_1}}{\subexprmtt{t_2}{\sigma_2}}{m}$ (\inlinerulename{wsmtt-eq-expr-cong-sub}) \\
    The premises of this inference rule tell us that $\sigmeqexpr{\hat{\Delta}}{t_1}{t_2}{m}$ and $\sigmeqsub{\sigma_1}{\sigma_2}{\hat{\Gamma}}{\hat{\Delta}}{m}$.
    From the induction hypothesis it then follows that $\translate{t_1} = \translate{t_2}$ and $\sfsigmeqsub{\translate{\sigma_1}}{\translate{\sigma_2}}$.
    By the definition of $\sfsigmeqsubsym$ and $\translate{\_}$ we therefore have that $\translate{\subexprmtt{t_1}{\sigma_1}} = \subexpr{\translate{t_1}}{\translate{\sigma_1}} = \subexpr{\translate{t_2}}{\translate{\sigma_2}} = \translate{\subexprmtt{s_2}{\sigma_2}}$.
  \item \case{} $\sigmeqexpr{\hat{\Gamma}}{\lam{\mu}{t_1}}{\lam{\mu}{t_2}}{n}$ (\inlinerulename{wsmtt-eq-expr-cong-lam}) \\
    The premise of this inference rule give us that $\sigmeqexpr{\extendsctx{\hat{\Gamma}}{\mu}}{t_1}{t_2}{n}$.
    Therefore, by the induction hypothesis we have $\translate{t_1} = \translate{t_2}$ and hence $\translate{\lam{\mu}{t_1}} = \lam{\mu}{\translate{t_1}} = \lam{\mu}{\translate{t_2}} = \translate{\lam{\mu}{t_2}}$.
    The other congruence rules for expression constructors are proved similarly.
  \item \case{} $\sigmeqsub{\sigma_1 \circ \tau_1}{\sigma_2 \circ \tau_2}{\hat{\Gamma}}{\hat{\Xi}}{m}$ (\inlinerulename{wsmtt-eq-sub-cong-compose}) \\
    We know from the premises that $\sigmeqsub{\sigma_1}{\sigma_2}{\hat{\Delta}}{\hat{\Xi}}{m}$ and $\sigmeqsub{\tau_1}{\tau_2}{\hat{\Gamma}}{\hat{\Delta}}{m}$ and hence via the induction hypothesis $\sfsigmeqsub{\translate{\sigma_1}}{\translate{\sigma_2}}$ and $\sfsigmeqsub{\translate{\tau_1}}{\translate{\tau_2}}$.
    For an arbitrary expression $\sfexpr{\hat{\Xi}}{t}{m}$ we then have that
    \begin{align*}
      \subexpr{t}{\translate{\sigma_1 \circ \tau_1}}
      &= \subexpr{t}{\concat{\translate{\sigma_1}}{\translate{\tau_1}}} && \explanation{Definition of $\translate{\_}$} \\
      &= \subexpr{\subexpr{t}{\translate{\sigma_1}}}{\translate{\tau_1}} \\
      &= \subexpr{\subexpr{t}{\translate{\sigma_2}}}{\translate{\tau_1}} && \explanation{Definition of $\sfsigmeqsub{\sigma_1}{\sigma_2}$} \\
      &= \subexpr{\subexpr{t}{\translate{\sigma_2}}}{\translate{\tau_2}} && \explanation{Definition of $\sfsigmeqsub{\tau_1}{\tau_2}$} \\
      &= \subexpr{t}{\translate{\sigma_2 \circ \tau_2}},
    \end{align*}
    which proves that $\sfsigmeqsub{\translate{\sigma_1 \circ \tau_1}}{\translate{\sigma_2 \circ \tau_2}}$.
  \item \case{} $\sigmeqsub{\sigma_1 . t_1}{\sigma_2 . t_2}{\hat{\Gamma}}{\extendsctx{\hat{\Delta}}{\mu}}{n}$ (\inlinerulename{wsmtt-eq-sub-cong-extend}) \\
    The premises tell us that $\sigmeqsub{\sigma_1}{\sigma_2}{\hat{\Gamma}}{\hat{\Delta}}{n}$ and $\sigmeqexpr{\lock{\hat{\Gamma}}{\mu}}{t_1}{t_2}{m}$ and hence by the induction hypothesis $\sfsigmeqsub{\translate{\sigma_1}}{\translate{\sigma_2}}$ and $\translate{t_1} = \translate{t_2}$.
    \Cref{lem:lift-preserves-sigma-eq} then gives us that $\sfsigmeqsub{\lift{\translate{\sigma_1}}}{\lift{\translate{\sigma_2}}}$ from which it follows that
    \begin{align*}
      \translate{\sigma_1.t_1}
      &\mathrlap{{}={}}\hphantom{{}\sfsigmeqsubsym{}}\rensubcons{\lift{\translate{\sigma_1}}}{(\ida{}.\translate{t_1})} && \explanation{Definition of $\translate{\_}$} \\
      &\sfsigmeqsubsym{} \rensubcons{\lift{\translate{\sigma_2}}}{(\ida{}.\translate{t_1})} && \explanation{$\sfsigmeqsub{\lift{\translate{\sigma_1}}}{\lift{\translate{\sigma_2}}}$} \\
      &\mathrlap{{}={}}\hphantom{{}\sfsigmeqsubsym{}}\rensubcons{\lift{\translate{\sigma_2}}}{(\ida{}.\translate{t_2})} \\
      &\mathrlap{{}={}}\hphantom{{}\sfsigmeqsubsym{}}\translate{\sigma_2.t_2}.
    \end{align*}
  \item \case{} $\sigmeqsub{\lock{\sigma_1}{\mu}}{\lock{\sigma_2}{\mu}}{\lock{\hat{\Gamma}}{\mu}}{\lock{\hat{\Delta}}{\mu}}{m}$ (\inlinerulename{wsmtt-eq-sub-cong-lock}) \\
    From the premise we know that $\sigmeqsub{\sigma_1}{\sigma_2}{\hat{\Gamma}}{\hat{\Delta}}{n}$ and hence via induction $\sfsigmeqsub{\translate{\sigma_1}}{\translate{\sigma_2}}$.
    We can then use \cref{lem:lock-preserves-sigma-eq} to see that $\translate{\lock{\sigma_1}{\mu}} = \lock{\translate{\sigma_1}}{\mu}$ is observationally equivalent to $\translate{\lock{\sigma_2}{\mu}} = \lock{\translate{\sigma_2}}{\mu}$.
  \item \case{} $\sigmeqexpr{\hat{\Gamma}}{\subexprmtt{\left(\lam{\mu}{t}\right)}{\sigma}}{\lam{\mu}{\subexprmtt{t}{\lift{\sigma}}}}{n}$ (\inlinerulename{wsmtt-eq-expr-lam-sub}) \\
    Since all atomic \subfree{} substitutions can be pushed through $\lam{\mu}{\_}$ (see \cref{eq:push-atomic-lam}) and the lifting of a regular substitution consists of the lifted atomic substitutions, we have (also making use of the definition of $\translate{\_}$)
    \[
      \translate{\subexprmtt{\left(\lam{\mu}{t}\right)}{\sigma}}
      = \subexpr{\translate{\lam{\mu}{t}}}{\translate{\sigma}}
      = \subexpr{\lam{\mu}{\translate{t}}}{\translate{\sigma}}
      = \lam{\mu}{\subexpr{\translate{t}}{\lift{\translate{\sigma}}}}.
    \]
    On the other hand we know that $\translate{\lam{\mu}{\subexprmtt{t}{\lift{\sigma}}}} = \lam{\mu}{\subexpr{\translate{t}}{\translate{\lift{\sigma}}}}$.
    We conclude that both expressions are indeed equal because $\sfsigmeqsub{\translate{\lift{\sigma}}}{\lift{\translate{\sigma}}}$ by \cref{lem:lift-mtt-to-sfmtt}.
  \item \case{} $\sigmeqexpr{\hat{\Gamma}}{\subexprmtt{\left(\app{\mu}{f}{t}\right)}{\sigma}}{\app{\mu}{\subexprmtt{f}{\sigma}}{\subexprmtt{t}{\lock{\sigma}{\mu}}}}{n}$ (\inlinerulename{wsmtt-eq-expr-app-sub}) \\
    We have
    \begin{align*}
      &\translate{\subexprmtt{\left(\app{\mu}{f}{t}\right)}{\sigma}} \\
      &= \subexpr{\left(\app{\mu}{\translate{f}}{\translate{t}}\right)}{\translate{\sigma}} && \explanation{Definition of $\translate{\_}$} \\
      &= \app{\mu}{\subexpr{\translate{f}}{\translate{\sigma}}}{\subexpr{\translate{t}}{\lock{\translate{\sigma}}}} && \explanation{Repeated use of \cref{eq:push-atomic-app}}
    \end{align*}
    and
    \[
      \translate{\app{\mu}{\subexprmtt{f}{\sigma}}{\subexprmtt{t}{\lock{\sigma}{\mu}}}}
      = \app{\mu}{\subexpr{\translate{f}}{\translate{\sigma}}}{\subexpr{\translate{t}}{\translate{\lock{\sigma}{\mu}}}}.
    \]
    The result follows immediately since $\translate{\lock{\sigma}{\mu}} = \lock{\translate{\sigma}}{\mu}$.

    The cases for pushing substitutions through all other expression constructors are proved similarly.
  \item \case{} $\sigmeqsub{\sigma}{\emptysub{}}{\hat{\Gamma}}{\emptyctx{}}{m}$ (\inlinerulename{wsmtt-eq-sub-empty-unique}) \\
    We use \cref{prop:sfmtt-sigma-equiv-var-locktele} to prove that $\sfsigmeqsub{\translate{\sigma}}{\translate{!}}$.
    The condition of that proposition is immediately satisfied since there are no variables in the scoping context $\extendsctx{\emptyctx}{\Lambda}$ for any lock telescope $\Lambda$.
  \item \case{} $\sigmeqexpr{\lock{\hat{\Gamma}}{\mu}}{\subexprmtt{\vzero{}}{\lock{(\sigma . t)}{\mu}}}{t}{m}$ (\inlinerulename{wsmtt-eq-expr-extend-var}) \\
    We compute (using among others the definition of $\translate{\_}$)
    \begin{align*}
      &\translate{\subexprmtt{\vzero{}}{\lock{(\sigma . t)}{\mu}}} \\
      &= \subexpr{\translate{\vzero{}}}{\translate{\lock{(\sigma . t)}{\mu}}} \\
      &= \asubexpr{\subexpr{\vzero{1_\mu}}{\lock{\lift{\translate{\sigma}}}{\mu}}}{\lock{(\ida{} . \translate{t})}{\mu}} \\
      &= \asubexpr[\locksym_\mu]{\vzero{1_\mu}}{\ida{} . \translate{t}} && \explanation{Repeated application of \cref{lem:lift-asub-var}} \\
      &= \arenexpr{\translate{t}}{\key{1_\mu}{\locksym_\mu}{\locksym_\mu}{\hat{\Gamma}}} && \explanation{\cref{eq:asubvar-extend-vzero}} \\
      &= \translate{t}. && \explanation{\cref{prop:key-unit-cell}}
    \end{align*}
  \item \case{} $\sigmeqsub{\pi \circ (\sigma . t)}{\sigma}{\hat{\Gamma}}{\hat{\Delta}}{n}$ (\inlinerulename{wsmtt-eq-sub-extend-weaken}) \\
    We have that%
    \footnote{Note that $\rensubcons{\_}{\_}$ actually takes a regular substitution as left argument and an atomic substitution as right argument.
    We slightly abuse this notation by putting an atomic substitution to the left of the right-hand side of the following equation.}
    \[
      \translate{\pi \circ (\sigma . t)}
      = \concat{\translate{\pi}}{\translate{\sigma . t}}
      = \rensubcons{\pi}{\rensubcons{\translate{\sigma}^+}{(\ida{} . \translate{t})}}.
    \]
    Since $\asubexpr{s}{\pi} = \arenexpr{s}{\pi}$ (which is easy to prove using \cref{prop:mixseq-equiv-locktele}), we get that
    \begin{align*}
      \subexpr{s}{\translate{\pi \circ (\sigma . t)}}
      &= \asubexpr{\asubexpr{\asubexpr{s}{\pi}}{\lift{\translate{\sigma}}}}{\ida{} . \translate{t}} \\
      &= \asubexpr{\asubexpr{\asubexpr{s}{\translate{\sigma}}}{\pi}}{\ida{} . \translate{t}} && \explanation{\cref{lem:pi-lift-commute-sub}}
    \end{align*}
    for all expressions $s$.
    It therefore suffices to show that $\asubexpr{\asubexpr{s'}{\pi}}{\ida{} . \translate{t}} = s'$ for every $s'$.
    We do this using \cref{prop:mixseq-equiv-locktele}, so we take an arbitrary lock telescope $\Lambda : \Locktele{o}{n}$ and variable $\sfvar{\addtele{\hat{\Gamma}}{\Lambda}}{v}{o}$. We can then compute that
    \begin{align*}
      \asubexpr[\Lambda]{\asubexpr[\Lambda]{v}{\pi}}{\ida{} . \translate{t}}
      &= \asubexpr[\Lambda]{\arenexpr[\Lambda]{v}{\pi}}{\ida{} . \translate{t}} \\
      &= \asubexpr[\Lambda]{\suc{v}}{\ida{} . \translate{t}} \\
      &= \asubexpr[\Lambda]{v}{\ida{}} = v. && \explanation{\cref{eq:asubvar-extend-vsuc,eq:asubvar-id}}
    \end{align*}
  \item \case{} $\sigmeqsub{\sigma}{(\pi \circ \sigma) . (\subexprmtt{\vzero{}}{\lock{\sigma}{\mu}})}{\hat{\Gamma}}{\extendsctx{\hat{\Delta}}{\mu}}{n}$ (\inlinerulename{wsmtt-eq-sub-extend-eta}) \\
    We have that
    \begin{align*}
      \translate{(\pi \circ \sigma) . (\subexprmtt{\vzero{}}{\lock{\sigma}{\mu}})}
      &= \rensubcons{\lift{\translate{\pi \circ \sigma}}}{\left(\ida{} . \translate{\subexprmtt{\vzero{}}{\lock{\sigma}{\mu}}}\right)} \\
      &= \rensubcons{\lift{(\concat{\translate{\pi}}{\translate{\sigma}})}}{\left(\ida{} . \subexpr{\translate{\vzero{}}}{\translate{\lock{\sigma}{\mu}}}\right)} \\
      &= \rensubcons{\lift{\pi}}{\rensubcons{\lift{\translate{\sigma}}}{\left(\ida{} . \subexpr{\vzero{1_\mu}}{\lock{\translate{\sigma}}{\mu}}\right)}}.
    \end{align*}
    We now use \cref{prop:mixseq-equiv-locktele}, so for any lock telescope $\Lambda : \Locktele{o}{n}$ and variable $\sfvar{\addtele{\extendsctx{\hat{\Delta}}{\mu}}{\Lambda}}{v}{o}$, we need to show that
    \[
      \asubexpr[\Lambda]{\subexpr{\asubexpr[\Lambda]{v}{\lift{\pi}}}{\addtele{\lift{\translate{\sigma}}}{\Lambda}}}{\ida{} . \subexpr{\vzero{1_\mu}}{\lock{\translate{\sigma}}{\mu}}} = \subexpr{v}{\addtele{\translate{\sigma}}{\Lambda}}.
    \]
    We distinguish two cases for $v$.
    \begin{itemize}
    \item \case{} $v = \vzero{\alpha}$ with $\twocell{\alpha}{\mu}{\locks{\Lambda}}$. \\
      Then we get that
      \begin{align*}
        &\asubexpr[\Lambda]{\subexpr{\asubexpr[\Lambda]{\vzero{\alpha}}{\lift{\pi}}}{\addtele{\lift{\translate{\sigma}}}{\Lambda}}}{\ida{} . \subexpr{\vzero{1_\mu}}{\lock{\translate{\sigma}}{\mu}}} \\
        &= \asubexpr[\Lambda]{\subexpr{\vzero{\alpha}}{\addtele{\lift{\translate{\sigma}}}{\Lambda}}}{\ida{} . \subexpr{\vzero{1_\mu}}{\lock{\translate{\sigma}}{\mu}}} && \explanation{\cref{lem:lift-asub-var}} \\
        &= \asubexpr[\Lambda]{\vzero{\alpha}}{\ida{} . \subexpr{\vzero{1_\mu}}{\lock{\translate{\sigma}}{\mu}}} && \explanation{\cref{lem:lift-asub-var}, repeated} \\
        &= \arenexpr{\subexpr{\vzero{1_\mu}}{\lock{\translate{\sigma}}{\mu}}}{\key{\alpha}{\,\locksym_\mu}{\Lambda}{\hat{\Gamma}}} && \explanation{\cref{eq:asubvar-extend-vzero}} \\
        &= \subexpr{\arenexpr{\vzero{1_\mu}}{\key{\alpha}{\,\locksym_\mu}{\Lambda}{\extendsctx{\hat{\Delta}}{\mu}}}}{\addtele{\translate{\sigma}}{\Lambda}} && \explanation{\cref{prop:key-natural}} \\
        &= \subexpr{\vzero{\alpha}}{\addtele{\translate{\sigma}}{\Lambda}}.
      \end{align*}
    \item \case{} $v = \suc{v'}$ with $\sfvar{\addtele{\hat{\Delta}}{\Lambda}}{v'}{o}$ \\
      Now we can compute
      \begin{align*}
        &\asubexpr[\Lambda]{\subexpr{\asubexpr[\Lambda]{\suc{v'}}{\lift{\pi}}}{\addtele{\lift{\translate{\sigma}}}{\Lambda}}}{\ida{} . \subexpr{\vzero{1_\mu}}{\lock{\translate{\sigma}}{\mu}}} \\
        &= \asubexpr[\Lambda]{\subexpr{\arenexpr[\Lambda]{\asubexpr[\Lambda]{v'}{\pi}}{\pi}}{\addtele{\lift{\translate{\sigma}}}{\Lambda}}}{\ida{} . \subexpr{\vzero{1_\mu}}{\lock{\translate{\sigma}}{\mu}}} && \explanation{\cref{lem:lift-asub-var}} \\
        &= \asubexpr[\Lambda]{\subexpr{\arenexpr[\Lambda]{\suc{v'}}{\pi}}{\addtele{\lift{\translate{\sigma}}}{\Lambda}}}{\ida{} . \subexpr{\vzero{1_\mu}}{\lock{\translate{\sigma}}{\mu}}} \\
        &= \asubexpr[\Lambda]{\arenexpr[\Lambda]{\subexpr{\suc{v'}}{\addtele{\translate{\sigma}}{\Lambda}}}{\pi}}{\ida{} . \subexpr{\vzero{1_\mu}}{\lock{\translate{\sigma}}{\mu}}} && \explanation{\cref{lem:pi-lift-commute-sub}} \\
        &= \subexpr{\suc{v'}}{\addtele{\translate{\sigma}}{\Lambda}},
      \end{align*}
      where the last equation is proved as in the case of \inlinerulename{wsmtt-eq-sub-extend-weaken}.
    \end{itemize}
  \item \case{} $\sigmeqsub{\lock{\id}{\mu}}{\id}{\lock{\hat{\Gamma}}{\mu}}{\lock{\hat{\Gamma}}{\mu}}{m}$ (\inlinerulename{wsmtt-eq-sub-lock-id}) \\
    The translations of both sides of this equivalence are the empty sequence of atomic \subfree{} substitutions, so this case is trivial.
  \item \case{} $\sigmeqsub{\lock{(\sigma \circ \tau)}{\mu}}{(\lock{\sigma}{\mu}) \circ (\lock{\tau}{\mu})}{\lock{\hat{\Gamma}}{\mu}}{\lock{\hat{\Xi}}{\mu}}{m}$ (\inlinerulename{wsmtt-eq-sub-lock-compose}) \\
    Again this case is trivial since a lock is applied to every atomic substitution in a sequence and hence it distributes over sequence concatenation.
  \item \case{} $\sigmeqsub{\key{\alpha}{\Lambda}{\Theta}{\hat{\Delta}} \circ (\extendsctx{\sigma}{\Theta})}{(\extendsctx{\sigma}{\Lambda}) \circ \key{\alpha}{\Lambda}{\Theta}{\hat{\Gamma}}}{\extendsctx{\hat{\Gamma}}{\Theta}}{\extendsctx{\hat{\Delta}}{\Lambda}}{n}$ (\inlinerulename{wsmtt-eq-sub-key-natural}) \\
    This is a direct consequence of \cref{prop:key-natural}.
  \item \case{} $\sigmeqsub{\key{1_{\locks{\Lambda}}}{\Lambda}{\Lambda}{\hat{\Gamma}}}{\id}{\extendsctx{\hat{\Gamma}}{\Lambda}}{\extendsctx{\hat{\Gamma}}{\Lambda}}{n}$ (\inlinerulename{wsmtt-eq-sub-key-unit}) \\
    Applying an \subfree{} key substitution is exactly the same as applying the corresponding key renaming (which can be easily proved using \cref{prop:mixseq-equiv-locktele}), so this case follows immediately from \cref{prop:key-unit-cell}.
  \item \case{} $\sigmeqsub{\key{\beta \circ \alpha}{\Lambda}{\Psi}{\hat{\Gamma}}}{\key{\alpha}{\Lambda}{\Theta}{\hat{\Gamma}} \circ \key{\beta}{\Theta}{\Psi}{\hat{\Gamma}}}{\extendsctx{\hat{\Gamma}}{\Psi}}{\extendsctx{\hat{\Gamma}}{\Lambda}}{n}$ (\inlinerulename{wsmtt-eq-sub-key-compose-vertical}) \\
    In the same way, the result in this case is proved by \cref{prop:key-vert-comp-cell}.
  \item \case{} $\scriptstyle \sigmeqsub{\key{\beta \star \alpha}{\extendsctx{\Lambda_1}{\Theta_1}}{\extendsctx{\Lambda_2}{\Theta_2}}{\hat{\Gamma}}}{(\extendsctx{\key{\beta}{\Lambda_1}{\Lambda_2}{\hat{\Gamma}}}{\Theta_1}) \circ \key{\alpha}{\Theta_1}{\Theta_2}{\extendsctx{\hat{\Gamma}}{\Lambda_2}}}{\extendsctx{\extendsctx{\hat{\Gamma}}{\Lambda_2}}{\Theta_2}}{\extendsctx{\extendsctx{\hat{\Gamma}}{\Lambda_1}}{\Theta_1}}{o}$ (\inlinerulename{wsmtt-eq-sub-key-compose-horizontal}) \\
    This is a direct consequence of \cref{prop:key-hor-comp-cell}. \qedhere
  \end{itemize}
\end{proof}

\section{Soundness}
\label{sec:soundness}

We want to prove the soundness of our substitution algorithm with respect to the notion of $\upsigma$-equivalence introduced in \cref{fig:sigma-equiv}.
In other words, whenever we compute all substitutions away in a WSMTT expression $t$, the result should be $\upsigma$-equivalent to the expression $t$ we started from.
\begin{theorem}
  \label{thm:soundness}
  Let $\mttexpr{\hat{\Gamma}}{t}{m}$ be a WSMTT expression. Then we have that $\sigmeqexpr{\hat{\Gamma}}{\embed{\translate{t}}}{t}{m}$.
\end{theorem}
The proof of this theorem appears at the end of this section.

\subsection{Embedding of SFMTT into WSMTT}

Note that in \cref{sec:translation-embedding} we first defined an embedding of \subfree{} expressions to WSMTT and then an embedding for atomic and regular rensubs.
This is unlike the translation function from WSMTT to SFMTT, which is defined mutually recursively for expressions and substitutions.
The reason for this is that SFMTT substitutions do not occur in the syntax of \subfree{} expressions.
However, the proof of \cref{thm:soundness} is easier to formulate if we do have an embedding of rensubs at our disposal.
In particular, the core result for proving soundness will be \cref{prop:embed-subst-term}.

In this section on the soundness proof, we will extensively use the fact that composition of WSMTT substitutions is associative and that $\id$ is its unit, all up to $\upsigma$-equivalence.
Moreover, congruence rules with respect to WSMTT $\upsigma$-equivalence will also regularly be used.
We will not explicitly mention the use of any of these rules from \cref{fig:sigma-equiv}.

\begin{example}[Embedding does not preserve observational equivalence] \label{ex:embedding-does-not-preserve}
  Given that we have introduced the notion of observational equivalence for SFMTT substitutions in \cref{sec:observ-equiv}, it is natural to ask whether $\sfsigmeqsub \sigma \tau$ implies $\embed{\sigma} \sigmeq \embed{\tau}$.
  The answer is no, and we can give a counterexample similar to \cref{ex:subst-existence}.
  Again, let the mode theory be the walking arrow.
  Let $\hat \Gamma = (\lock{\emptyctx}{\mu})$ and $\hat \Delta = (\lock{\extendsctx{\emptyctx}{\unitMod}}{\mu})$.
  As argued in \cref{ex:subst-existence}, all substitutions to $\hat \Delta$ are observationally equivalent.
  However, the embeddings of $\sfasub{(\lock{\extendsctx{!}{\true}}{\mu}), (\lock{\extendsctx{!}{\false}}{\mu})}{\hat \Gamma}{\hat \Delta}{m}$ are not $\upsigma$-equivalent.
\end{example}

\begin{lemma}
  \label{lem:embed-lift}
  For an \subfree{} renaming or substitution $\sfrensub{\sigma}{\hat{\Gamma}}{\hat{\Delta}}{m}$ we have that $\sigmeqsub{\embed{\lift{\sigma}}}{\lift{\embed{\sigma}}}{\extendsctx{\hat{\Gamma}}{\mu}}{\extendsctx{\hat{\Delta}}{\mu}}{m}$.
\end{lemma}
\begin{proof}
  Since $\lift{\id} \sigmeq \id$ and $\lift{(\sigma \circ \tau)} \sigmeq \lift{\sigma} \circ \lift{\tau}$ (which can be proved using \inlinerulename{wsmtt-eq-sub-extend-weaken}, \inlinerulename{wsmtt-eq-sub-extend-eta} and \inlinerulename{wsmtt-eq-expr-extend-var}), it suffices to prove this for an atomic rensub $\sigma$.
  Then we have that
  \begin{align*}
    &\embed{\lift{\sigma}} \\
    &= \embed{\weaken{\sigma} . \vzero{1_\mu}} && \explanation{SFMTT definition of $^+$, \eqref{eq:sfmtt-arensub-lift}} \\
    &= (\embed{\sigma} \circ \pi) . \left(\subexprmtt{\vzero{}}{\key{1_\mu}{\,\locksym_\mu}{\,\locksym_\mu}{\extendsctx{\hat{\Gamma}}{\mu}}}\right) && \explanation{Definition of $\embed{\_}$} \\
    &\sigmeq (\embed{\sigma} \circ \pi) . \vzero{} && \explanation{\inlinerulename{wsmtt-eq-sub-key-unit}} \\
    &= \lift{\embed{\sigma}}. && \explanation{WSMTT definition of $^+$, \eqref{eq:wsmtt-sub-lift}} \qedhere
  \end{align*}
\end{proof}

\subsection{Embedding and Renaming/Substitution}

The core property for proving the soundness theorem is \cref{prop:embed-subst-term}, which states that $\embed{\subexpr{t}{\sigma}} \sigmeq \subexprmtt{\embed{t}}{\embed{\sigma}}$ for every $t$ and $\sigma$.
In order to prove such a result, we will adopt a similar technique as in \cref{sec:observ-equiv} for proving observational equivalence of SFMTT substitutions.
First we show that it is sufficient to prove the result for variables after adding an arbitrary scoping telescope $\Phi$ to $\sigma$ (\cref{lem:embed-ren-term-tele}).
Then we prove that actually the scoping telescope $\Phi$ only needs to be a lock telescope (\cref{lem:embed-ren-term-lock-tele,lem:embed-sub-term-lock-tele}).
\begin{lemma}
  \label{lem:embed-ren-term-tele}
  Let $\sfarensub{\sigma}{\hat{\Gamma}}{\hat{\Delta}}{m}$ be an atomic \subfree{} rensub and assume that $\sigmeqexpr{\addtele{\hat{\Gamma}}{\Phi}}{\embed{\arensubexpr{v}{\addtele{\sigma}{\Phi}}}}{\subexprmtt{\embed{v}}{\embed{\addtele{\sigma}{\Phi}}}}{n}$ for any scoping telescope $\Phi : \Tele{n}{m}$ and variable $\sfvar{\addtele{\hat{\Delta}}{\Phi}}{v}{n}$.
  Then we have that $\sigmeqexpr{\hat{\Gamma}}{\embed{\arensubexpr{t}{\sigma}}}{\subexprmtt{\embed{t}}{\embed{\sigma}}}{m}$ for all expressions $\sfexpr{\hat{\Delta}}{t}{m}$.
\end{lemma}
\begin{proof}
  We will prove the more general result that $\sigmeqexpr{\addtele{\hat{\Gamma}}{\Phi}}{\embed{\arensubexpr{t}{\addtele{\sigma}{\Phi}}}}{\subexprmtt{\embed{t}}{\embed{\addtele{\sigma}{\Phi}}}}{n}$ for all scoping telescopes $\Phi : \Tele{n}{m}$ and expressions $\sfexpr{\addtele{\hat{\Delta}}{\Phi}}{t}{n}$.
  This proof proceeds by induction on $t$.
  We only show the cases for variables, lambda abstraction and the modal term constructor.
  The other cases can be proved similarly.
  \begin{itemize}
  \item \case{} $\sfexpr{\addtele{\hat{\Delta}}{\Phi}}{v}{n}$ (\inlinerulename{sf-expr-var}) \\
    The result is exactly what we assumed in the lemma.
  \item \case{} $\sfexpr{\addtele{\hat{\Delta}}{\Phi}}{\lam{\mu}{t}}{n}$ (\inlinerulename{sf-expr-lam}) \\
    We have that
    \begin{align*}
      &\embed{\arensubexpr{\left(\lam{\mu}{t}\right)}{\addtele{\sigma}{\Phi}}} \\
      &= \embed{\lam{\mu}{\arensubexpr{t}{\lift{(\addtele{\sigma}{\Phi})}}}} && \explanation{\cref{eq:push-atomic-lam}} \\
      &= \lam{\mu}{\embed{\arensubexpr{t}{\lift{(\addtele{\sigma}{\Phi})}}}} && \explanation{Definition of $\embed{\_}$} \\
      &\sigmeq \lam{\mu}{\subexprmtt{\embed{t}}{\embed{\lift{(\addtele{\sigma}{\Phi})}}}} && \explanation{Induction hypothesis} \\
      &\sigmeq \lam{\mu}{\subexprmtt{\embed{t}}{\lift{\big(\embed{\addtele{\sigma}{\Phi}}\big)}}} && \explanation{\cref{lem:embed-lift}} \\
      &\sigmeq \subexprmtt{\left(\lam{\mu}{\embed{t}}\right)}{\embed{\addtele{\sigma}{\Phi}}} && \explanation{\inlinerulename{wsmtt-eq-expr-lam-sub}} \\
      &= \subexprmtt{\embed{\lam{\mu}{t}}}{\embed{\addtele{\sigma}{\Phi}}}. && \explanation{Definition of $\embed{\_}$}
    \end{align*}
    Note that we can indeed apply the induction hypothesis where it is indicated since $\lift{(\addtele{\sigma}{\Phi})} = \addtele{\sigma}{(\extendsctx{\Phi}{\mu})}$.
  \item \case{} $\sfexpr{\addtele{\hat{\Delta}}{\Phi}}{\modtm{\mu}{t}}{n}$ (\inlinerulename{sf-expr-mod-tm}) \\
    Now we can compute that
    \begin{align*}
      &\embed{\arensubexpr{\left(\modtm{\mu}{t}\right)}{\addtele{\sigma}{\Phi}}} \\
      &= \embed{\modtm{\mu}{\arensubexpr{t}{\lock{(\addtele{\sigma}{\Phi})}{\mu}}}} && \explanation{\cref{eq:push-atomic-modtm}} \\
      &= \modtm{\mu}{\embed{\arensubexpr{t}{\lock{(\addtele{\sigma}{\Phi})}{\mu}}}} && \explanation{Definition of $\embed{\_}$} \\
      &\sigmeq \modtm{\mu}{\subexprmtt{\embed{t}}{\embed{\lock{(\addtele{\sigma}{\Phi})}{\mu}}}} && \explanation{Induction hypothesis} \\
      &= \modtm{\mu}{\subexprmtt{\embed{t}}{\lock{\big(\embed{\addtele{\sigma}{\Phi}}\big)}{\mu}}} && \explanation{Definition of $\embed{\_}$} \\
      &\sigmeq \subexprmtt{\big(\modtm{\mu}{\embed{t}}\big)}{\embed{\addtele{\sigma}{\Phi}}}. && \explanation{\inlinerulenamestyle{wsmtt-eq-expr-mod-tm-sub}} \\
      &= \subexprmtt{\embed{\modtm{\mu}{t}}}{\embed{\addtele{\sigma}{\Phi}}} && \explanation{Definition of $\embed{\_}$}
    \end{align*}
    Again we can apply the induction hypothesis because $\lock{(\addtele{\sigma}{\Phi})}{\mu} = \addtele{\sigma}{(\lock{\Phi}{\mu})}$.
    The rule \inlinerulenamestyle{wsmtt-eq-expr-mod-tm-sub} is not included in \cref{fig:sigma-equiv}, but it is similar to \inlinerulename{wsmtt-eq-expr-lam-sub}. \qedhere
  \end{itemize}
\end{proof}

\begin{lemma}
  \label{lem:embed-ren-term-lock-tele}
  Let $\sfaren{\sigma}{\hat{\Gamma}}{\hat{\Delta}}{m}$ be an atomic \subfree{} renaming and assume that $\sigmeqexpr{\addtele{\hat{\Gamma}}{\Lambda}}{\embed{\arenexpr{v}{\addtele{\sigma}{\Lambda}}}}{\subexprmtt{\embed{v}}{\embed{\addtele{\sigma}{\Lambda}}}}{n}$ for every lock telescope $\Lambda : \Tele{n}{m}$ and variable $\sfvar{\addtele{\hat{\Delta}}{\Lambda}}{v}{n}$.
  Then we have that $\sigmeqexpr{\hat{\Gamma}}{\embed{\arenexpr{t}{\sigma}}}{\subexprmtt{\embed{t}}{\embed{\sigma}}}{m}$ for all expressions $\sfexpr{\hat{\Delta}}{t}{m}$.
\end{lemma}
\begin{proof}
  By making use of \cref{lem:embed-ren-term-tele}, we have to show that $\sigmeqexpr{\addtele{\hat{\Gamma}}{\Phi}}{\embed{\arenexpr{v}{\addtele{\sigma}{\Phi}}}}{\subexprmtt{\embed{v}}{\embed{\addtele{\sigma}{\Phi}}}}{n}$ for all $\Phi : \Tele{n}{m}$ and $\sfvar{\addtele{\hat{\Delta}}{\Phi}}{v}{n}$.
  We do this by induction on the number of variables in $\Phi$.
  \begin{itemize}
  \item \case{} $\Phi = \Lambda$, so $\Phi$ has no variables \\
    The result is exactly what we assume in this lemma.
  \item \case{} $\Phi = \extendsctx{\extendsctx{\Phi'}{\mu}}{\Lambda}$ \\
    Now we distinguish between two cases for the variable $v$.
    \begin{itemize}
    \item \case{} $v = \vzero{\alpha}$ with $\twocell{\alpha}{\mu}{\locks{\Lambda}}$ \\
      For the left-hand side, we have that
      \begin{align*}
        &\embed{\arenexpr{\vzero{\alpha}}{\addtele{\extendsctx{\addtele{\sigma}{\Phi'}}{\mu}}{\Lambda}}} \\
        &= \embed{\arenexpr[\Lambda]{\vzero{\alpha}}{\lift{(\addtele{\sigma}{\Phi'})}}} \\
        &= \embed{\vzero{\alpha}} && \explanation{\cref{lem:lift-aren-var}} \\
        &= \subexprmtt{\vzero{}}{\key{\alpha}{\,\locksym_\mu}{\Lambda}{\addtele{\hat{\Gamma}}{\extendsctx{\Phi'}{\mu}}}}. && \explanation{Definition of $\embed{\_}$}
      \end{align*}
      On the other hand, we have
      \begin{align*}
        &\subexprmtt{\embed{\vzero{\alpha}}}{\embed{\addtele{\extendsctx{\addtele{\sigma}{\Phi'}}{\mu}}{\Lambda}}} \\
        &= \subexprmtt{\subexprmtt{\vzero{}}{\key{\alpha}{\,\locksym_\mu}{\Lambda}{\addtele{\hat{\Delta}}{\extendsctx{\Phi'}{\mu}}}}}{\embed{\addtele{\lift{(\addtele{\sigma}{\Phi'})}}{\Lambda}}} && \explanation{Definition of $\embed{\vzero{\alpha}}$} \\
        &\sigmeq \subexprmtt{\subexprmtt{\vzero{}}{\key{\alpha}{\,\locksym_\mu}{\Lambda}{\addtele{\hat{\Delta}}{\extendsctx{\Phi'}{\mu}}}}}{\addtele{\lift{(\embed{\addtele{\sigma}{\Phi'}})}}{\Lambda}} && \explanation{\cref{lem:embed-lift}} \\
        &\sigmeq \subexprmtt{\subexprmtt{\vzero{}}{\lock{\lift{(\embed{\addtele{\sigma}{\Phi'}})}}{\mu}}}{\key{\alpha}{\locksym_\mu}{\Lambda}{\addtele{\hat{\Gamma}}{\extendsctx{\Phi'}{\mu}}}} && \explanation{\inlinerulename{wsmtt-eq-sub-key-natural}} \\
        &\sigmeq \subexprmtt{\vzero{}}{\key{\alpha}{\locksym_\mu}{\Lambda}{\addtele{\hat{\Gamma}}{\extendsctx{\Phi'}{\mu}}}}. && \explanation{\inlinerulename{wsmtt-eq-expr-extend-var}}
      \end{align*}
    \item \case{} $v = \suc{v'}$ with $\sfvar{\addtele{\addtele{\hat{\Delta}}{\Phi'}}{\Lambda}}{v'}{n}$ \\
      Now we see that
      \begin{align*}
        &\embed{\arenexpr{\suc{v'}}{\addtele{\extendsctx{\addtele{\sigma}{\Phi'}}{\mu}}{\Lambda}}} \\
        &= \embed{\arenexpr[\Lambda]{\suc{v'}}{\lift{\left(\addtele{\sigma}{\Phi'}\right)}}} \\
        &= \embed{\suc{\arenexpr[\Lambda]{v'}{\addtele{\sigma}{\Phi'}}}} && \explanation{\cref{lem:lift-aren-var}} \\
        &= \subexprmtt{\embed{\arenexpr{v'}{\addtele{\addtele{\sigma}{\Phi'}}{\Lambda}}}}{\addtele{\pi}{\Lambda}} && \explanation{Definition of $\embed{\_}$} \\
        &\sigmeq \subexprmtt{\subexprmtt{\embed{v'}}{\embed{\addtele{\addtele{\sigma}{\Phi'}}{\Lambda}}}}{\addtele{\pi}{\Lambda}}. && \explanation{Induction hypothesis}
      \end{align*}
      Furthermore, we have
      \begin{align*}
        &\subexprmtt{\embed{\suc{v'}}}{\embed{\addtele{\sigma}{\addtele{\extendsctx{\Phi'}{\mu}}{\Lambda}}}} \\
        &= \subexprmtt{\subexprmtt{\embed{v'}}{\addtele{\pi}{\Lambda}}}{\addtele{\lift{(\embed{\addtele{\sigma}{\Phi'}})}}{\Lambda}} && \explanation{Definition of $\embed{\suc{v'}}$} \\
        &\sigmeq \subexprmtt{\embed{v'}}{\addtele{(\pi \circ \lift{(\embed{\addtele{\sigma}{\Phi'}})})}{\Lambda}} && \explanation{\textasteriskcentered} \\
        &\sigmeq \subexprmtt{\embed{v'}}{\addtele{(\embed{\addtele{\sigma}{\Phi'}} \circ \pi)}{\Lambda}} && \mathrlap{\explanation{\inlinerulename{wsmtt-eq-sub-extend-weaken}}} \\
        &\sigmeq \subexprmtt{\subexprmtt{\embed{v'}}{\embed{\addtele{\addtele{\sigma}{\Phi'}}{\Lambda}}}}{\addtele{\pi}{\Lambda}}. && \explanation{\textasteriskcentered} \qedhere
      \end{align*}
      The steps marked with (\textasteriskcentered) make use of \inlinerulename{wsmtt-eq-expr-sub-compose} and \inlinerulename{wsmtt-eq-sub-lock-compose}.
    \end{itemize}
  \end{itemize}
\end{proof}

\begin{lemma}
  \label{lem:embed-pi-ren}
  Up to $\upsigma$-equivalence, applying a weakening renaming commutes with the embedding function.
  In other words, for every lock telescope $\Lambda : \Locktele{n}{m}$ and \subfree{} expression $\sfexpr{\addtele{\hat{\Gamma}}{\Lambda}}{t}{n}$, we have that $\sigmeqexpr{\addtele{\extendsctx{\hat{\Gamma}}{\mu}}{\Lambda}}{\embed{\arenexpr{t}{\addtele{\pi}{\Lambda}}}}{\subexprmtt{\embed{t}}{\addtele{\pi}{\Lambda}} \sigmeq \subexprmtt{\embed{t}}{\embed{\addtele{\pi}{\Lambda}}}}{n}$.
\end{lemma}
\begin{proof}
  We first prove the second $\upsigma$-equivalence by computing the following.
  \begin{align*}
    \embed{\addtele{\pi}{\Lambda}}
    &= \addtele{\embed{\pi}}{\Lambda} = \addtele{\embed{\weaken{\ida}}}{\Lambda} \\
    &= \addtele{(\embed{\ida} \circ \pi)}{\Lambda} = \addtele{(\id \circ \pi)}{\Lambda} \\
    &\sigmeq \addtele{\pi}{\Lambda} && \explanation{\inlinerulenamestyle{wsmtt-eq-sub-id-left}}
  \end{align*}
  The rule \inlinerulenamestyle{wsmtt-eq-sub-id-left} is not included in \cref{fig:sigma-equiv}, but it is similar to \inlinerulename{wsmtt-eq-sub-id-right}.

  To prove the other $\upsigma$-equivalence we use \cref{lem:embed-ren-term-lock-tele}, so we take an arbitrary lock telescope $\Theta : \Locktele{o}{n}$ and a variable $\sfvar{\addtele{\addtele{\hat{\Gamma}}{\Lambda}}{\Theta}}{v}{o}$ and then show that $\embed{\arenexpr{v}{\addtele{\addtele{\pi}{\Lambda}}{\Theta}}} = \subexprmtt{\embed{v}}{\embed{\addtele{\addtele{\pi}{\Lambda}}{\Theta}}}$.
  This can be easily proved by expanding the definition of $\embed{\_}$ as follows.
  \begin{align*}
    \embed{\arenexpr[\addtele{\Lambda}{\Theta}]{v}{\pi}}
    &= \embed{\suc{v}} \\
    &= \subexprmtt{\embed{v}}{\addtele{\addtele{\pi}{\Lambda}}{\Theta}} \\
    &\sigmeq \subexprmtt{\embed{v}}{\embed{\addtele{\addtele{\pi}{\Lambda}}{\Theta}}} \qedhere
  \end{align*}
\end{proof}

Using \cref{lem:embed-pi-ren}, we can now prove a result similar to \cref{lem:embed-ren-term-lock-tele} but for substitutions instead of renamings.
\begin{lemma}
  \label{lem:embed-sub-term-lock-tele}
  Let $\sfasub{\sigma}{\hat{\Gamma}}{\hat{\Delta}}{m}$ be an atomic \subfree{} substitution and assume that $\sigmeqexpr{\addtele{\hat{\Gamma}}{\Lambda}}{\embed{\asubexpr{v}{\addtele{\sigma}{\Lambda}}}}{\subexprmtt{\embed{v}}{\embed{\addtele{\sigma}{\Lambda}}}}{n}$ for every lock telescope $\Lambda : \Tele{n}{m}$ and variable $\sfvar{\addtele{\hat{\Delta}}{\Lambda}}{v}{n}$.
  Then we have that $\sigmeqexpr{\hat{\Gamma}}{\embed{\asubexpr{t}{\sigma}}}{\subexprmtt{\embed{t}}{\embed{\sigma}}}{m}$ for all expressions $\sfexpr{\hat{\Delta}}{t}{m}$.
\end{lemma}
\begin{proof}
  The proof is very similar to that of \cref{lem:embed-ren-term-lock-tele}.
  Again we make use of \cref{lem:embed-ren-term-tele}, so we take an arbitrary $\Phi : \Tele{n}{m}$ and $\sfvar{\addtele{\hat{\Delta}}{\Phi}}{v}{n}$ and show that $\sigmeqexpr{\addtele{\hat{\Gamma}}{\Phi}}{\embed{\asubexpr{v}{\addtele{\sigma}{\Phi}}}}{\subexprmtt{\embed{v}}{\embed{\addtele{\sigma}{\Phi}}}}{n}$ by induction on the number of variables in $\Phi$.
  \begin{itemize}
  \item \case{} $\Phi = \Lambda$, so $\Phi$ contains no variables \\
    The result we need to show is exactly the assumption in the lemma.
  \item \case{} $\Phi = \extendsctx{\extendsctx{\Phi'}{\mu}}{\Lambda}$ \\
    We proceed by case distinction for the variable $v$.
    \begin{itemize}
    \item \case{} $v = \vzero{\alpha}$ with $\twocell{\alpha}{\mu}{\locks{\Lambda}}$ \\
      For the left-hand side, we get
      \begin{align*}
        &\embed{\asubexpr{\vzero{\alpha}}{\addtele{\extendsctx{\addtele{\sigma}{\Phi'}}{\mu}}{\Lambda}}} \\
        &= \embed{\asubexpr[\Lambda]{\vzero{\alpha}}{\lift{(\addtele{\sigma}{\Phi'})}}} \\
        &= \embed{\vzero{\alpha}} && \explanation{\cref{lem:lift-asub-var}} \\
        &= \subexprmtt{\vzero{}}{\key{\alpha}{\,\locksym{\mu}}{\Lambda}{\extendsctx{\addtele{\hat{\Gamma}}{\Phi'}}{\mu}}}. && \explanation{Definition of $\embed{\vzero{\alpha}}$}
      \end{align*}
      The right-hand side can be computed in exactly the same way as in the proof of \cref{lem:embed-ren-term-lock-tele}.
    \item \case{} $v = \suc{v'}$ with $\sfvar{\addtele{\addtele{\hat{\Delta}}{\Phi'}}{\Lambda}}{v'}{n}$ \\
      The left-hand side now becomes
      \begin{align*}
        &\embed{\asubexpr{\suc{v'}}{\addtele{\extendsctx{\addtele{\sigma}{\Phi'}}{\mu}}{\Lambda}}} \\
        &= \embed{\asubexpr[\Lambda]{\suc{v'}}{\lift{(\addtele{\sigma}{\Phi'})}}} \\
        &= \embed{\arenexpr[\Lambda]{\asubexpr[\Lambda]{v'}{\addtele{\sigma}{\Phi'}}}{\pi}} && \explanation{\cref{lem:lift-asub-var}} \\
        &\sigmeq \subexprmtt{\embed{\asubexpr[\Lambda]{v'}{\addtele{\sigma}{\Phi'}}}}{\addtele{\pi}{\Lambda}} && \explanation{\cref{lem:embed-pi-ren}} \\
        &\sigmeq \subexprmtt{\subexprmtt{\embed{v'}}{\embed{\addtele{\addtele{\sigma}{\Phi'}}{\Lambda}}}}{\addtele{\pi}{\Lambda}}. && \explanation{Induction hypothesis}
      \end{align*}
      Again, the right-hand side can be computed in entirely the same way as in the proof of \cref{lem:embed-ren-term-lock-tele}. \qedhere
    \end{itemize}
  \end{itemize}
\end{proof}

\begin{lemma}
  \label{lem:embed-key-ren}
  Given lock telescopes $\Lambda, \Theta : \Locktele{n}{m}$ and a 2-cell $\twocell{\alpha}{\locks{\Lambda}}{\locks{\Theta}}$, we have that
  \[
    \sigmeqexpr{\addtele{\addtele{\hat{\Gamma}}{\Theta}}{\Psi}}{\embed{\arenexpr{t}{\addtele{\key{\alpha}{\Lambda}{\Theta}{\hat{\Gamma}}}{\Psi}}}}{\subexprmtt{\embed{t}}{\embed{\addtele{\key{\alpha}{\Lambda}{\Theta}{\hat{\Gamma}}}{\Psi}}}}{o}
  \]
  for all lock telescopes $\Psi : \Locktele{o}{n}$ and expressions $\sfexpr{\addtele{\addtele{\hat{\Gamma}}{\Lambda}}{\Psi}}{t}{o}$.
\end{lemma}
\begin{proof}
  We again use \cref{lem:embed-ren-term-lock-tele}, so we take a lock telescope $\Upsilon : \Locktele{p}{o}$ and a variable $\sfvar{\addtele{\addtele{\addtele{\hat{\Gamma}}{\Lambda}}{\Psi}}{\Upsilon}}{v}{p}$.
  We then distinguish between two cases for $v$.
  \begin{itemize}
  \item \case{} $v = \vzero{\beta}$ with $\hat{\Gamma} = \addtele{\extendsctx{\hat{\Gamma}'}{\mu}}{\Omega}$ and $\twocell{\beta}{\mu}{\locks{\addtele{\addtele{\addtele{\Omega}{\Lambda}}{\Psi}}{\Upsilon}}}$ \\
    Now we can compute that
    \begin{align*}
      &\embed{\arenexpr[\addtele{\Psi}{\Upsilon}]{\vzero{\beta}}{\key{\alpha}{\Lambda}{\Theta}{\addtele{\extendsctx{\hat{\Gamma}'}{\mu}}{\Omega}}}} \\
      &= \embed{\vzero{(1_\Omega \star (\alpha \star 1_{(\addtele{\Psi}{\Upsilon})})) \circ \beta}} && \explanation{\cref{eq:arenvar-key,eq:twocell-vzero}} \\
      &= \subexprmtt{\vzero{}}{\key{(1_\Omega \star (\alpha \star 1_{(\addtele{\Psi}{\Upsilon})})) \circ \beta}{\,\locksym_\mu}{\addtele{\addtele{\addtele{\Omega}{\Theta}}{\Psi}}{\Upsilon}}{\extendsctx{\hat{\Gamma}'}{\mu}}} && \explanation{Definition of $\embed{\_}$} \\
      &\sigmeq \subexprmtt{\subexprmtt{\vzero{}}{\key{\beta}{\,\locksym_\mu}{\addtele{\addtele{\addtele{\Omega}{\Lambda}}{\Psi}}{\Upsilon}}{\extendsctx{\hat{\Gamma}'}{\mu}}}}{\addtele{\key{1_\Omega}{\Omega}{\Omega}{\extendsctx{\hat{\Gamma}'}{\mu}}}{\addtele{\addtele{\Lambda}{\Psi}}{\Upsilon}}} \\
      &\phantom{{}\sigmeq} \qquad\quad \mathrlap{\subexprmtt{\subexprmtt{}{{\addtele{\key{\alpha}{\Lambda}{\Theta}{\addtele{\extendsctx{\hat{\Gamma}'}{\mu}}{\Omega}}}{\addtele{\Psi}{\Upsilon}}}}}{\key{1_{(\addtele{\Psi}{\Upsilon})}}{\addtele{\Psi}{\Upsilon}}{\addtele{\Psi}{\Upsilon}}{\addtele{\addtele{\extendsctx{\hat{\Gamma}'}{\mu}}{\Omega}}{\Theta}}}} && \explanation{\textasteriskcentered} \\
      &\sigmeq \subexprmtt{\subexprmtt{\vzero{}}{\key{\beta}{\,\locksym_\mu}{\addtele{\addtele{\addtele{\Omega}{\Lambda}}{\Psi}}{\Upsilon}}{\extendsctx{\hat{\Gamma}'}{\mu}}}}{\addtele{\key{\alpha}{\Lambda}{\Theta}{\addtele{\extendsctx{\hat{\Gamma}'}{\mu}}{\Omega}}}{\addtele{\Psi}{\Upsilon}}} && \explanation{\inlinerulename{wsmtt-eq-sub-key-unit}} \\
      &= \subexprmtt{\embed{\vzero{\beta}}}{\embed{\addtele{\key{\alpha}{\Lambda}{\Theta}{\addtele{\extendsctx{\hat{\Gamma}'}{\mu}}{\Omega}}}{\addtele{\Psi}{\Upsilon}}}}. && \explanation{Definition of $\embed{\_}$}
    \end{align*}
    In the step marked by (\textasteriskcentered) we use of the rules \inlinerulename{wsmtt-eq-sub-key-compose-vertical} and \inlinerulename{wsmtt-eq-sub-key-compose-horizontal} from \cref{fig:sigma-equiv}.
  \item \case{} $v = \suc{v'}$ with $\hat{\Gamma} = \addtele{\extendsctx{\hat{\Gamma}'}{\mu}}{\Omega}$ and $\sfvar{\addtele{\addtele{\addtele{\addtele{\hat{\Gamma}'}{\Omega}}{\Lambda}}{\Psi}}{\Upsilon}}{v'}{p}$ \\
    In this case we have that
    \begin{align*}
      &\embed{\arenexpr[\addtele{\Psi}{\Upsilon}]{\suc{v'}}{\key{\alpha}{\Lambda}{\Theta}{\addtele{\extendsctx{\hat{\Gamma}'}{\mu}}{\Omega}}}} \\
      &= \embed{\suc{\arenexpr[\addtele{\Psi}{\Upsilon}]{v'}{\key{\alpha}{\Lambda}{\Theta}{\addtele{\hat{\Gamma}'}{\Omega}}}}} && \explanation{\cref{eq:arenvar-key,eq:twocell-vsuc}} \\
      &= \subexprmtt{\embed{\arenexpr[\addtele{\Psi}{\Upsilon}]{v'}{\key{\alpha}{\Lambda}{\Theta}{\addtele{\hat{\Gamma}'}{\Omega}}}}}{\addtele{\addtele{\addtele{\addtele{\pi}{\Omega}}{\Theta}}{\Psi}}{\Upsilon}} && \explanation{Definition of $\embed{\_}$} \\
      &\sigmeq \mathrlap{\subexprmtt{\subexprmtt{\embed{v'}}{\embed{\addtele{\addtele{\key{\alpha}{\Lambda}{\Theta}{\addtele{\hat{\Gamma}'}{\Omega}}}{\Psi}}{\Upsilon}}}}{\addtele{\addtele{\addtele{\addtele{\pi}{\Omega}}{\Theta}}{\Psi}}{\Upsilon}}} \\
      &&& \explanation{Induction hypothesis} \\
      &\sigmeq \mathrlap{\subexprmtt{\subexprmtt{\embed{v'}}{\addtele{\addtele{\addtele{\addtele{\pi}{\Omega}}{\Lambda}}{\Psi}}{\Upsilon}}}{\embed{\addtele{\addtele{\key{\alpha}{\Lambda}{\Theta}{\addtele{\hat{\Gamma}'}{\Omega}}}{\Psi}}{\Upsilon}}}} \\
      &&& \explanation{\inlinerulename{wsmtt-eq-sub-key-natural}} \\
      &= \subexprmtt{\embed{\suc{v'}}}{\embed{\addtele{\addtele{\key{\alpha}{\Lambda}{\Theta}{\addtele{\hat{\Gamma}'}{\Omega}}}{\Psi}}{\Upsilon}}}. && \explanation{Definition of $\embed{\_}$} \qedhere
    \end{align*}
  \end{itemize}
\end{proof}

We can now prove that the condition in \cref{lem:embed-sub-term-lock-tele} is actually always satisfied.
\begin{lemma}
  \label{lem:embed-subst-var}
  Given an atomic \subfree{} substitution $\sfasub{\sigma}{\hat{\Gamma}}{\hat{\Delta}}{m}$, a lock telescope $\Lambda : \Locktele{n}{m}$ and a variable $\sfvar{\addtele{\hat{\Delta}}{\Lambda}}{v}{n}$, then we have that $\sigmeqexpr{\addtele{\hat{\Gamma}}{\Lambda}}{\embed{\asubexpr[\Lambda]{v}{\sigma}}}{\subexprmtt{\embed{v}}{\embed{\addtele{\sigma}{\Lambda}}}}{n}$.
\end{lemma}
\begin{proof}
  This proof proceeds by induction on the atomic substitution $\sigma$.
  \begin{itemize}
  \item \case{} $\sfasub{\emptysub}{\hat{\Gamma}}{\emptyctx{}}{m}$ (\inlinerulename{sf-arensub-empty}) \\
    In this case there can be no variable in the scoping context $\addtele{\emptyctx}{\Lambda}$, so the statement we have to prove is vacuously true.
  \item \case{} $\sfasub{\ida}{\hat{\Gamma}}{\hat{\Gamma}}{m}$ (\inlinerulename{sf-arensub-id}) \\
    Now $\embed{\asubexpr[\Lambda]{v}{\ida{}}} = \embed{v}$ and on the other hand
    \begin{align*}
      \subexprmtt{\embed{v}}{\embed{\addtele{\ida{}}{\Lambda}}}
      &= \subexprmtt{\embed{v}}{\addtele{\id{}}{\Lambda}} && \explanation{Definition of $\embed{\_}$} \\
      &\sigmeq \subexprmtt{\embed{v}}{\id} && \explanation{\inlinerulename{wsmtt-eq-sub-lock-id}} \\
      &\sigmeq \embed{v}. && \explanation{\inlinerulename{wsmtt-eq-expr-sub-id}}
    \end{align*}
  \item \case{} $\sfasub{\weaken{\sigma}}{\extendsctx{\hat{\Gamma}}{\mu}}{\hat{\Delta}}{m}$ (\inlinerulename{sf-arensub-weaken}) \\
    In this case we can compute
    \begin{align*}
      &\embed{\asubexpr[\Lambda]{v}{\weaken{\sigma}}} \\
      &= \embed{\arenexpr{\asubexpr[\Lambda]{v}{\sigma}}{\addtele{\pi}{\Lambda}}} && \explanation{\cref{eq:asubvar-weaken}} \\
      &\sigmeq \subexprmtt{\embed{\asubexpr[\Lambda]{v}{\sigma}}}{\embed{\addtele{\pi}{\Lambda}}} && \explanation{\cref{lem:embed-pi-ren}} \\
      &\sigmeq \subexprmtt{\subexprmtt{\embed{v}}{\embed{\addtele{\sigma}{\Lambda}}}}{\embed{\addtele{\pi}{\Lambda}}} && \explanation{Induction hypothesis} \\
      &\sigmeq \subexprmtt{\embed{v}}{\addtele{\big(\embed{\sigma} \circ \pi\big)}{\Lambda}} && \explanation{\textasteriskcentered}\\
      &= \subexprmtt{\embed{v}}{\embed{\addtele{\weaken{\sigma}}{\Lambda}}}. && \explanation{Definition of $\embed{\_}$}
    \end{align*}
    In the step marked with (\textasteriskcentered) we made use of \inlinerulename{wsmtt-eq-expr-sub-compose} and \inlinerulename{wsmtt-eq-sub-lock-compose}.
  \item \case{} $\sfasub{\lock{\sigma}{\mu}}{\lock{\hat{\Gamma}}{\mu}}{\lock{\hat{\Delta}}{\mu}}{m}$ (\inlinerulename{sf-arensub-lock}) \\
    Then we have that
    \begin{align*}
      \embed{\asubexpr[\Lambda]{v}{\lock{\sigma}{\mu}}}
      &= \embed{\asubexpr[\addtele{\locksym_\mu}{\Lambda}]{v}{\sigma}} && \explanation{\cref{eq:asubvar-lock}} \\
      &= \subexprmtt{\embed{v}}{\embed{\addtele{\lock{\sigma}{\mu}}{\Lambda}}}. && \explanation{Induction hypothesis}
    \end{align*}
  \item \case{} $\sfasub{\key{\alpha}{\Theta}{\Psi}{\hat{\Gamma}}}{\extendsctx{\hat{\Gamma}}{\Psi}}{\extendsctx{\hat{\Gamma}}{\Theta}}{n}$ (\inlinerulename{sf-arensub-key}) \\
    In this case the result is a direct consequence of \cref{lem:embed-key-ren} because $\asubexpr[\Lambda]{v}{\key{\alpha}{\Theta}{\Psi}{\hat{\Gamma}}} = \arenexpr[\Lambda]{v}{\key{\alpha}{\Theta}{\Psi}{\hat{\Gamma}}}$.
  \item \case{} $\sfasub{\sigma . t}{\hat{\Gamma}}{\extendsctx{\hat{\Delta}}{\mu}}{n}$ (\inlinerulename{sf-asub-extend}) \\
    Now we distinguish between two cases for the variable $v$.
    \begin{itemize}
    \item \case{} $v = \vzero{\alpha}$ \\
      On the one hand, by \cref{eq:asubvar-extend-vzero} we have that
      \[
        \embed{\asubexpr[\Lambda]{\vzero{\alpha}}{\sigma . t}} = \embed{\arenexpr{t}{\key{\alpha}{\locksym_\mu}{\Lambda}{\hat{\Gamma}}}}.
      \]
      On the other hand, we can compute
      \begin{align*}
        &\subexprmtt{\embed{\vzero{\alpha}}}{\embed{\addtele{(\sigma . t)}{\Lambda}}} \\
        &= \mathrlap{\subexprmtt{\subexprmtt{\vzero{}}{\key{\alpha}{\,\locksym_\mu}{\Lambda}{\hat{\Delta}}}}{\addtele{\big(\embed{\sigma} . \embed{t}\big)}{\Lambda}}} \\
        &&& \explanation{Definition of $\embed{\_}$} \\
        &\sigmeq \mathrlap{\subexprmtt{\subexprmtt{\vzero}{\lock{\big(\embed{\sigma} . \embed{t}\big)}{\mu}}}{\key{\alpha}{\,\locksym_\mu}{\Lambda}{\hat{\Gamma}}}} \\
        &&& \explanation{\inlinerulename{wsmtt-eq-sub-key-natural}} \\
        &\sigmeq \subexprmtt{\embed{t}}{\key{\alpha}{\,\locksym_\mu}{\Lambda}{\hat{\Gamma}}}. && \explanation{\inlinerulename{wsmtt-eq-expr-extend-var}}
      \end{align*}
      Combining these two computations, the result follows from \cref{lem:embed-key-ren}.
    \item \case{} $v = \suc{v'}$ \\
      The left-hand side now reduces to
      \begin{align*}
        &\embed{\asubexpr[\Lambda]{\suc{v'}}{\sigma . t}} \\
        &= \embed{\asubexpr[\Lambda]{v'}{\sigma}} && \explanation{\cref{eq:asubvar-extend-vsuc}} \\
        &\sigmeq \subexprmtt{\embed{v'}}{\embed{\addtele{\sigma}{\Lambda}}}. && \explanation{Induction hypothesis}
      \end{align*}
      For the right-hand side, we have
      \begin{align*}
        &\subexprmtt{\embed{\suc{v'}}}{\embed{\addtele{(\sigma . t)}{\Lambda}}} \\
        &= \subexprmtt{\subexprmtt{\embed{v'}}{\addtele{\pi}{\Lambda}}}{\addtele{\big(\embed{\sigma} . \embed{t}\big)}{\Lambda}} && \explanation{Definition of $\embed{\_}$} \\
        &\sigmeq \subexprmtt{\embed{v'}}{\addtele{\big(\pi \circ \big(\embed{\sigma} . \embed{t}\big)\big)}{\Lambda}} \\
        &\sigmeq \subexprmtt{\embed{v'}}{\embed{\addtele{\sigma}{\Lambda}}}.
      \end{align*}
      In the last two steps we made use of \inlinerulename{wsmtt-eq-expr-sub-compose}, \inlinerulename{wsmtt-eq-sub-lock-compose} and \inlinerulename{wsmtt-eq-sub-extend-weaken}. \qedhere
    \end{itemize}
  \end{itemize}
\end{proof}

\begin{proposition}
  \label{prop:embed-subst-term}
  Given an \subfree{} expression $\sfexpr{\hat{\Delta}}{t}{m}$ and a substitution $\sfsub{\sigma}{\hat{\Gamma}}{\hat{\Delta}}{m}$, we have that
  $\sigmeqexpr{\hat{\Gamma}}{\embed{\subexpr{t}{\sigma}}}{\subexprmtt{\embed{t}}{\embed{\sigma}}}{m}$.
\end{proposition}
\begin{proof}
  Because of the rules \inlinerulename{wsmtt-eq-expr-sub-id} and \inlinerulename{wsmtt-eq-expr-sub-compose}, it suffices to prove this result for an atomic substitution $\sigma$.
  This follows directly by combining \cref{lem:embed-sub-term-lock-tele,lem:embed-subst-var}.
\end{proof}

\subsection{Proof of \cref{thm:soundness}}

Just like the completeness theorem, we will prove a more general statement than \cref{thm:soundness} that also takes substitution into account.
\begin{theorem}[Soundness]
  For every WSMTT expression $\mttexpr{\hat{\Gamma}}{t}{m}$ we have $\sigmeqexpr{\hat{\Gamma}}{\embed{\translate{t}}}{t}{m}$ and for every WSMTT substitution $\mttsub{\sigma}{\hat{\Gamma}}{\hat{\Delta}}{m}$ we have $\sigmeqsub{\embed{\translate{\sigma}}}{\sigma}{\hat{\Gamma}}{\hat{\Delta}}{m}$.
\end{theorem}
\begin{proof}
  The proof proceeds by induction on the expression $t$ and the substitution $\sigma$.
  All cases for the expression constructors that are shared between SFMTT and WSMTT are trivial from the induction hypotheses, but we show two of them (\inlinerulename{wsmtt-expr-arrow} and \inlinerulename{wsmtt-expr-lam}) as illustration.
  In particular, all constructors from \cref{fig:raw-mtt-expr-sub} are covered below.
  \begin{itemize}
  \item \case{} $\mttexpr{\hat{\Gamma}}{\modfunc{\mu}{T}{S}}{n}$ (\inlinerulename{wsmtt-expr-arrow}) \\
    By definition of $\translate{\_}$ and $\embed{\_}$ we have that
    \[
      \embed{\translate{\modfunc{\mu}{T}{S}}} = \modfunc{\mu}{\embed{\translate{T}}}{\embed{\translate{S}}}.
    \]
    Hence the result follows from the induction hypothesis applied to the subexpressions $T$ and $S$.
  \item \case{} $\mttexpr{\hat{\Gamma}}{\lam{\mu}{t}}{n}$ (\inlinerulename{wsmtt-expr-lam}) \\
    Again, by expanding the definitions of $\translate{\_}$ and $\embed{\_}$, we get $\embed{\translate{\lam{\mu}{t}}} = \lam{\mu}{\embed{\translate{t}}}$,
    so that the result follows from the induction hypothesis applied to the subexpression $t$.
  \item \case{} $\mttexpr{\lock{\extendsctx{\hat{\Gamma}}{\mu}}{\mu}}{\vzero{}}{m}$ (\inlinerulename{wsmtt-expr-var}) \\
    Now we have that
    \[
      \embed{\translate{\vzero{}}} = \embed{\vzero{1_\mu}} = \subexprmtt{\vzero{}}{\key{1_\mu}{\locksym_\mu}{\locksym_\mu}{\extendsctx{\hat{\Gamma}}{\mu}}}.
    \]
    This last expression is indeed $\upsigma$-equivalent to $\vzero{}$ because of \inlinerulename{wsmtt-eq-sub-key-unit} and \inlinerulename{wsmtt-eq-expr-sub-id}.
  \item \case{} $\mttexpr{\hat{\Gamma}}{\subexprmtt{t}{\sigma}}{m}$ (\inlinerulename{wsmtt-expr-sub}) \\
    In this case we have
    \begin{align*}
      \embed{\translate{\subexprmtt{t}{\sigma}}}
      &= \embed{\subexpr{\translate{t}}{\translate{\sigma}}} && \explanation{Definition of $\translate{\_}$} \\
      &\sigmeq \subexprmtt{\embed{\translate{t}}}{\embed{\translate{\sigma}}} && \explanation{\cref{prop:embed-subst-term}} \\
      &\sigmeq \subexprmtt{t}{\embed{\translate{\sigma}}} && \explanation{Induction hypothesis for $t$} \\
      &\sigmeq \subexprmtt{t}{\sigma}. && \explanation{Induction hypothesis for $\sigma$}
    \end{align*}
  \item \case{} $\mttsub{\emptysub}{\hat{\Gamma}}{\emptyctx{}}{m}$ (\inlinerulename{wsmtt-sub-empty}) \\
    Since $\embed{\translate{!}}$ is a WSMTT substitution from $\hat{\Gamma}$ to the empty scoping context $\emptyctx$, the result follows immediately from \inlinerulename{wsmtt-eq-sub-empty-unique}.
  \item \case{} $\mttsub{\id}{\hat{\Gamma}}{\hat{\Gamma}}{m}$ (\inlinerulename{wsmtt-sub-id}) \\
    By the definition of translation and embedding, we immediately have $\embed{\translate{\id}} = \id$.
  \item \case{} $\mttsub{\pi}{\extendsctx{\hat{\Gamma}}{\mu}}{\hat{\Gamma}}{n}$ (\inlinerulename{wsmtt-sub-weaken}) \\
    Now we have that
    \begin{align*}
      \embed{\translate{\pi}}
      &= \embed{\rensubcons{\id}{\weaken{\ida}}} && \explanation{Definition of $\translate{\_}$ and \cref{eq:weakening-arensub}} \\
      &= \id \circ (\id \circ \pi). && \explanation{Definition of $\embed{\_}$}
    \end{align*}
    This last substitution is indeed $\upsigma$-equivalent to $\pi$ by \inlinerulenamestyle{wsmtt-eq-sub-id-left}.
  \item \case{} $\mttsub{\sigma \circ \tau}{\hat{\Gamma}}{\hat{\Xi}}{m}$ (\inlinerulename{wsmtt-sub-compose}) \\
    Now we compute that $\embed{\translate{\sigma \circ \tau}} = \embed{\concat{\translate{\sigma}}{\translate{\tau}}}$.
    Since the embedding of a sequence of atomic \subfree{} substitutions is the composition of the embedding of these atomic substitutions and since WSMTT substitution composition is associative up to $\upsigma$-equivalence, we have that $\embed{\concat{\translate{\sigma}}{\translate{\tau}}} \sigmeq \embed{\translate{\sigma}} \circ \embed{\translate{\tau}}$.
    From this the result follows via the induction hypothesis applied to $\sigma$ and $\tau$.
  \item \case{} $\mttsub{\lock{\sigma}{\mu}}{\lock{\hat{\Gamma}}{\mu}}{\lock{\hat{\Delta}}{\mu}}{m}$ (\inlinerulename{wsmtt-sub-lock}) \\
    In this case we get that $\embed{\translate{\lock{\sigma}{\mu}}} = \embed{\lock{\translate{\sigma}}{\mu}} \sigmeq \lock{\embed{\translate{\sigma}}}{\mu}$,
    where the last equivalence follows from \inlinerulename{wsmtt-eq-sub-lock-id} and \inlinerulename{wsmtt-eq-sub-lock-compose}.
    The desired result is then a consequence of the induction hypothesis applied to $\sigma$.
  \item \case{} $\mttsub{\key{\alpha}{\Theta}{\Psi}{\hat{\Gamma}}}{\addtele{\hat{\Gamma}}{\Psi}}{\addtele{\hat{\Gamma}}{\Theta}}{n}$ (\inlinerulename{wsmtt-sub-key}) \\
    We can now compute that
    \begin{align*}
      \embed{\translate{\key{\alpha}{\Theta}{\Psi}{\hat{\Gamma}}}}
      &= \embed{\rensubcons{\id}{\key{\alpha}{\Theta}{\Psi}{\hat{\Gamma}}}} && \explanation{Definition of $\translate{\_}$} \\
      &= \id \circ \key{\alpha}{\Theta}{\Psi}{\hat{\Gamma}}, && \explanation{Definition of $\embed{\_}$}
    \end{align*}
    which is indeed $\upsigma$-equivalent to $\key{\alpha}{\Theta}{\Psi}{\hat{\Gamma}}$ because of \inlinerulenamestyle{wsmtt-eq-sub-id-left}
  \item \case{} $\mttsub{\sigma . t}{\hat{\Gamma}}{\extendsctx{\hat{\Delta}}{\mu}}{n}$ (\inlinerulename{wsmtt-sub-extend}) \\
    Expanding the definitions of $\translate{\_}$ and $\embed{\_}$, we have that
    \[
      \embed{\translate{\sigma . t}} = \embed{\rensubcons{\lift{\translate{\sigma}}}{\left(\ida{} . \translate{t}\right)}} = \embed{\lift{\translate{\sigma}}} \circ \left(\id . \embed{\translate{t}}\right).
    \]
    By \cref{lem:embed-lift} we know that $\embed{\lift{\translate{\sigma}}} \sigmeq \lift{\embed{\translate{\sigma}}}$ and combining this with the induction hypothesis for $\sigma$ and $t$, we get that
    \[
      \embed{\translate{\sigma . t}} \sigmeq \lift{\sigma} \circ (\id . t).
    \]
    This last substitution can be proven $\upsigma$-equivalent to $\sigma . t$ by the rules \inlinerulename{wsmtt-eq-sub-extend-eta}, \inlinerulename{wsmtt-eq-sub-extend-weaken} and \inlinerulename{wsmtt-eq-expr-extend-var}. \qedhere
  \end{itemize}
\end{proof}

\bibliography{bibliography}

\end{document}

\typeout{get arXiv to do 4 passes: Label(s) may have changed. Rerun}